\documentclass{theoretics}

\title{Minimum Star Partitions of Simple Polygons in Polynomial Time}

\ThCSauthor[Copenhagen]{Mikkel Abrahamsen}{miab@di.ku.dk}[0000-0003-2734-4690]
\ThCSauthor[CWI]{Joakim Blikstad}{joakim.blikstad@cwi.nl}[0009-0004-0874-2356]
\ThCSauthor[Inria]{Andr\'{e} Nusser}{andre.nusser@cnrs.fr}[0000-0002-6349-869X]
\ThCSauthor[Copenhagen]{Hanwen Zhang}{hazh@di.ku.dk}[0000-0002-3149-7799]
\ThCSaffil[Copenhagen]{University of Copenhagen, Denmark}
\ThCSaffil[CWI]{Centrum Wiskunde \& Informatica, Amsterdam, Netherlands}
\ThCSaffil[Inria]{Universit\'{e} C\^{o}te d'Azur, CNRS, Inria, Sophia Antipolis, France}
\ThCSthanks{A preliminary version of this article appeared at STOC 2024 \cite{abrahamsen2024minimum}.\\
  Mikkel Abrahamsen and Hanwen Zhang are supported by Independent Research Fund Denmark, grant 1054-00032B, and by the Carlsberg Foundation, grant CF24-1929.
}
\ThCSshortnames{M.\ Abrahamsen, J.\ Blikstad, A.\ Nusser, H.\ Zhang}  
\ThCSshorttitle{Minimum Star Partitions of Simple Polygons in Polynomial Time}
\ThCSyear{2026}
\ThCSarticlenum{2}
\ThCSreceived{Mar 6, 2025}
\ThCSrevised{Oct 6, 2025}
\ThCSaccepted{Nov 3, 2025}
\ThCSpublished{Jan 13, 2026}
\ThCSdoicreatedtrue
\ThCSkeywords{Computational geometry, star partition, dynamic programming}

\usepackage[utf8]{inputenc}
\usepackage[english]{babel}
\usepackage{csquotes}
\usepackage{graphics}
\usepackage{amsfonts}
\usepackage{mathtools}
\usepackage{algpseudocode}

\usepackage[capitalise,noabbrev,nameinlink]{cleveref}

\crefname{observation}{Observation}{Observations}
\crefname{claim}{Claim}{Claims}
\crefname{cfigure}{Figure}{Figures}

\newcommand{\RR}{\mathbb{R}}
\newcommand{\NN}{\mathbb{N}}

\newcommand\Scen{S^{\textsc{(centers)}}}
\newcommand\Sint{S^{\textsc{(internal)}}}
\newcommand\Sbor{S^{\textsc{(border)}}}

\newcommand\poly{\mathrm{poly}}

\SetKwComment{tcp}{$\triangleright$\ }{}

\makeatletter
\newcommand\IfRestateTF{%
  \ifx\label\thmt@gobble@label 
    \expandafter\@firstoftwo
  \else
    \expandafter\@secondoftwo
  \fi
}
\makeatother
\newcommand{\RestateRemark}{\IfRestateTF{{\normalfont\bfseries (Restated) }}{}}

\begin{document}

\maketitle

\begin{abstract}
We devise a polynomial-time algorithm for partitioning a simple polygon $P$
into a minimum number of star-shaped polygons.
The question of whether such an algorithm exists has been open for more than four decades [Avis and Toussaint, Pattern Recognit., 1981] and it has been repeated frequently, for example in O'Rourke's famous book [\emph{Art Gallery Theorems and Algorithms}, 1987].
In addition to its strong theoretical motivation, the problem is also motivated by practical domains such as CNC pocket milling, motion planning, and shape parameterization.

The only previously known algorithm for a non-trivial special case is for $P$ being both monotone and rectilinear [Liu and Ntafos, Algorithmica, 1991].
For general polygons, an algorithm was only known for the restricted version in which Steiner points are disallowed [Keil, SIAM J. Comput., 1985], meaning that each corner of a piece in the partition must also be a corner of $P$. Interestingly, the solution size for the restricted version may be linear for instances where the unrestricted solution has constant size.
The covering variant in which the pieces are star-shaped but allowed to overlap---known as the \emph{Art Gallery Problem}---was recently shown to be $\exists\mathbb R$-complete and is thus likely not in NP [Abrahamsen, Adamaszek and Miltzow, STOC 2018 \&\ J. ACM 2022]; this is in stark contrast to our result.
Arguably the most related work to ours is the polynomial-time algorithm to partition a simple polygon into a minimum number of \emph{convex} pieces by
Chazelle and Dobkin~[STOC, 1979 \&\ Comp. Geom., 1985].
\end{abstract}
\begin{figure}
\centering
\includegraphics[page=10,width=0.7\linewidth]{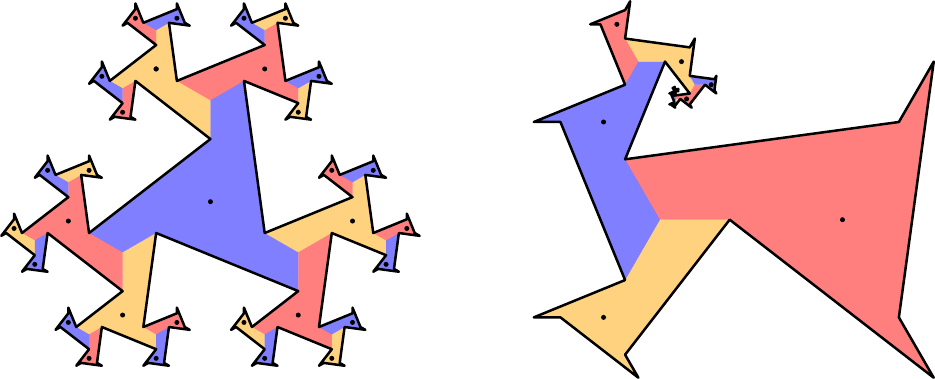}
\caption{Repeating the patterns, we obtain polygons where star centers and corners of pieces of arbitrarily high degree are required.
}
\label{fig:art}
\end{figure}






\section{Introduction}\label{sec:intro}

A simple polygon $Q$ is called \emph{star-shaped} if there is a point $A$ in $Q$ such that for all points $B$ in $Q$, the line segment $AB$ is contained in $Q$.
Such a point $A$ is called a \emph{star center} of $Q$.
A \emph{star partition} of a polygon $P$ is a set of pairwise non-overlapping star-shaped simple polygons whose union equals $P$; see \Cref{fig:art}.
The polygons constituting the star partition are called the \emph{pieces} of the partition.

Avis and Toussaint~\cite{DBLP:journals/pr/AvisT81} described in 1981, an algorithm running in $O(n\log n)$ time to partition a simple polygon (i.e., a polygon without holes) into at most $\lfloor n/3\rfloor$ star-shaped pieces---where $n$ denotes the number of corners of the polygon---based on Fisk's constructive proof~\cite{DBLP:journals/jct/Fisk78a} of Chv\'{a}tal's Art Gallery Theorem~\cite{chvatal1975combinatorial}.
Avis and Toussaint~\cite{DBLP:journals/pr/AvisT81} wrote: ``\emph{An interesting open problem would be to try to find the decomposition into the minimum number of star-shaped polygons.}''
This question has been repeated in several other papers~\cite{toussaint1982computational,DBLP:journals/tit/ORourkeS83,DBLP:journals/siamcomp/Keil85,DBLP:journals/pieee/Shermer92} and also in O'Rourke's well-known book~\cite{o1987art}: ``\emph{Can a variant of Keil's dynamic programming approach~\cite{DBLP:journals/siamcomp/Keil85} be used to find star partitions permitting Steiner points\footnote{\label{footnote:steiner-points}A \emph{Steiner point} is a corner of a piece in the partition which is \emph{not} a corner of the input polygon. We discuss the challenges and importance of allowing Steiner points later in this section.}? Chazelle was able to achieve $O(n^3)$ for minimum convex partition with Steiner points via a very complex dynamic programming algorithm~\cite{chazelle1980computational}, but star partitions seem even more complicated.}''
Before our work, the problem was not known to be in \textsf{NP} and not even an exponential-time algorithm was known.
In this paper, we resolve the open problem by providing a polynomial-time algorithm,
thereby closing a research question that has been open for more than four decades.

\begin{restatable}{theorem}{MainTheorem}
\label{thm:main-theorem}\RestateRemark
There is an algorithm performing $O(n^{105})$ arithmetic operations that partitions a simple polygon with $n$ corners into a minimum number of star-shaped pieces.
The number of bits used to represent each Steiner point in the constructed solution is $O(K)$ where $K$ is the total number of bits used to represent the corners of $P$.
\end{restatable}

\paragraph{Related work.}
The minimum star partition problem belongs to the class of \emph{decomposition problems}, which forms an old and large sub-field in computational geometry.
In all of these problems, we want to \emph{decompose} a polygon $P$ into polygonal \emph{pieces} which are in some sense simpler than the original polygon $P$.
Here, the union of the pieces should be $P$, and we usually seek a decomposition into as few pieces as possible.
A decomposition where the pieces may overlap is called a \emph{cover}, and a decomposition where the pieces are pairwise interior-disjoint is called a \emph{partition}.
This leads to a wealth of interesting problems,
depending on the assumptions about the input polygon $P$ and the requirements on the pieces.
There is a vast literature about such decomposition problems, as documented in several highly-cited books and survey papers that give an overview of the state-of-the-art at the time of publication~\cite{keil1985minimum,chazelle1985approximation,o1987art,DBLP:journals/pieee/Shermer92,chazelle1994decomposition,keil1999polygon,o2004polygons}.
Some of the most common variations are
\begin{itemize}[noitemsep,topsep=0pt,parsep=0pt,partopsep=0pt]
\item whether the input polygon $P$ is simple or may have holes, 
\item whether we seek a cover or a partition,
\item whether we allow Steiner points${}^\text{\ref{footnote:steiner-points}}$ or not,
\item what shape of pieces we allow; let us mention that for \emph{partitioning} a simple polygon, variants have been studied with polygonal pieces that are convex~\cite{DBLP:journals/tc/FengP75,DBLP:journals/tc/Schachter78,chazelle1985optimal,DBLP:journals/siamcomp/Keil85,DBLP:conf/focs/Chazelle82,DBLP:journals/iandc/HertelM85,DBLP:journals/ijcga/KeilS02}, star-shaped~\cite{DBLP:journals/tc/FengP75,DBLP:journals/siamcomp/Keil85,DBLP:journals/algorithmica/LiuN91}, monotone~\cite{DBLP:journals/ipl/GareyJPT78,DBLP:journals/ipl/LiuN88}, spiral-shaped~\cite{DBLP:journals/siamcomp/Keil85}, ``fat''~\cite{DBLP:journals/comgeo/Kreveld98,DBLP:journals/comgeo/Damian04,DBLP:conf/cccg/BuchinS21}, ``small''~\cite{DBLP:conf/esa/ArkinD0GMPT20,damian2004computing}, ``circular''~\cite{DBLP:conf/cccg/Damian-IordacheO03}, triangles~\cite{DBLP:journals/jal/AsanoAP86,DBLP:journals/dcg/Chazelle91}, quadrilaterals~\cite{DBLP:conf/compgeom/Lubiw85,DBLP:conf/compgeom/Muller-HannemannW97} and trapezoids~\cite{DBLP:journals/jacm/AsanoAI86}.
\end{itemize}

Closely related to our problem is that of \emph{covering} a polygon with a minimum number of star-shaped pieces.
This is usually known as the \emph{Art Gallery Problem} and described equivalently as the task of placing guards (star centers) so that each point in the polygon can be seen by at least one guard.
Interestingly, the Art Gallery Problem has been shown to be $\exists \mathbb R$-complete~\cite{DBLP:journals/jacm/AbrahamsenAM22} and it is thus not likely to be in \textsf{NP}. This is in stark contrast to our main result, which shows that the corresponding \emph{partitioning} problem is in \textsf{P}.
Covering a polygon with a minimum number of convex pieces is likewise $\exists \mathbb R$-complete~\cite{DBLP:conf/focs/Abrahamsen21}.

If the polygon $P$ can have holes, the minimum star partition problem is known to be \textsf{NP}-hard, whether or not Steiner points are allowed~\cite{o1987art}; again in contrast to our result.

Keil~\cite{DBLP:journals/siamcomp/Keil85} gave polynomial-time algorithms for partitioning simple polygons into various types of pieces where Steiner points are not allowed.
Among these algorithms is an $O(n^7\log n)$ time algorithm for finding a minimum star partition of a simple polygon without Steiner points, but the unrestricted version of the problem (with Steiner points allowed) remained open.
Let us mention that there are polygons where $\Theta(n)$ pieces are needed when Steiner points are not allowed, whereas $2$ pieces are sufficient when they are allowed; see \Cref{fig:badwithoutSteiner} (left).
Therefore, our algorithm in general constructs partitions that are significantly smaller.
This highlights an interesting difference between minimum star partitions and convex partitions:
A minimum convex partition without Steiner points has at most $4$ times as many pieces as when Steiner points are allowed~\cite{DBLP:journals/iandc/HertelM85}.
Another difference is that an arbitrarily small perturbation of a single corner can change the size of the minimum star partition between $1$ and $\Theta(n)$, whereas the change in size of the minimum convex partition is at most $1$; see \Cref{fig:badwithoutSteiner} (right).
In that sense, minimum star partitions are much more sensitive to the input.

\begin{figure}
\centering
\includegraphics[page=21]{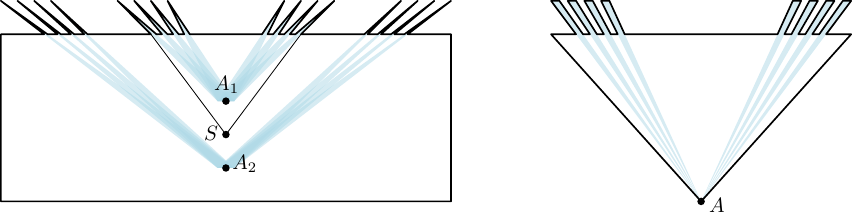}
\caption{
\emph{Left:}
A polygon that is partitioned into two star-shaped pieces using the Steiner point $S$.
A star center that can see the two middle groups of spikes must be placed at or close to $A_1$, while a star center that sees the outer groups must be at or close to $A_2$.
Without Steiner points we need $\Theta(n)$ pieces to partition at least one of the groups of spikes.
\emph{Right:}
Moving the bottom corner a bit up changes the size of a minimum star partition from $1$ to $\Theta(n)$.
}
\label{fig:badwithoutSteiner}
\end{figure}

Unrestricted partitioning problems (that is, allowing Steiner points), are seemingly much more challenging to design algorithms for.
Chazelle and Dobkin~\cite{chazelle1980computational,chazelle1985optimal} proved already in 1979 that a simple polygon can be partitioned into a minimum number of \emph{convex} pieces in $O(n^3)$ time, by designing a rather complicated dynamic program.
Asano, Asano and Imai~\cite{DBLP:journals/jacm/AsanoAI86} gave an $O(n^2)$-time algorithm for partitioning a simple polygon into a minimum number of trapezoids, each with a pair of horizontal edges.
However, the minimum partitioning problem has remained open for most other shapes of pieces (e.g.\ triangles, spiral-shaped, and---until now---star-shaped).

Liu and Ntafos~\cite{DBLP:journals/algorithmica/LiuN91} also studied the minimum star partition problem, but with restrictions on the input polygon. They describe an algorithm for partitioning simple monotone and rectilinear\footnote{A polygon $P$ is \emph{monotone} if there is a line $\ell$ such that the intersection of $P$ with any line orthogonal to $\ell$ is connected, and $P$ is \emph{rectilinear} if all sides are either vertical or horizontal.} polygons into a minimum number of star-shaped polygons in $O(n)$ time, and a $6$-approximation algorithm for simple rectilinear polygons that are not necessarily monotone.

\paragraph{Challenges.}
As argued above, star partitions are very sensitive to the input polygon, and allowing Steiner points is in general necessary to obtain a partition with few pieces (\Cref{fig:badwithoutSteiner}). In order to demonstrate the complicated nature of optimal star partitions, let us also consider \Cref{fig:art}, which shows (representatives of) two families of polygons with arbitrarily many corners and unique optimal star partitions.
In both examples, some star centers and Steiner points depend on as many as $\Theta(n)$ corners of $P$.
The example to the right shows that star centers and Steiner points of \emph{degree} $\Theta(n)$ are also needed, where points $V_i$ of degree $i$ are defined as follows.
The points $V_0$ are the corners of $P$; and $V_{i+1}$ are the intersection points between two non-parallel lines, each through a pair of points in $V_i$.
The size of $V_i$ grows as $\Theta(n^{4^i})$, so we cannot iterate through the possible star centers and Steiner points.
This is in contrast to the problem without Steiner points studied by Keil~\cite{DBLP:journals/siamcomp/Keil85}.
Here, by definition, the corners of the pieces are in $V_0$ and it is not hard to see that the star centers can be chosen from $V_1$, of which there are ``only'' $O(n^4)$.

Since we cannot iterate through all possible star centers and Steiner points, we devise a two-phase algorithm, as follows.
In the first phase, we find polynomially many relevant points, so that we are sure that an optimal solution can be constructed using a subset of those points as star centers and Steiner points.
In the second phase, we use dynamic programming to find optimal solutions to larger and larger subpolygons, using only the constructed points from the first phase.
We note however that the phases are intertwined as the algorithm for the first phase calls the complete partitioning algorithm recursively on subpolygons.
The argument that the set of points constructed in the first phase is sufficient relies on several structural results about optimal star partitions which we believe are interesting in their own right.

\subsection{Practical Motivation}\label{sec:practical}

Besides being interesting from a theoretical angle, star partitions are useful in various practical domains; below we mention a few examples.
Many of the papers mentioned below describe algorithms for computing star partitions with no guarantee of finding an optimal one.

\paragraph{CNC pocket milling.}
Our first motivation comes from the generation of toolpaths for milling machines.
CNC milling is the computer-aided process of cutting some specified shape into a piece of material---such as steel, wood, ceramics, or plastic---using a milling machine.
When milling a pocket, spirals are a popular choice of toolpath, since the entire pocket can be machined without retracting the tool and sharp corners on the path can be largely avoided.
Some of the proposed methods to generate spirals require the shape of the pocket to be star-shaped, for instance because they rely on radial interpolation between curves that morph a single point (a star center) to the boundary of the pocket~\cite{bieterman2003curvilinear,patel2019effect,banerjee2012process}.
When milling a non-star-shaped pocket, we therefore seek to first partition the pocket into star-shaped regions, each of which can then be milled by their own spiral.
We want a star partition rather than a star cover, since it is a waste of time to cover the same area more than once.
In order to minimize the number of retractions (lifting and moving the tool from one spiral to the next), we want a partition into a minimum number of star-shaped regions.

\paragraph{Motion planning.}
Star partitions are also useful in the domain of motion planning.
Varadhan and Manocha~\cite{varadhan2005star} describe such an approach. They first partition the free space into star-shaped regions to subsequently construct a route for an agent from one point to another in the free space using the stars.
In each star, we route from the point of entrance to the star center and from there to a common boundary point with the next star.
Similar applications of star partitions are described in~\cite{dobson2014sparse, SPBG:SPBG07:073-080, lien2009planning, zhang2006fast}.

\paragraph{Capturing the shape of a polygon.}
Star partitions can be used to blend/morph one polygon into another~\cite{shapira1995shape,etzion1997compatible}, for shape matching and retrieval~\cite{yu2011computing}, and they are also used in shape parameterization~\cite{xia2010parameterization}.

\subsection{Technical Overview}
\label{sec:technical-overview}

To enable our algorithm, we had to identify a multitude of interesting structural properties of optimal star partitions, which are interesting in their own right.
In this section, we outline the most important of these properties and explain informally how they are used to derive a polynomial-time algorithm.
Naturally, we sometimes stay vague or glance over complicated details in order to hide complexity to make the technical overview easily accessible.

\begin{figure}
\centering
\includegraphics[page=22]{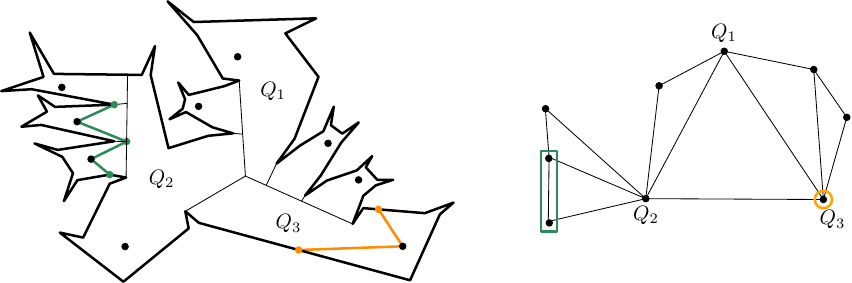}
\caption{
\emph{Left:}
A polygon with a star partition and an example of a short (orange) and a long (green) separator.
\emph{Right:}
The dual graph of the partition.
The short and long separators of the partition correspond to vertices, respectively edges, in the graph.
}
\label{fig:partitionDual}
\end{figure}

\paragraph{Separators.}
Similar to the algorithms for related partitioning problems~\cite{chazelle1985optimal,DBLP:journals/siamcomp/Keil85}, we use dynamic programming:
We compute optimal star partitions of larger and larger subpolygons $P'$ contained in the input polygon $P$.
For dynamic programming to work, we need an appropriate type of separator which separates the subpolygon $P'$ from the rest of $P$.
To this end, a useful (and non-trivial) property is that there exists an optimal partition in which each piece shares boundary with~$P$; as we will see in \Cref{sec:structure} (\Cref{cor:outer-planarity}).
This suggests that we use separators consisting of two or four segments of the following forms:

\begin{itemize}
\item \emph{Short separator:} $B_1$-$A_1$-$B_2$.
A piece with star center $A_1$ that shares boundary points $B_1$ and $B_2$ with~$P$.

\item \emph{Long separator:} $B_1$-$A_1$-$S$-$A_2$-$B_2$.
Each $A_i$ is the star center of a piece that shares the boundary point $B_i$ with $P$ and the point $S$ is a common point of the boundaries of the two pieces.
\end{itemize}

A state of our dynamic program consists of a separator and is used to calculate how many pieces we need to partition the associated subpolygon, which is the part of $P$ on one side of the separator.
We start with trivial short separators of two types: (i) degenerate ones of the form $B$-$A$-$B$ for a star center $A$ that can see a boundary point $B$, and (ii) $B_1$-$A$-$B_2$ where $B_1$ and $B_2$ are points on the same edge of $P$ so that the separator encloses a triangle.
We describe a few elementary operations to create partitions of larger subpolygons from smaller ones by merging two compatible separators into one that covers the union of the two subpolygons.

The main difficulty lies in choosing polynomially many candidates for the star centers $A_i$, the boundary points $B_i$ and the common points $S$, so that we can be sure that our algorithm eventually constructs an optimal partition.
As already mentioned, our algorithm has two phases, and in the first phase we compute a set of $O(n^6)$ points that are guaranteed to contain the star centers of an optimal partition. 
In Section~\ref{sec:area-max}, we show that we can use these potential star centers to also specify polynomially many candidates for the points $B_i$ and $S$.
In a second phase, the algorithm uses the constructed points to iterate through all relevant separators.

\begin{figure}
\centering
\includegraphics[page=25]{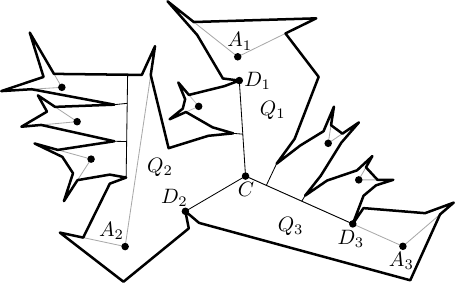}
\caption{
The same polygon and partition as in \Cref{fig:partitionDual}, where the pieces $Q_1,Q_2,Q_3$ form a tripod with supports $D_1,D_2,D_3$ and tripod point $C$.
The star centers are coordinate maximum and the gray segments show how they are constructed.
The tripod is used to construct $A_3$.
}
\label{fig:extripod}
\end{figure}

\paragraph{Tripods.}
A structure that plays a crucial role in our characterization of the star centers is that of a \emph{tripod}; see \Cref{fig:extripod} for an example of a partition with one tripod.
Three pieces $Q_1,Q_2,Q_3$ with star centers $A_1,A_2,A_3$ form a \emph{tripod} with \emph{tripod point} $C$ if the following two properties hold.
\begin{itemize}
\item
There are concave corners $D_1,D_2,D_3$ of $P$ such that $D_i\in A_iC$ for each $i\in\{1,2,3\}$.
These corners are called the \emph{supports} of the tripod.

\item
The union $Q_1\cup Q_2\cup Q_3$ contains a (sufficiently small) disk centered at $C$.
\end{itemize}
Note that it follows that the segment $D_iC$ is on the boundary of the piece $Q_i$.
Such a segment $D_iC$ is called a \emph{leg} of the tripod.
Furthermore, the edges of $P$ incident to the supports $D_i$ are either all to the left or all to the right of the legs (when each leg $D_iC$ is oriented from $D_i$ to $C$); otherwise a disk around $C$ could not be seen by the star centers $A_i$.

\paragraph{Constructing star centers.}
We can define a set of points containing the star centers as follows.
Let $V_0$ be the corners of $P$ and define recursively $V_{i+1}$ as the intersection points between any two non-parallel lines each containing two points from $V_i$.
It follows that $V_i\subset V_{i+1}$.
Tripods cause star centers to depend on each other in complex ways:
If two of the participating star centers $A_1$ and $A_2$ are in $V_i\setminus V_{i-1}$, then the tripod point $C$ is in general in $V_{i+1}\setminus V_i$ and the third star center $A_3$ will be in $V_{i+2}$.
See for instance \Cref{fig:art} for two examples; both with unique optimal star partitions.
Here, all neighbouring pieces form tripods, and in the right figure only $V_i$ with $i=\Omega(n)$ contains all the star centers of the optimal partition.

We obtain powerful insights about the solution structure by considering a so-called \emph{coordinate maximum} optimal partition.
We can write the star centers $A_1,\ldots,A_k$ of an optimal partition in increasing lexicographic order (that is, sorted with respect to $x$-coordinates and using the $y$-coordinates to break ties).
We can then consider the vector of star centers $\langle A_1,\ldots,A_k\rangle$ which is maximum in lexicographic order among all sets of star centers of optimal partitions.
We show that there exists a partition realizing the maximum, which is our \emph{coordinate maximum} partition (\Cref{lem:coordmax}).
The star centers of the partition in \Cref{fig:extripod} have been maximized in this sense.
By analyzing a coordinate maximum partition, we conclude in \Cref{sec:structure} (\Cref{lem:property-of-coordmax}) that there are essentially only two ways in which a star center $A$ can be restricted.
In both cases, $A$ is forced to be contained in a specific half-plane $H$ bounded by a line $\ell$, and $\ell$ is of one of the following types:
(i) $\ell$ contains two corners of $P$, (ii) $\ell$ contains a tripod point $C$ and one of the associated supporting concave corners $D_i$.
The star center $A$ can then be chosen as the intersection point between two lines, each of type (i) or (ii).
Note that in each tripod, the legs $D_1C\cup D_2C\cup D_3C$ partition $P$ into three parts; since $P$ is a simple polygon, it is thus impossible that the star centers depend on each other in a cyclic way.
It follows that the star centers can be chosen from $V_i$ for a sufficiently high value of $i$.

\begin{figure}
\centering
\includegraphics[page=23]{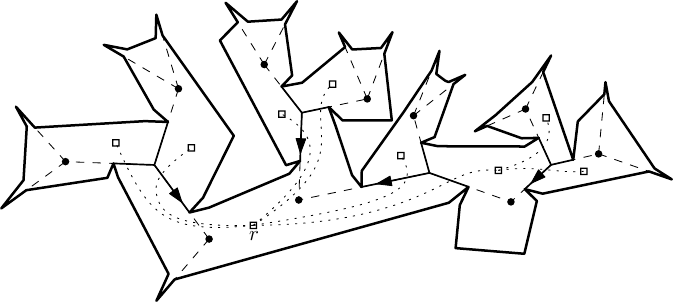}
\caption{A polygon $P$ with a star partition using ten pieces and four tripods.
The legs of the tripods partition $P$ into nine faces $\mathcal F$.
The disks are star centers and the squares denote the vertices of the dual graph $G$ of the faces $\mathcal F$.
The dashed segments indicate how the star centers are defined or used to define other centers by tripods.
The tripods have consistent orientation towards the root $r$ and the edges of the tree $\mathcal T$ are shown as dotted curves.
}
\label{fig:tripodconsistent}
\end{figure}

\paragraph{Orientation of tripods.}
Each tripod is defined from two of the participating star centers, say $A_1$ and $A_2$, and takes part in defining the third star center, $A_3$.
Hence, we can consider the tripod to have an orientation: it is directed from $(A_1, A_2)$ towards $A_3$.
The legs of all tripods partition $P$ into a set of faces $\mathcal F$; see \Cref{fig:tripodconsistent}.
One face can contain several pieces, since the tripod legs are in general only a subset of the piece boundaries.
We will denote one of the faces as the root $r$.
The faces $\mathcal F$ induce a dual graph $G$, in which each tripod corresponds to a triangle in $G$.
Traversing $\mathcal F$ in breadth-first search order from the root $r$ defines a rooted tree $T$, which is a subgraph of $G$.
Each node $u$ in $T$ has an even number of children---two for each tripod for which $u$ is the face closest to $r$ among the three faces containing the pieces of the tripod.
In order to successfully apply dynamic programming, we need the tripods to have a \emph{consistent} orientation in the following sense:
If the face $u$ is a parent of $v$, then the corresponding tripod should be directed towards the star center in~$u$.
As we will see in \Cref{sec:structure} (part of \Cref{thm:constructability}), there exists an optimal partition where the tripods have a consistent orientation.
This requires a modification to the coordinate maximum partition: Whenever a tripod violates the desired orientation, we choose a subset of the star centers and move them in a specific direction as much as possible to eliminate the illegal tripod.
We describe such a process that must terminate, and then we are left with tripods of consistent orientation.

With consistent orientation, the star centers in the leaves of $T$ belong to the set $V_1$ (which is constructed from lines through the corners of $P$), and in general, the star centers in a face $u$ can be constructed by tripods involving centers in the children of $u$ as well as lines through the corners of $P$ bounding the face $u$.
The star centers in the root face $r$ are constructed at last, potentially depending on all previously constructed star centers.

\begin{figure}
\centering
\includegraphics[page=24]{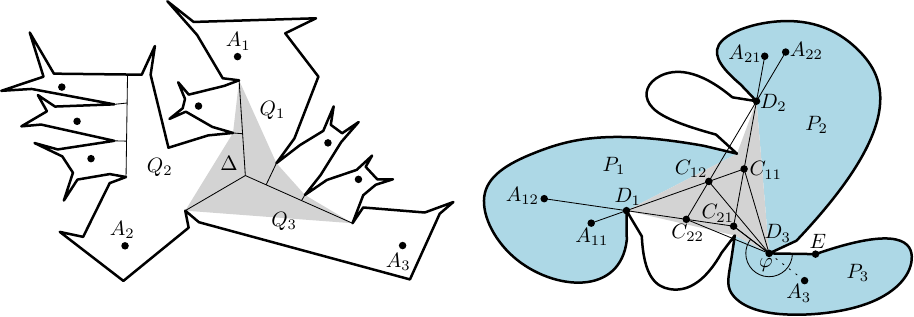}
\caption{\emph{Left:} A partition with a tripod and the pseudo-triangle shown in gray.
\emph{Right:}
For a tripod with supports $D_1,D_2,D_3$ directed towards $P_3$, we find the optimal partitions in the subpolygons $P_1$ and $P_2$.
There are two choices for the star centers that see $D_1$ and $D_2$, leading to four possible sets of legs of the tripod.
We want to choose the combination that minimizes the angle $\varphi$ from the edge $D_3E$ clockwise to the leg to $D_3$.
The angle is minimized when choosing $A_{12}$ and $A_{22}$, but this choice is invalid since the legs $D_2C_{22}$ and $D_3C_{22}$ would intersect the boundary of $P$.
We will therefore choose $A_{12}$ and $A_{21}$, which give the second-best option with the tripod point $C_{21}$.
Curved parts indicate that the details have not been shown.
}
\label{fig:greedy}
\end{figure}

\paragraph{Greedy choice.}
Consider three concave corners $D_1,D_2,D_3$ of $P$ which are the supports of a tripod in an optimal partition $\mathcal Q$.
Let the associated star centers be $A_1,A_2,A_3$ and suppose that the tripod is directed towards $A_3$.
The three shortest paths in $P$ between these supports enclose a region which we call the \emph{pseudo-triangle}~$\Delta$ of the tripod; see \Cref{fig:greedy} (left).
As we will see in \Cref{sec:structure} (part of \Cref{lem:property-of-coordmax}), there is an optimal star partition where no star centers are in the pseudo-triangle of any tripod, and this is a property we maintain throughout our modifications. 
Consider a connected component $P'$ of $P^- \coloneqq P\setminus \Delta$.
Note that $P'$ is separated from the rest of $P$ by a single diagonal of $P$ which is part of the boundary of $\Delta$.
Since no star centers are in $\Delta$, the restriction of $\mathcal Q$ to $P'$ is a star partition of $P'$.
Furthermore, in the optimal star partition we are working with, we can assume that this restriction is a \emph{minimum} partition of $P'$, since otherwise we could replace the partition in $P'$ by one with less pieces and use an extra piece to cover $\Delta$, thereby obtaining an equally good partition of $P$ without this tripod. 

Let $P_i$ be the connected component of $P^{-}$ containing $A_i$, so that the tripod is used to define the star center $A_3$ in $P_3$ using star centers $A_1$ and $A_2$ in $P_1$ and $P_2$, respectively; see \Cref{fig:greedy} (right).
In all connected components except $P_1,P_2,P_3$, we can choose an arbitrary optimal partition.
There may be several optimal partitions of $P_1$ and $P_2$, and any combination of two partitions may lead to different legs of the tripod (since $A_1$ and $A_2$ may be placed differently) and thus to different restrictions on the center $A_3$ in $P_3$.
In fact, there can be an exponential number of possible restrictions on $A_3$.
However, as shown in \Cref{sec:greedy-choice}, we can apply a \emph{greedy choice}:
We can use the combination of partitions of $P_1$ and $P_2$ that leads to the mildest restriction on $A_3$, in the sense that we want to minimize the angle $\varphi$ inside $P$ between the leg $D_3C$ and the edge of $P$ incident to $D_3$ which is also an edge of $P_3$.
Hence, we can use the greedy choice to restrict our attention to a single pair of optimal partitions of the subpolygons $P_1$ and $P_2$.

\paragraph{Bounding star centers and Steiner points.}
There are $O(n^3)$ possible triples of supports of tripods and using the greedy choice, we can restrict our attention to a specific pair of star centers that define the third star center for each tripod.
Since a star center may be defined from two tripods, we get a bound of $O(n^6)$ on the number of star centers that we need to consider.

We also need polynomial bounds on the other points defining the separators, namely the boundary points $B_i$ that the pieces share with $P$ and the points $S$ that neighbouring pieces share with each other.
Some of these points may be corners of $P$, but the rest will be Steiner points, i.e., not corners of $P$.
Suppose that we know the star centers $\mathcal A=\{A_1,\ldots,A_k\}$ of the pieces in an optimal partition.
We can then consider a partition $\mathcal Q = \{Q_1,\ldots,Q_k\}$ where $A_i$ is a star center of $Q_i$ and we have maximized the vector of areas $\langle a(Q_1),\ldots,a(Q_k)\rangle$ in lexicographic order.
As we will see (\Cref{lem:areamax}), such a partition $\mathcal Q$ exists (for any fixed set of star centers $\mathcal A$), and in \Cref{sec:area-max} we show that $\mathcal Q$ has the property that each Steiner point in the interior of $P$ is defined by at most five star centers and two corners of $P$.
Hence, there are at most $O(n^{6\cdot 5}\cdot n^2)=O(n^{32})$ relevant Steiner points to try out.
A Steiner point on the boundary $\partial P$ will be defined by an edge of $P$ and, in the worst case, a line through two star centers, which gives $O(n^{13})$ possibilities.
Hence, we can bound the number of possible long separators by $O(n^{13}\cdot n^6\cdot n^{32}\cdot n^6\cdot n^{13})=O(n^{70})$.

In \Cref{app:struct}, we give an elementary proof of a structural result that reduces the number of Steiner points needed on the boundary of $P$ to only $O(n)$.
It might be possible to use this result to design a faster algorithm than the one presented here, but the proof relies on many modifications to the partition, so it is not clear if our algorithm can be modified to find the resulting partition.

\paragraph{Algorithm.}
Our algorithm now works as follows.
In the first phase, we consider each diagonal of $P$, and we recursively find all relevant optimal partitions of the subpolygon on one side of the diagonal.
Once this has been done for all diagonals, we consider each possible triple of concave corners of $P$ supporting a tripod, and we use the greedy choice to select the pair of star centers that can be used to define the third star center of the tripod.
We then construct $O(n^6)$ possible star centers by considering all pairs of (i) tripods, (ii) lines through two corners of $P$, and (iii) one tripod and one line through two corners of $P$.
The set of potential star centers leads to polynomially many Steiner points and separators as described above.
In the second phase, we use dynamic programming to find out how many pieces we need in the subpolygon defined by each separator. The total running time turns out to be $O(n^{107})$ or within $O(n^{105})$ arithmetic operations.

\subsection{Open Problems \& Discussion}

Although polynomial, our algorithm is too slow to be of much practical use.
Our main result is showing that the problem is polynomial-time solvable, so
in order to facilitate understanding and verification of our work, we decided to give a description of the algorithm that is as simple as possible, and consequently we did not further optimize the running time.
Although we believe that it is possible to optimize the algorithm significantly (for instance using structural insights from \Cref{app:struct}), it seems that our approach will remain impractical.
Hence, it is interesting whether a practical constant-factor approximation algorithm exists.
For the minimum convex partition problem, the following wonderfully simple algorithm produces a partition with at most twice as many pieces as the minimum~\cite{chazelle1985approximation}:
For each concave corner $C$ of the input polygon $P$, cut $P$ along an extension of an edge incident to $C$ until we reach the boundary of $P$ or a previously constructed cut.
It would be valuable to find a practical and simple algorithm for star partitions with similar approximation guarantees.

Higher-dimensional versions of the minimum star partition problem are also of great interest and we are not aware of any work on such problems from a theoretical point of view.
The high-dimensional problems are similarly well-motivated from a practical angle, since in motion-planning the configuration space is in general high-dimensional and a star partition of the free space can then be used to find a path from one configuration to another, as described in \Cref{sec:practical} (in fact, all the cited papers related to motion planning~\cite{varadhan2005star,dobson2014sparse, SPBG:SPBG07:073-080, lien2009planning, zhang2006fast} also describe a high-dimensional setting).
We note that the three-dimensional version of the minimum convex partition problem already received some attention, e.g.~\cite{DBLP:journals/algorithmica/ChazelleP97,DBLP:journals/siamcomp/BajajD92,DBLP:journals/siamcomp/Chazelle84}.

Many interesting partitioning problems of simple polygons with Steiner points are still open.
Surprisingly, one problem that remains open is arguably the most basic of all problems of this type, namely, that of partitioning a simple polygon $P$ into a minimum number of \emph{triangles}.
If $P$ has $n$ corners, a maximal set of pairwise interior-disjoint diagonals always partitions $P$ into $n-2$ triangles and finding such a triangulation is a well-understood problem with a long history, culminating in Chazelle's famous linear-time algorithm~\cite{DBLP:journals/dcg/Chazelle91}.
In general, however, there exist partitions into fewer than $n-2$ triangles and it is an open problem whether an optimal partition can be found in polynomial time.
Asano, Asano, and Pinter~\cite{DBLP:journals/jal/AsanoAP86} showed that a minimum triangulation without Steiner points can be found in polynomial time.
When Steiner points are allowed, they gave examples of polygons in which points from the set $V_1$ are needed, and they conjecture that there are instances in which points from the set $V_i$ for arbitrarily large values of $i$ are needed (i.e.\ points which have arbitrarily large degrees).
Another classical open problem is to partition a simple polygon into a minimum number of spirals with Steiner points allowed.
A \emph{spiral} is a polygon where all concave corners appear in succession.
The problem of partitioning into spirals was originally motivated by feature generation for syntactic pattern recognition~\cite{DBLP:journals/tc/FengP75} and a polynomial-time algorithm finding the optimal solution to the problem without Steiner points is known~\cite{DBLP:journals/siamcomp/Keil85}. However, no algorithm is known for the unrestricted problem.

We hope that our techniques may be useful when designing algorithms to solve the above-mentioned problems. In particular, considering extreme partitions can lead to natural piece boundaries which in turn can be exploited using a dynamic programming approach. Computing such partitions in two phases, first computing potential locations of Steiner points that are subsequently used in guessing separators of pieces in an optimal solution, presents itself as a general paradigm to attack problems of this type.

\subsection{Organization}

The remainder of this work is organized as follows.
In~\Cref{sec:prelim}, we define various types of polygons, partitions, and other central concepts.
We also give lemmas ensuring the existence of partitions that are extreme in terms of the coordinates of the star centers or the areas of the pieces.
In \Cref{sec:structure}, we study coordinate maximum partitions and the structures arising from tripods. These structural properties help us find a set of polynomially many potential points to use as star centers.
In \Cref{sec:area-max}, we study area maximum partitions. The insights gained on their structure help us characterize all other Steiner points to use as corners in our partition, given a set of potential star centers (coming from the previous section).
Finally in \Cref{sec:algorithm}, we show how to use our structural results to design our two-phase dynamic programming algorithm.

\section{Preliminaries}\label{sec:prelim}

In this section we first cover some basic definitions to then turn towards partitions that are maximum with respect to either the area of the pieces or the coordinates of the star centers.

\subsection{Definitions}
We say that a pair of segments \emph{cross} if their interiors intersect.

\paragraph{Polygons.}
A \emph{simple polygon} is a compact region in the plane whose boundary is a simple, closed curve consisting of finitely many line segments.
For technical reasons, we allow the pieces of a partition to be weakly simple polygons. 
A \emph{weakly simple polygon} $Q$ is a simply-connected and compact region in the plane whose boundary is a union of finitely many line segments.
In particular, a simple polygon is also a weakly simple polygon, but the opposite is not true in general.
For instance, a weakly simple polygon $Q$ may have a disconnected or even empty interior.
However, just as for a simple polygon, a weakly simple polygon $Q$ can be defined by its edges in counterclockwise order around the boundary.
These edges form a closed boundary curve $\gamma$ of $Q$.
Since $Q$ is weakly simple, some corners may coincide, and edges may overlap.
A perturbation of $\gamma$ that is arbitrarily small with respect to the Fréchet distance can turn $Q$ into a simple polygon~\cite{DBLP:journals/dcg/AkitayaAET17,DBLP:conf/soda/ChangEX15}.
This perturbation may involve the introduction of more corners.
For instance, if $Q$ is just a line segment, then $Q$ has only two corners, and one more is needed to obtain a simple polygon.
We denote the boundary of a (weakly) simple polygon $Q$ as $\partial Q$.

We sometimes consider points that lie on so-called extensions.
Given a polygon $P$ and a segment $CD\subset P$, the \emph{extension} of $CD$ is the maximal segment $C'D'$ such that $CD\subset C'D'\subset P$.

\paragraph{Star-Shaped Polygons.}
A (weakly) simple polygon $Q$ is called \emph{star-shaped} if there is a point $A$ in $Q$ such that for all points $B$ in $Q$, the line segment $AB$ is contained in $Q$.
Such a point $A$ is called a \emph{star center} of $Q$.
We denote by $\ker(Q)$ the set of all star centers of $Q$, and it is well-known that $\ker(Q)$ is a convex polygonal region in $Q$.
Throughout the paper, we use the symbol ``$Q$'' to denote a star-shaped polygon and ``$A$'' to denote a fixed star center of such a polygon.
When proving our structural results, we repeatedly use the following lemma to trim some of the pieces of a star partition.

\begin{lemma}\label{lemma:cut}
Let $Q$ be a star-shaped polygon with star center $A$, and let $H$ be an open half-plane bounded by a line $h$.
Let $C$ be a connected component of the intersection $Q\cap H$ and suppose that $A\notin C$.
Then $A\notin H$ and $Q' \coloneqq Q\setminus C$ is also a star-shaped polygon with star center $A$.
\end{lemma}

\begin{proof}
First, if $A \in H$, then $Q\cap H$ consists of a single connected component as $Q$ is star-shaped. However, this implies $A \in C$, which we assumed not to be the case. Thus, $A \notin H$.

Now consider a point $B\in Q'$.
If $B\notin H$, then $AB\cap H=\emptyset$. Hence we also have $AB\subset Q'$, since $AB\subset Q$.
Otherwise (if $B \in H$), then $B$ is in a connected component $D$ of $Q\cap H$ with $D$ different from $C$.
Let $X$ be the intersection point of $h$ and $AB$.
Since $AB\subset Q$, we must have $XB\subset D$.
As $D\subset Q'$, we then have $AB\subset Q'$.
We therefore conclude that $Q'$ is star-shaped.
\end{proof}

\paragraph{Partitions.}
We will eventually consider star partitions of a modification $\tilde P$ of the input polygon $P$ obtained by making incisions into the interior of $P$ from corners.
Thus, $\tilde P$ is a weakly simple polygon covering the same region as $P$, but $\tilde P$ has some extra edges on top of each other that stick into the interior of $P$.
To accommodate this, we define star partitions in a way that allows both the input polygon $P$ and the pieces to be weakly simple polygons.
We define a \emph{star partition} of a weakly simple polygon $P$ to be a set of weakly simple star-shaped polygons $Q_1,\ldots,Q_k$ such that after an arbitrarily small perturbation of $P$ and $Q_1, \dots, Q_k$, we obtain \emph{simple} polygons $P'$ and $Q'_1,\ldots,Q'_k$ with the following properties:

\begin{enumerate}
    \item The polygons $Q'_1,\ldots,Q'_k$ are pairwise interior-disjoint. \label{partprop:1}
    
    \item $\bigcup_{i=1}^k Q'_i =P'$. \label{partprop:2}
\end{enumerate}

Note that this implies that the weakly simple polygons $Q_1,\ldots,Q_k$ must also have properties \ref{partprop:1} and \ref{partprop:2} (with $P'$ replaced by $P$ and $Q'_i$ replaced by $Q_i$ for all $i\in\{1,\ldots,k\}$), since otherwise a large perturbation would be needed for them to be transformed into simple polygons with the required properties.
However, it would not be sufficient to define a partition as a set of weakly simple polygons with properties \ref{partprop:1} and \ref{partprop:2} alone.
This would, for instance, allow two pieces with empty interiors (such as two segments) to properly intersect each other, which is not intended.
Our algorithm may produce weakly simple pieces which are not simple, since the boundary can meet itself at the star center; see \Cref{fig:weaklysimple}.
As demonstrated in the figure, by applying \Cref{lemma:cut}, such a piece $Q$ can ``steal'' a bit from the neighbouring pieces, which turns $Q$ into a simple polygon $Q'$, resulting in a partition consisting of simple polygons.

\begin{figure}
\centering
\includegraphics[page=34]{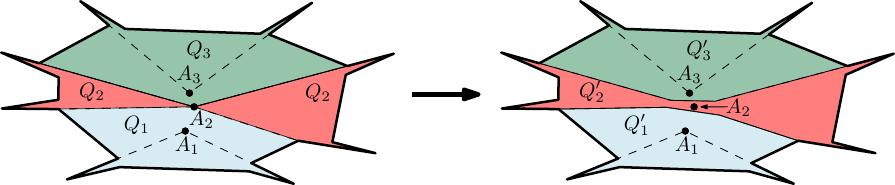}
\caption{\emph{Left:}
A partition that our algorithm may produce, involving the piece $Q_2$ which is only weakly simple.
\emph{Right:}
We can assign a bit of the neighbouring pieces to $Q_2$ and obtain a partition into simple polygons.
}
\label{fig:weaklysimple}
\end{figure}

\paragraph{(Important) Sight Lines.}
Given a star-shaped polygon $Q$ and a star center $A \in \ker(Q)$, each segment that connects a corner of $Q$ with the center $A$ is called a \emph{sight line} of $Q$. 
A sight line $\ell$ is called an \emph{important sight line} if it contains a corner $D$ of $P$ in its interior. We call $D$ the \emph{support} of $\ell$.
If there are multiple candidates, we define the corner farthest from the star center as the support.

\paragraph{Tripods.}
In a star partition, three pieces $Q_1,Q_2,Q_3$ with star centers $A_1,A_2,A_3$ form a \emph{tripod} with \emph{tripod point} $C$ if the following properties hold.
\begin{itemize}
\item
$A_iC$ is an important sight line of $Q_i$ with support $D_i$, for each $i \in \{1, 2, 3\}$. 
These concave corners $D_1, D_2, D_3$ of $P$ are called the \emph{supports} of the tripod.
\item
The union $Q_1\cup Q_2\cup Q_3$ contains a (sufficiently small) disk centered at $C$.
\item The three pieces $Q_1, Q_2, Q_3$ have \emph{strictly} convex corners at $C$.
\end{itemize}
Tripods can be necessary in optimal solutions, see Figure~\ref{fig:art} for such an example.

The three segments $D_1C, D_2C, D_3C$ are called the \emph{legs} of the tripod.
The polygon bounded by the three shortest paths in $P$ between pairs of the supports $D_1,D_2,D_3$ is called the \emph{pseudo-triangle} of the tripod, and these shortest paths are called \emph{pseudo-diagonals}. 

\begin{lemma}
    \label{lem:separate-tripods}
    Let $\mathcal{T}_1, \mathcal{T}_2$ be two distinct tripods in a star partition $\mathcal{Q}$. The interiors of the pseudo-triangles of $\mathcal{T}_1$ and $\mathcal{T}_2$ are disjoint. 
\end{lemma}

\begin{proof}
The legs of $\mathcal{T}_1$ partition $P$ into three regions $R_0, R_1, R_2$.
    Since tripod legs are boundary segment of pieces, they cannot cross each other. 
    Hence, all the tripod legs of $\mathcal{T}_2$ must lie in one of $R_0$, $R_1$, $R_2$; without loss of generality, assume they lie in $R_0$. 
    Then the pseudo-triangle of $\mathcal{T}_2$ is a subpolygon of $R_0$. 
    Towards a contradiction, assume that the interiors of the two pseudo-triangles are not disjoint.
    It is impossible that $\mathcal{T}_2$ is contained in $\mathcal{T}_1$, since it would mean that the corners of $\mathcal{T}_2$ are a subset of the corners of one pseudo-diagonal of $\mathcal{T}_1$, and a pseudo-triangle cannot be made from corners on a concave chain.
    Hence, if $\mathcal{T}_1$ and $\mathcal{T}_2$ are not interior-disjoint, there is an edge $e$ of the pseudo-triangle of $\mathcal{T}_2$ that crosses the boundary of the pseudo-triangle of $\mathcal{T}_1$; see \Cref{fig:separate-tripods}.
    As the pseudo-triangle of $\mathcal{T}_1$ in $R_0$ is bounded by the two tripod legs (which $e$ does not cross) and a concave chain, the segment $e$ must have an endpoint inside the pseudo-triangle. As both endpoints of $e$ are supported by a concave corner of $P$, we obtain a contradiction with the fact that no vertex of $P$ is contained in the interior of the pseudo-triangle.
\end{proof}

\begin{figure}
\centering
\includegraphics{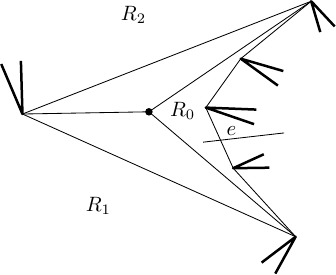}
\caption{If the interiors of the pseudo-triangles of two tripods intersect, then either a vertex of a pseudo-triangle is in the other pseudo-triangle, or a pseudo-triangle crosses a tripod leg of another tripod.}
\label{fig:separate-tripods}
\end{figure}

\subsection{Coordinate and Area Maximum Partitions}
\label{sec:existence}

\paragraph{Coordinate Maximum Partition.}
We define the lexicographic order $\preceq$ of vectors $v_1, v_2 \in \RR^d$ so that $v_1 \preceq v_2$ iff $v_1 = v_2$ or $i$ is the first dimension that $v_1$ and $v_2$ differ in and $v_1[i] < v_2[i]$.
Note that this definition carries over to star centers in a straightforward manner.
For a star-shaped polygon $Q$, the \emph{maximum} star center of $Q$ is the star center (i.e.\ point in $\ker(Q)$) with the lexicographically largest value.
For a polygon $P$, consider an optimal star partition $\mathcal Q$ with maximum star centers $A_1,\ldots,A_k$ sorted in lexicographic order, and define $c(\mathcal Q)=\langle A_1,\ldots,A_k\rangle$ to be the combined coordinate vector.
If $c(\mathcal Q)$ is maximum in lexicographic order among all optimal star partitions of $P$, we say that $\mathcal Q$ is a \emph{coordinate maximum} optimal partition.
In other cases, it is useful to consider a partition with given star centers where the vector of \emph{areas} of the pieces has been maximized.
In this section, we provide lemmas that ensure the existence of such partitions.
The proofs are deferred to~\Cref{apx:existence}.

\begin{lemma}\label{lem:coordmax}
For any simple polygon $P$, there exists a coordinate maximum optimal star partition.
\end{lemma}

\paragraph{Restricted Coordinate Maximum Partitions.}
It will sometimes be necessary to change the direction in which we maximize a specific subset of star centers, while keeping the remaining ones fixed.
Furthermore, we often have to restrict the star centers that we are optimizing to a subpolygon $F \subseteq P$.
For this, we use the following generalization of \Cref{lem:coordmax}, the proof of which is analogous.
Given a vector $d \in \RR^2$, we define $d^\perp \in \RR^2$ to be the vector orthogonal to $d$ obtained by rotating $d$ counterclockwise by $\pi/2$.

\begin{lemma}[Restricted coordinate maximization]\label{lem:coordmaxsubset}
Consider a simple polygon $P$ and an optimal star partition with star centers $A_1,\ldots,A_k$.
Let $i\leq k$ and suppose that $A_i,\ldots,A_k\in F$ for a polygon $F\subseteq P$. 
Let $d\in\RR^2 \setminus \{(0,0)\}$ be a vector.
There exists a star partition of $P$ with star centers $A_1,A_2,\ldots,A_{i-1},A^*_i,A^*_{i+1},\ldots,A^*_k$ where $A^*_i,A^*_{i+1},\ldots,A^*_k\in F$ and $\langle A^*_i\cdot d, A^*_i\cdot d^\perp,A^*_{i+1}\cdot d, A^*_{i+1}\cdot d^\perp,\ldots,A^*_k\cdot d,A^*_k\cdot d^\perp\rangle$ is maximum in lexicographic order among all star partitions with fixed star centers $A_1,\ldots,A_{i-1}$ and for which the remaining star centers are restricted to~$F$. 
\end{lemma}

The partition described in \Cref{lem:coordmaxsubset} is called the \emph{restricted coordinate maximum} optimal star partition along $d$, within $F$ and with fixed star centers $A_1, \dots, A_{i-1}$.
Note that a coordinate maximum optimal partition is a restricted coordinate maximum one along $d=(1,0)$, within $P$ and with no fixed star center.

\paragraph{Area Maximum Partition.}
Consider a polygon $P$ and a star partition $\mathcal Q=\{Q_1,\ldots,Q_k\}$ of $P$ with corresponding star centers $\mathcal A=\{A_1,\ldots,A_k\}$.
We say that $\mathcal Q$ is \emph{area maximum} with respect to $\mathcal A$ if the vector of areas $a(\mathcal Q)=\langle a(Q_1),\ldots,a(Q_k)\rangle$ is maximum in lexicographic order among all partitions of $P$ with star centers $\mathcal A$.

\begin{lemma}\label{lem:areamax}
Let $P$ be a polygon and suppose that there exists a star partition of $P$ with star centers $\mathcal A=\{A_1,\ldots,A_k\}$.
Then there exists a partition which is area maximum with respect to $\mathcal A$.
\end{lemma}

\section{Structural Results on Tripods and Star Centers}
\label{sec:structure}

In this section, we will present a construction process which can construct all the star centers and tripods in some optimal solution within linearly many steps. 
To achieve this goal, we first need to pick an optimal solution with good properties. We do this by considering \emph{restricted coordinate maximum partitions} (see \cref{lem:coordmaxsubset}).

\begin{lemma}
    Consider a simple polygon $P$ and an optimal star partition $\mathcal{Q} = \{ Q_1, \dots, Q_k \}$ with corresponding star centers $A_1,\ldots,A_k$.
    There exists an optimal star partition consisting of simple polygons with the same star centers, such that no four pieces meet in the same point and no star center lies in the interior of a sight line.  
\end{lemma}

\begin{proof}
    We first turn the weakly simple star partition into a star partition with simple polygons; see Section~\ref{sec:prelim}.
    We then modify the partition, not moving the star centers, so that no four pieces contain the same point.
    Assume that there exists a point $C$ such that $C \in Q_1 \cap \dots \cap Q_m$ for some $m \geq 4$. 
    Without loss of generality, assume $Q_1, \dots, Q_m$ appear in clockwise order around $C$. 
    Let $\alpha_i$ be the angle of $Q_i$ at $C$. 
    Since $\sum_{i = 1}^m \alpha_i \leq 2\pi$, we have either $\alpha_1 + \alpha_2 \leq \pi$ or $\alpha_3 + \alpha_4 \leq \pi$. 
    Without loss of generality, assume $\alpha_1 + \alpha_2 \leq \pi$. 
    We now decrease the number of pieces containing $C$, while not creating an intersection point of four or more pieces. 
    Recall that $A_i$ is the star center of $Q_i$.
    We consider two cases; see \cref{fig:reduce-pieces}:

    \begin{itemize}
        \item \textbf{\boldmath $A_1$ is not on an extension of the shared boundary with $A_2$ or vice versa.} 
        If $A_1$ is not on an extension of the shared boundary with $A_2$, then $Q_1$ can take a small enough triangle around $C$ from $Q_2$, while the two new Steiner points in the partition are contained in at most three pieces, namely $Q_1, Q_2, Q_3$. A similar modification is possible if $A_2$ is not on an extension of the shared boundary with $A_1$. 
        \item \textbf{\boldmath Both $A_1$ and $A_2$ are on an extension of their shared boundary.} Without loss of generality, assume $A_1$ is closer to $C$ than $A_2$. Then $Q_1$ can take a sufficiently small triangle around the segment $A_1C$, while not creating any new point where four pieces meet.
    \end{itemize}
    Hence, eventually we obtain a star partition that has no four pieces meeting in the same point. 
    
    \begin{figure}
        \centering
        \includegraphics{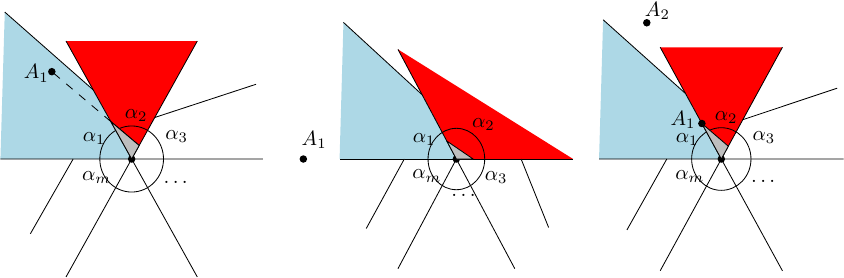}
        \caption{Reduce the number of pieces containing the same point. The two figures on the left show the modification we perform if none of the star centers are on an extension of the shared boundary with the other star center. The right figure shows the case when both star centers lie on an extension of the shared segment. Blue region marks the piece $Q_1$, red regions is the piece $Q_2$, and the gray region is going to be transferred from $Q_2$ to $Q_1$. }
        \label{fig:reduce-pieces}
    \end{figure}

    We now modify the partition to remove all sight lines that contain some star center in their interior. 
    Let $\ell = A_0C$ be a sight line that contains $A_1, A_2, \dots, A_m$ in its interior. 
    We first choose a sequence of points $C_1, C_2, \dots, C_m$ along the next edge of $Q_0$; see \cref{fig:center-in-sight-line}. We then give the quadrilateral $A_iA_{i+1}C_{i+1}C_i$ to piece $Q_i$ for all $i \in \{1, 2, \dots, m-1\}$. Finally we give $A_mCC_m$ to $Q_m$. 
    This modification removes all star centers from the interior of one sight line while no newly created sight line contains a star center in its interior.
    It is easy to check that this modification of the partition does not make four pieces meet. 
    \begin{cfigure}
        \centering
        \includegraphics{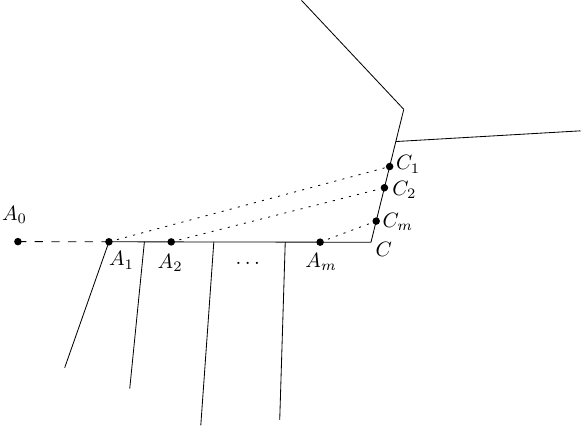}
        \caption{Redistributing pieces to remove a sight line that contains star centers in its interior.}
        \label{fig:center-in-sight-line}
    \end{cfigure}
\end{proof}

The main tool in this section is \emph{restricted coordinate maximum partitions} defined in \cref{sec:prelim}. 
The following lemma captures one of our key combinatorial results on a star center in a restricted coordinate maximum partition from \cref{lem:coordmaxsubset}. 
Intuitively, if one moves a star center $A_k$ of an optimal star partition in the direction $d$ as far as possible without moving other star centers (but possibly changing what region of $P$ each piece contains), then there are only a few reasons to get stuck.

\begin{lemma}
\label{lem:property-of-coordmax}
Consider the restricted coordinate maximum optimal star partition consisting of simple polygons $\mathcal Q=\{Q_1,\ldots,Q_k\}$ along $d$ and with fixed star centers $A_1,\ldots,A_{k-1}$ (so that only the coordinates of the last center $A_k$ have been maximized). 
Assume $A_k$ is restricted within a polygon $F \subseteq P$. 
Suppose that no four pieces meet in the same point and no star center is in the interior of any sight line.
Then $A_k$ lies on the intersection of two non-parallel segments of the following types:
\begin{itemize}
    \item an edge of $F$,
    \item an important sight line of $Q_k$, not containing any other star center, which is 
    \begin{itemize}
        \item an extension of an edge of $P$, or
        \item on the extension of a diagonal of $P$ that connects two concave corners, or
        \item an extension of a tripod leg (see \cref{sec:prelim}), and no star center is in the interior of the pseudo-triangle of this tripod.
    \end{itemize}
\end{itemize}
\end{lemma}

\begin{proof}
    We can choose $A_k$ freely inside $\ker(Q_k) \cap F$ while all the pieces remain the same. 
    By coordinate maximization, $A_k$ must be a corner of $\ker(Q_k) \cap F$.
    Let $\mathcal{S}$ denote the set of edges $e$ of $Q_k$ such that $A_k$ is on the extension of $e$. 
    Let $\mathcal{S'}$ denote the set of edges of $F$ that $A_k$ lies on. 
    Since all edges of $\ker(Q_k)$ come from extensions of edges of $Q_k$, there must be two non-parallel segments in $\mathcal{S} \cup \mathcal{S'}$.

    We call a segment in $\mathcal{S}$ \emph{good} if it is collinear to a segment that is of the described types in the lemma statement; otherwise, we call it \emph{bad}.
    In the remainder of the proof, we modify the partition $\mathcal{Q}$ while \textit{not moving any star centers} and \textit{never creating any new important sight lines of $Q_k$}, which means that the two good segments we find at last are also good segment of the initial star partition $\mathcal Q$. 
    At the same time, we decrease the number of bad segments in $\mathcal{S}$ until all segments in $\mathcal{S}$ are good.
    In the end, either $A_k$ satisfies the lemma, or else we cannot find two non-parallel segments in $\mathcal{S} \cup \mathcal{S}'$, which would mean that $A_k$ is not at a corner of $\ker(Q_k) \cap F$ therefore contradicting that $A_k$ is optimal with respect to coordinate maximization.

    In the remainder of the proof, all star centers and corners of $P$ on $\partial Q_k$ are considered as corners of $Q_k$, so there can be several collinear consecutive segments of $Q_k$. 
    First, we consider the case that $A_k$ is the \emph{endpoint} of a bad segment in $\mathcal{S}$. 
    \begin{enumerate}
        \item \textbf{\boldmath $A_k$ is a corner of $P$.} Then $A_k$ must be a corner of $F$ as $F \subset P$.
        Hence $A_k$ is at two edges of $P$ and the lemma holds. 
        \item \textbf{\boldmath $A_k$ is in the interior of an edge of $P$.}
        If $A_k$ is a corner of $F$, then the lemma holds. Otherwise, $A_k$ is in the interior of an edge of $F$ that is collinear to the boundary of $P$. 
        We can assign an arbitrarily small region around $A_k$ to $Q_k$ such that $A_k$ is not an endpoint of a bad segment anymore; see \cref{fig:center-on-boundary-P}. 
        Since we do not create any new important sight line in $Q_k$, we do not introduce any new good segments in $\mathcal{S}$. 
        The only new bad segment we introduce is parallel to the boundary edge of $P$ that $A_k$ lies on, which can also be removed from $S$ without breaking the assumption that $S \cup S'$ contains two non-parallel segments--- since there must be a parallel edge of $F$. 
        \begin{figure}
            \centering
            \includegraphics{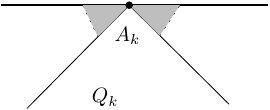}
            \caption{Dealing with the case that $A_k \in \partial P$. The gray region marks the region that we give to $Q_k$. }
            \label{fig:center-on-boundary-P}
        \end{figure}
        \item \textbf{\boldmath $A_k$ is in the interior of $P$.} This can happen when $A_k$ is either a convex corner of all pieces touching it, a concave corner of the piece $Q_k$, or a concave corner of some other piece. In all cases, we transfer a sufficiently small area around $A_k$ to $Q_k$ and thereby make $A_k$ not be an endpoint of a bad segment; see \cref{fig:center-on-boundary}. 
        \begin{cfigure}
            \centering
            \includegraphics{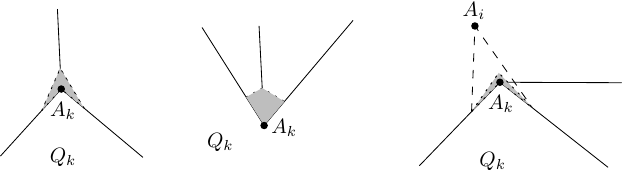}
            \caption{Dealing with the case that $A_k \in \partial Q_k$. The gray region marks the region that we give to $Q_k$. The first figure shows the case that $A_k$ is a convex corner for all pieces; the second shows the case that $A_k$ is a concave corner of $Q_k$; the third shows the case that $A_k$ is a concave corner of other pieces. }
            \label{fig:center-on-boundary}
        \end{cfigure}
    \end{enumerate}

    In the following $A_k$ is not an \emph{endpoint} of a bad segment in $\mathcal{S}$.
    Let $e = C_1C_2$ be a segment in $\mathcal{S}$ and let $C_2$ be the farther end of $e$ from $A_k$.
    Let $C_0$ be the nearest vertex along $A_kC_1$ to $A_k$. 
    Note that $C_0C_1 \subset \partial Q_k$. 
    Note that this implies $C_0 \neq A_k$, as $A_k$ is not the endpoint of a bad segment in $\mathcal{S}$.
    Furthermore, let $C_3$ be the farthest point from $A_k$ in the direction of $C_2$ on the extension of $A_k C_2$ such that $C_2C_3 \subset \partial Q_k$ (it might be the case that $C_3 = C_2$). 
    Let $C_4$ be the next corner of $C_3$ on $\partial Q_k$, and let $C_{-1}$ be the previous corner of $C_0$ on $\partial Q_k$. We again consider multiple cases:
    \begin{enumerate}
        \item \textbf{Another star center $A_i$ is on $A_kC_3$.} According to our assumptions in the lemma statement, no star center is in the interior of a sight line, thus we have $A_i = C_3$. 
        We can then transfer the triangle $C_0C_3C_4$ from $Q_k$ to $Q_i$ and the number of bad segments in $\mathcal{S}$ is thereby reduced; see \cref{fig:end-at-center}.
        \begin{figure}
            \centering
            \includegraphics{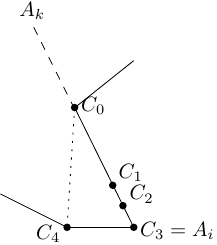}
            \caption{If the sight line ends at another star center $A_i$, we can give the triangle $C_0C_3C_4$ to $Q_i$ and reduce $|\mathcal{S}|$.}
            \label{fig:end-at-center}
        \end{figure}
        \item \textbf{\boldmath No corner of $P$ is in the interior of the sight line $A_kC_2$.} Or equivalently, $A_kC_2$ is not an important sight line. In this case we give a sufficiently small triangle $C'C_0C_2$ to $Q_k$, where $C'$ is sufficiently close to $C_0$ on the segment $C_{-1}C_0$; see \cref{fig:no-corner}. 
        According to \cref{lemma:cut}, all pieces that are cut by the segment $C'C_2$ are still star shaped. 
        This way we reduce the size of $\mathcal{S}$ by removing $C_1C_2$.
        \begin{cfigure}
            \centering
            \includegraphics{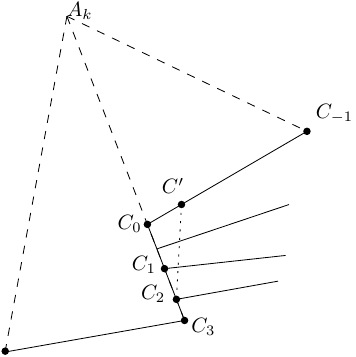}
            \caption{No corners of $P$ on $A_kC_2$.}
            \label{fig:no-corner}
        \end{cfigure}
        
        In the remainder we can assume that the sight line $A_kC_2$ is supported by a corner of $P$, i.e., it is an important sight line. 
        Since $A_kC_2$ is covered by $A_kC_3$, $A_kC_3$ is also an important sight line. Let $D$ be the support (\cref{sec:prelim}) of $A_kC_3$.
        In the remainder we try to remove $DC_3$ from~$\mathcal{S}$. 
        
        \item  \label{case:ak_in_interior_of_p} \textbf{\boldmath $C_3$ is a convex corner of $P$ and not adjacent to $D$.} Note that if $C_3$ was a convex corner of $P$ adjacent to $D$, then $A_kC_3$ would be an extension of an edge $DC_3$ of $P$, which is a good segment in $\mathcal{S}$.
        We can remove a sufficiently small triangle $DC_3C'$ from $Q_k$ for $C'$ close enough to $C_3$ on segment $C_3C_4$, and distribute the triangle to the neighboring pieces by extending the edges that end at $DC_3$; see \cref{fig:end-at-convex-corner}.
        Since we do not create concave corners in any pieces they remain star-shaped.
        The same argument is also applicable if $DC_3$ ends in the interior of an edge of $P$ as we can consider the intersection point as a degenerate convex corner of $P$. 
        \begin{figure}
            \centering
            \includegraphics{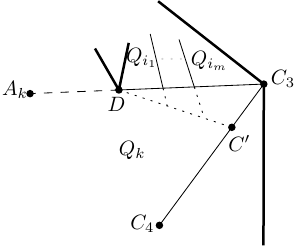}
            \caption{The case when the boundary of $Q_k$ ends at a convex corner of $P$. }
            \label{fig:end-at-convex-corner}
        \end{figure}
        
        In the remainder $C_3$ is in the interior of $P$.
        
        \item \textbf{\boldmath $C_3$ is a concave corner of some piece $Q_i$.} Then $C_3$ is a convex corner of $P \setminus Q_i$ and we can use a similar modification to remove $DC_3$ from $\mathcal{S}$ as in the previous case; see \cref{fig:tripod-2}. 
        \begin{cfigure}
            \centering
            \includegraphics[width=\linewidth]{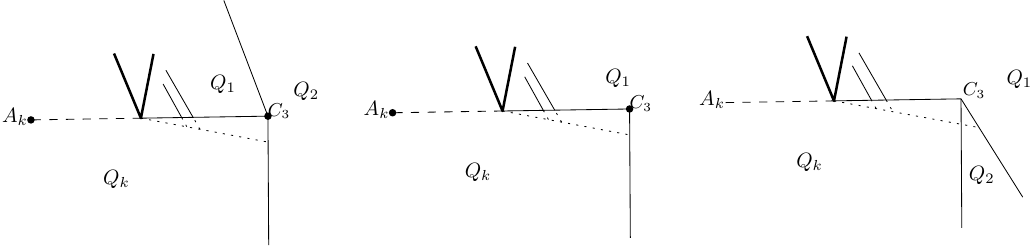}
            \caption{$C_3$ is a concave corner of some piece.}
            \label{fig:tripod-2}
        \end{cfigure}

        In the remainder, $C_3$ a convex corner of all its adjacent pieces.
        According to the assumption of the lemma that no four pieces meet at the same point, $C_3$ is contained in at most three pieces. 
        Since the angle of $Q_k$ at $C_3$ is strictly less than $\pi$, there actually must be exactly three pieces containing $C_3$. 
        With slight abuse of notation, let $Q_0 = Q_k, Q_1, Q_2$ be these three pieces in clockwise order; let $\alpha_i$ be the angle of $Q_i$ at $C_3$; and let $A_i$ be the star center of $Q_i$. 
        
        \item \textbf{$C_3$ is not a tripod point.}
        If an edge at $C_3$ is not covered by an important sight line, we can modify the partition and remove $DC_3$ from $\mathcal{S}$; see \cref{fig:tripod-3}. 
        Otherwise, all edges at $C_3$ are covered by important sight lines. As $C_3$ is not a tripod point, we have that $\alpha_1 = \pi$ or $\alpha_2 = \pi$.
        If $\alpha_1 = \pi$, let $D'$ be the support of $A_1C_3$. Now $A_kC_3$ is on the extension of the diagonal $DD'$ that connects two concave corners of $P$. If $\alpha_2 = \pi$, then $C_3$ is a convex corner of $P \setminus Q_2$ and we can modify the partition similar to case~\ref{case:ak_in_interior_of_p}.
        \begin{figure}
            \centering
            \includegraphics{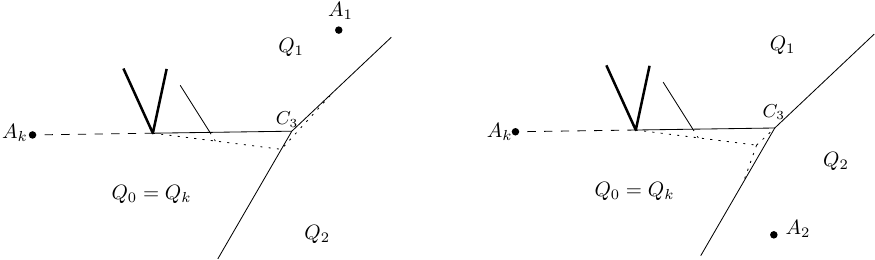}
            \caption{The case when an edge at $C_3$ is not covered by important sight line.}
            \label{fig:tripod-3}
        \end{figure}
    \end{enumerate}

    In the remainder, $C_3$ is a tripod point. Note that if the tripod associated with $C_3$ contains no star center in its pseudo-triangle, then the edge $C_1 C_2$ is a good edge. Hence, the only remaining case is the following.

    \begin{enumerate}
        \item[6.] \textbf{$C_3$ is a tripod point and a star center is in the interior of the pseudo-triangle of this tripod.}
        Let $A_3$ be any star center inside the pseudo-triangle. The tripod partitions $P$ into three regions $R_0,R_1,R_2$, where $R_i$ is the region containing $A_i$. 
        Let $R'_i$ be the intersection of $R_i$ and the pseudo-triangle. 
        First we prove that there exists a segment $XY$ from one leg of the tripod to another leg that contains a star center $A'_3$ and no star center is in the interior of triangle $C_3XY$. 
        Let $R'_j$ be the region that contains $A_3$. 
        Consider the convex hull $\mathcal{C}$ of all corners of this pseudo-triangle and all star centers in $R'_j$. 
        There exists a corner of $\mathcal{C}$ that is a star center $A'_3$, as $A_3$ lies in the interior of $R'_j$.
        Let $\ell$ be an arbitrary tangent of $\mathcal{C}$ at $A'_3$, and let $XY$ be the subsegment of $\ell$ that is contained in the pseudo-triangle. 
        Whichever region $A_3$ lies in, we can modify the partition so that the tripod is broken and $|\mathcal{S}|$ is not increased; see \Cref{fig:nostarcenters}.
        Since the number of possible tripods is finite, we can apply the argument a finite number of times and then either the size of $\mathcal{S}$ is decreased, or $A_k C_3$ becomes a good segment. \qedhere{}
        \begin{cfigure}
            \centering
            \includegraphics[page=20, width=\textwidth]{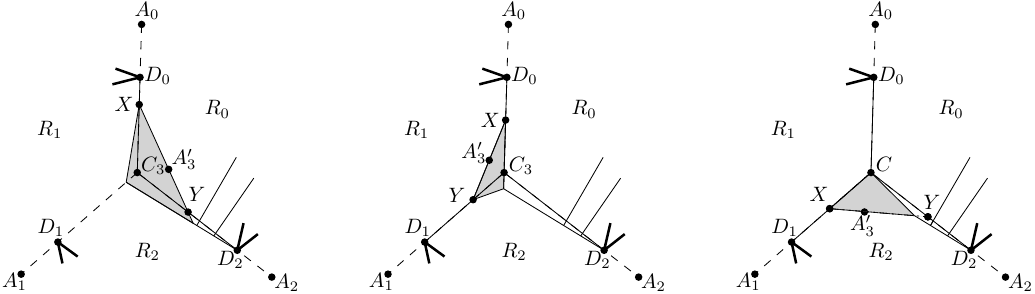}
            \caption{Three cases of where the star center $A'_3$ is and how we define the gray triangle to give to $A'_3$.}
            \label{fig:nostarcenters}
        \end{cfigure}
    \end{enumerate}
\end{proof}

\subsection{Tripod Trees} \label{sec:tripod-trees}
We now define what we call the \emph{tripod tree}---a description of the structure of tripods in an optimal solution (see also \cref{fig:tripod-tree}).
Given a star partition, consider the partition that is induced by the tripod legs.
Note that this partition is simply the star partition but with some pieces having been merged.
We construct a bipartite graph $G = (X \cup Y, E)$ as follows:
\begin{itemize}
\item We add a vertex to $X$ for each face of the partition induced by the legs.
\item We add a vertex to $Y$ for each tripod.
\item We add an edge $\{x,y\}$ to $E$ if and only if a tripod leg of $y$ forms part of the boundary of $x$.
\end{itemize}

\begin{observation}
\label{lem:tripod-tree}
Given a star partition, the tripod tree graph $G$ is indeed a tree.
\end{observation}

\begin{proof}
    If the tripods are considered degenerate pieces of the partition, then $G$ corresponds to the dual graph of the partition induced by the legs. Thus, $G$ is connected.
    Furthermore, note that every tripod cuts the polygon $P$ into three disconnected pieces, so the corresponding vertex in $Y$ is a cut vertex, which implies that $G$ is a tree. 
\end{proof}

We choose the root of the tripod tree to be the face that contains the first edge of $P$, merely for consistency. For every tripod $\mathcal{T}$ formed by pieces $Q_i, Q_j, Q_k$ where $Q_i$ is contained in the parent face of $\mathcal{T}$, we call the star center $A_i$ of $Q_i$ the \emph{parent star center} of $\mathcal{T}$ and the star centers $A_j, A_k$ of $Q_j, Q_k$ are both called \emph{child star centers} of $\mathcal{T}$. Note that we can directly identify the parent star center of a tripod without the full tripod tree. 

\begin{figure}
    \centering
        \includegraphics[width=0.9\linewidth, page=1]{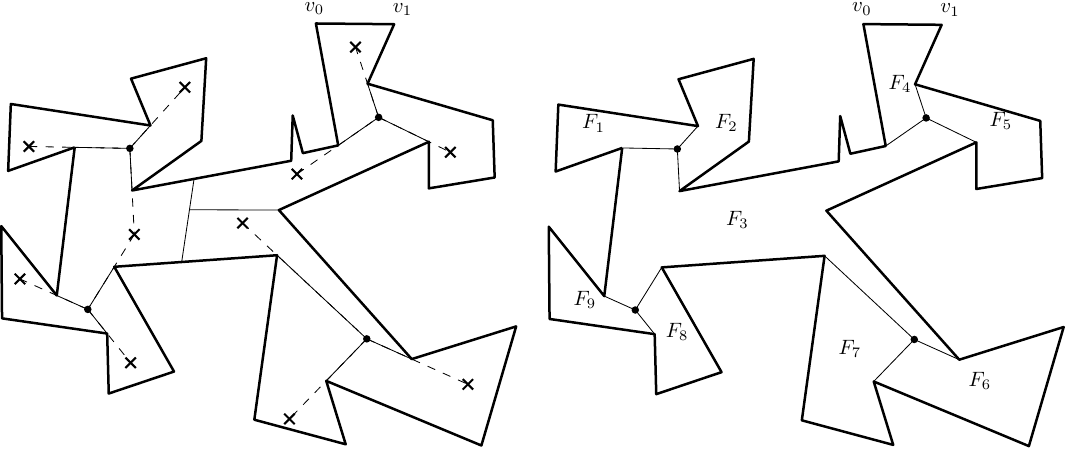}
        \caption{Illustration for a tripod tree. A cross represents a star center, a solid circle represents a tripod point. The left figure shows  star partition and the right figure shows the faces splits by all the tripod legs. }
      \end{figure}
      
       \begin{figure}
        \includegraphics[width=0.9\linewidth, page=6]{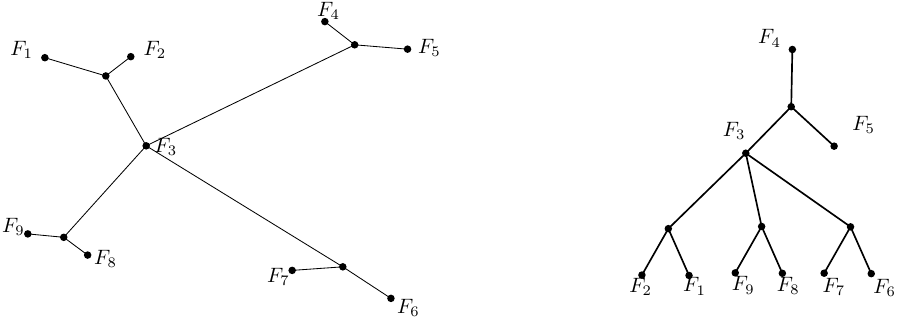}
        \caption{The combinatorial structure of tripod tree. We make it rooted by selecting a face to be the root, for example the one ($F_4$) containing edge $v_0v_1$. Note that we can easily distinguish parent and children without the full partition. }
    \label{fig:tripod-tree}
\end{figure}

\paragraph{Fake tripod.} In a star partition $\mathcal{Q}$, a \emph{fake tripod} $\mathcal{T}'$ with tripod point $C$ is defined by \emph{two} star centers $A_1, A_2$ of pieces $Q_1, Q_2$ and \emph{three} concave corners $D_1, D_2, D_3$ of $P$ if the following properties hold. 
\begin{itemize}
    \item $A_iC$ is an important sight line of $Q_i$ with support $D_i$, for each $i \in \{1, 2\}$. 
    \item $A_iC \subset \bigcup_{Q \in \mathcal{Q}}\partial Q$ for each $i \in \{1, 2\}$. 
    \item The three angles $D_1CD_2, D_2CD_3, D_3CD_1$ are strictly convex. 
    \item Let $F_1, F_2, F_3$ be the three connected components of $P$ cut by $D_1C \cup D_2C \cup D_3C$, where $A_1 \in F_1, A_2 \in F_2$, $F_3$ contains the first edge of $P$. 
    $D_3$ is a concave corner in $F_3$. 
    The union $Q_1 \cup Q_2 \cup F_3$ contains a (sufficiently small) disk at $C$. 
\end{itemize}
Similarly as for tripods, the three segments $D_1C, D_2C, D_3C$ are called the \emph{legs} of $\mathcal{T}'$, $A_1, A_2$ are called \emph{child star centers} of $\mathcal{T}'$, and the polygon bounded by the shortest paths between pairs of supports $D_1, D_2, D_3$ is called the \emph{pseudo-triangle} of $\mathcal{T}'$. 

Note that for every tripod $\mathcal{T}$, there is exactly one fake tripod $\mathcal{T}'$ with the same tripod point and legs, which is defined by its supports and the two child star centers of $\mathcal{T}$.
We say $\mathcal{T}'$ is the \emph{associated fake tripod} of $\mathcal{T}$. 
\begin{cfigure}
    \centering
    \includegraphics[page=35]{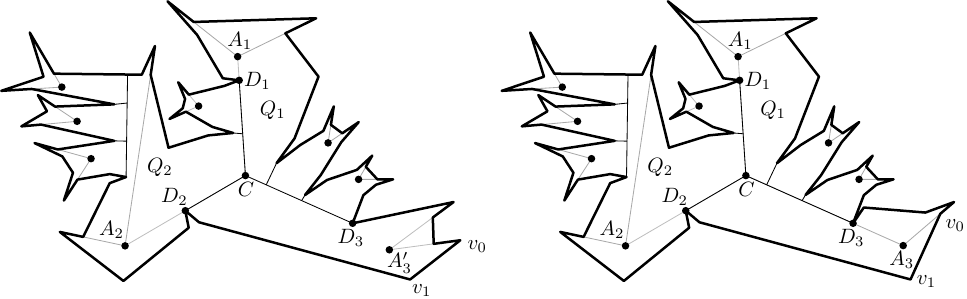}
    \caption{\emph{Left:} a fake tripod defined by star-centers $A_1, A_2$ and supports $D_1,D_2,D_3$. Note that $A_3'$ is not on the extension of $CD_3$. \emph{Right:} a tripod and its associated fake tripod, since there exists the third star-center $A_3$ on the extension of $CD_3$. The fake tripod is the same as the one in the left figure. } \label{fig:fake-tripod}
\end{cfigure}

\begin{remark} \label{rem:fake-tripod}
    We introduce the concept of \emph{fake} tripods (see also \cref{fig:fake-tripod}) to facilitate our algorithm. The final algorithm will simulate the construction process (described in the following \cref{sec:constproc}) and get a full partition in the end. We can not easily decide whether a (real) tripod exists unless we find both its two child star centers \emph{and} its parent star center. Algorithmically, it is much easier to construct the \emph{fake} tripods in a bottom-up fashion just using the two child star centers, without knowledge of where (or if) a potential third parent star-center might exist. 
\end{remark}

\subsection{Construction Process}\label{sec:constproc}
We now describe an iterative construction process of star centers and fake tripods. 
We will show that there exists an optimal star partition for which all star centers can be constructed using this process in linear steps. 
The construction process is with respect to a star partition $\mathcal{Q}$ and is a process to ``mark'' star centers and fake tripods of $\mathcal{Q}$ as ``constructable''. 
Formally, we call a star center or a fake tripod \emph{constructable} (with respect to $\mathcal{Q}$) if it can be \emph{marked} by the following process. 
At each step in the process, we can do one of the following operations:
\begin{itemize}
    \item Mark a star center $A_k$ at the intersection of two non-parallel segments of the following types:
    \begin{enumerate}
        \item An edge of $P$; 
        \item An edge of the pseudo-triangle of a marked fake tripod; 
        \item An important sight line of a piece $Q_k$, which is on the extension of 
        \begin{itemize}
            \item an edge of $P$; 
            \item a diagonal of $P$ that connects two concave corners of $P$; 
            \item a tripod leg of a tripod $\mathcal{T}$ whose extension contains the parent star center of $\mathcal{T}$, while the corresponding fake tripod $\mathcal{T}'$ of $\mathcal{T}$ is marked. 
        \end{itemize}
    \end{enumerate}
    \item Mark a fake tripod $\mathcal{T}'$ defined by two marked star centers $A_i, A_j$ and three concave corners $D_i, D_j, D_k$ of $P$. 
    Additionally, there must be no star center (marked or unmarked) in the interior of the pseudo-triangle of $\mathcal{T}'$. 
\end{itemize}
An optimal star partition $\mathcal{Q}$ is called \emph{constructable} if all the star centers in $\mathcal{Q}$ is constructable with respect to $\mathcal{Q}$. 

Now comes the major structural result in this section, which gives us a combinatorial way to describe some optimal star partition. 

\begin{theorem}[Construction of optimal star partition]
    \label{thm:constructability}
    There exists a constructable optimal star partition $\mathcal{Q}$. 
\end{theorem}

\begin{remark}
    We can also define a similar construction process if we can only mark tripod. 
    In fact, the two definitions agree on whether a partition is constructable or not. 
    When a fake tripod is used to mark a star center, it must be a tripod; otherwise, there is no need to mark that fake tripod. 
\end{remark}

This theorem also implies that the bit complexity of each star center is $O(n)$. 
\begin{corollary}
    \label{cor:bit-complexity}
    Each star center in a constructable optimal star partition can be encoded by a sequence of $O(n)$ corners of $P$, which specifies the process to mark it. 
\end{corollary}
\begin{remark}
    Using the same proof strategy, we can also prove $O(K)$ bits are enough to encode each star center in a constructable optimal star partition, where $K$ is the total number of bits to encode the input polygon $P$. 
\end{remark}
\begin{proof}
    We will prove that we can encode each star center $A_i$ by $4s(A_i)$ corners, where $s(A_i)$ is the size of the subtree in the fake tripod tree rooted at the face containing $A_i$, and encode each fake tripod point $C$ by $4s(C)$ corners, where $s(C)$ is the size of the subtree in the fake tripod tree rooted at $C$. 

    The proof is by induction on the fake tripod tree from leaf node to root node. 
    For each leaf node of the fake tripod tree, every star center in the corresponding face can only be marked by two lines that each of them are defined by two corners of $P$, since the face does not have any fake tripod point as its child. 
    Now consider the internal nodes of the fake tripod tree. 
    If it corresponds to a fake tripod $\mathcal{T}$, the tripod point $C$ can be encoded by its two child star centers $A_i, A_j$ together with two concave corners of $P$, so we need $s(A_i) + s(A_j) + 2 = 4s(C) - 2 \leq 4s(C)$ corners to encode $C$. 
    If it corresponds to a face $F$, then for any star center $A_k$ in $F$, it can be encoded by two lines, each of them is either defined by two corners of $P$ or defined by a child fake tripod $\mathcal{T}$ of $F$. 
    In all cases, the star center can be encoded by $4s(A_k)$ corners of $P$. 

    It remains to bound the size of the fake tripod tree. 
    According to Chvátal's art gallery theorem, we can partition any polygon into at most $\lfloor n/3\rfloor$ star-shape pieces. 
    Any leaf face in the fake tripod tree contains at least one star center, so the fake tripod tree contains at most $\lfloor n/3\rfloor$ leaves, therefore it has at most $2n/3$ nodes. 
\end{proof}

Instead of proving \cref{thm:constructability} directly, we prove the following stronger lemma, which allows us to extend a ``partially constructable'' optimal solution into a constructable one. This lemma also helps us to prove the correctness of the greedy choice (see \Cref{sec:greedy-choice}) used when choosing tripods, which is the main technique to improve the running time of our dynamic programming algorithm into polynomial time in \cref{sec:algorithm}. 

\begin{lemma}
    \label{lem:construct-from-middle}
    Let $\mathcal{Q}$ be an optimal star partition of $P$ such that some star centers $A_1, A_2, \dots, A_k$ and some fake tripods $\mathcal{T}'_1, \mathcal{T}'_2, \dots, \mathcal{T}'_l$ are constructable with respect to $\mathcal{Q}$. 
    Suppose there exists a star center in $\mathcal{Q}$ which is not constructable with respect to $\mathcal{Q}$, then there exists an optimal solution $\mathcal{Q}'$ containing $A_1, A_2, \dots, A_k, A_{k+1}$ as star centers and $\mathcal{T}'_1, \mathcal{T}'_2, \dots, \mathcal{T}'_l$ as fake tripods, such that $A_1, A_2, \dots, A_k, A_{k+1}$ and $\mathcal{T}'_1, \mathcal{T}'_2, \dots, \mathcal{T}'_l$ are constructable. 
\end{lemma}

\begin{proof}
A star center is called missing if it's not constructable with respect to $\mathcal{Q}$. 
Consider the construction process of all the constructable star centers and fake tripods. 
Throughout the construction process we maintain what we call the \emph{feasible region}, which initially consists of the whole polygon $P$.
When a fake tripod with center $C$ is marked by two marked star centers $A_1, A_2$ and three concave corners $D_1, D_2, D_3$, we partition the polygon into three parts (according to the legs of the fake tripod) and remove the pseudo-triangle of this fake tripod from the feasible region. 
Moreover, we add two segments $A_1D_1, A_2D_2$ as an incision into the boundary of both the polygon\footnote{Now the polygon becomes only weakly simple, even if it was a initially a simple polygon (see \cref{sec:prelim}).} and the feasible region. 
See \cref{fig:feasible-region-and-splitted-polygon} for illustration. 
This way, we fix the two important sight lines $A_1C$ and $A_2C$, and make sure that no star centers are in the interior of the pseudo-triangle of this fake tripod in the following steps. 
When a star center $A_i$ is marked by an important sight line $\ell$ of $Q_i$, we also add $\ell$ as an incision into the boundary of both the polygon and the feasible region. 

\begin{figure}
    \centering
    \includegraphics[page=36]{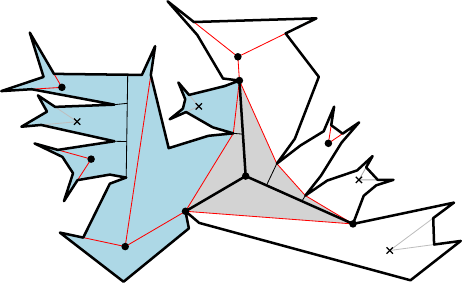}
    \caption{The figure shows the feasible region and the partition of the polygon where some star centers and fake tripods are marked. 
    Thick black lines represent how we partition the polygon into three parts. 
    Red segments mark the boundary of feasible region, and some of them are merely incisions. 
    Blue region is a weakly connected components of the feasible region in the left part, with some red segments as incisions. 
    The gray region marks the pseudo-triangle of the marked fake tripod, which is not a part of the feasible region, but the entire feasible region covers everything else.
    Each black circle represents a marked star center or a marked fake tripod, and each cross represents an unmarked star center. }
    \label{fig:feasible-region-and-splitted-polygon}
\end{figure}

Since the constructable fake tripods partition the polygon into disconnected pieces, we can consider the construction process within each connected part independently. Let $P'$ be such a connected part for which at least one star center is not constructable. If there are multiple choices, we will choose one later. Let $F$ be the feasible region inside $P'$. We now restrict our polygon to be~$P'$. 

We perform a restricted coordinate maximization along an arbitrary direction $d$ as described in \cref{lem:coordmaxsubset} within $F$ and with all the constructable star centers fixed. 
By optimality, for each missing star center, the partition is also restricted coordinate-maximal along $d$, within $F$ and with all the other star centers fixed.
Applying \cref{lem:property-of-coordmax}, we know that all the missing star centers lie on the intersection of two non-parallel segments of certain types. 
We call these segments the \emph{crucial segments}.
The goal of the following discussion is to mark
a new star center using crucial segments. 
Therefore, we enumerate the types of crucial segments and check whether they are allowed in the construction process.

\textbf{\boldmath Case 1: A crucial segment is an edge of $F$}.
Every edge of $F$ is an edge of $P$, or a segment on the boundary of a pseudo-triangle, or an important sight line that was used to construct a constructable star center.
Note that all segments on the boundary of a pseudo-triangle must be diagonals that connect two concave corners of $P$. Hence, all the edges of $F$ can be used to mark new star centers. 

\textbf{\boldmath Case 2: A crucial segment is an important sight line in $P'$}.
Note that all corners of $P'$ are either corners of $P$, or constructable star centers, or tripod points of a constructable fake tripod. 
Each tripod point is a convex corner in any connected part separated by the tripod legs, hence they can only induce convex corners of $P'$.
Consequently, a concave corner of $P'$ is either a concave corner of $P$ or a constructable star center.
According to \cref{lem:property-of-coordmax}, no other star center lies on a crucial segment, so the crucial segment that the star center lies on must be supported by a concave corner of $P$, which implies that the crucial segment is also an important sight line when we consider the full polygon $P$. 
If the crucial segment ends at a concave corner of $P'$, it must end at a concave corner of $P$, as crucial segments are not allowed to contain another star center. Hence, it is contained in an extension of a diagonal of $P$ that connects two concave corners of $P$. 

Thus, a crucial segment cannot be used to mark star centers only if it is an extension of a tripod leg, and either the corresponding fake tripod is not constructable, or the star center is not in the parent face of this tripod. 

We now consider the tripod tree of $\mathcal{Q}$. 
We call a tripod \emph{illegal} if there exists a missing star center in its child face. 
Otherwise, we call a tripod \emph{legal}. 
Then, a star center is missing only if one of its crucial segments end at the tripod point of an illegal tripod. 
As there exists a missing star center, there also exists an illegal tripod. 

Let $\mathcal{T}$ be an illegal tripod such that all tripods contained in the subtree rooted at $\mathcal{T}$ are legal. 
Note that there exists a missing star center in at least one of its child faces. 
Hence, all star centers in the two child faces of $\mathcal{T}$, except for the child star centers of $\mathcal{T}$, are constructable. 
Let $A$ be a missing child star center of $\mathcal{T}$. 
We choose $P'$ to be the connected component containing $A$. 

We perform a restricted coordinate maximization on $A$, where the polygon we are going to partition is the part of $P'$ cut by $\mathcal{T}$ that $A$ lies in, and the feasible region is $F$ excluding the pseudo-triangle of $\mathcal{T}$ restricted to the current polygon. 
We choose $d$ to be the direction perpendicular to the important sight line that $A$ ends at the tripod point of $\mathcal{T}$ and points to the unsupported side. 
The choice of direction $d$ makes it impossible that the new star center $A'$ lies on the same important sight lines from $\mathcal{T}$, unless it is also on other two non-parallel crucial segments.

By a similar analysis of the types of crucial segments that the new star center $A'$ lies on, $A'$ cannot be marked in the next step only if it is a child star center of an illegal tripod contained in $P'$, and we resolve this case recursively; see \cref{fig:construction}. Since the subpolygon cut by the new illegal tripod has strictly fewer corners, this recursion will finish in finite steps. Eventually, we can find an optimal partition in which one more star center $A'$ is constructable by the intersection of two non-parallel crucial segments. 
\end{proof}

\begin{figure}[!htb]
    \centering
    \includegraphics[page=2,width=.95\linewidth]{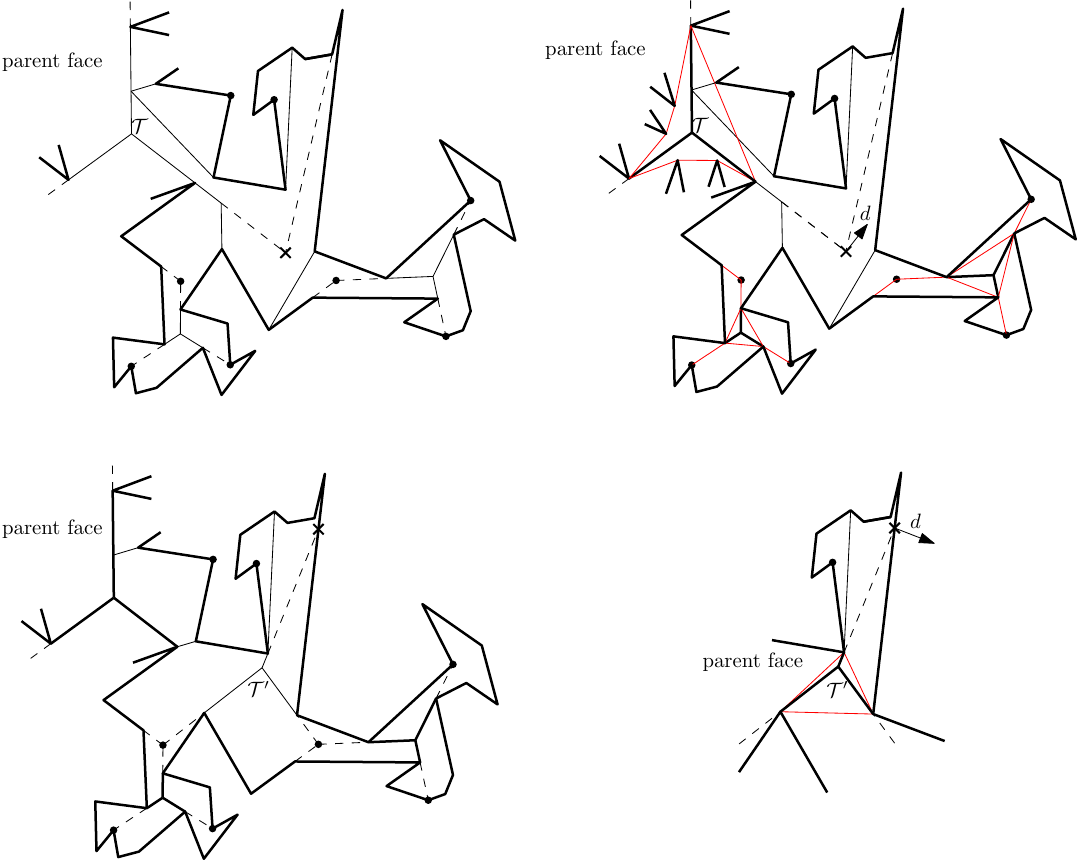}
    \caption{Illustration for the proof of \cref{thm:constructability}. A constructable star center is represented by a solid circle, and the missing star center is represented by a cross. 
    Red segments mark the boundary of the feasible region, and thick black segments mark how we partition the polygon into parts. 
    The top-left figure shows the partition $\mathcal{Q}$, and the top-right figure shows the restricted coordinate maximization problem used in the proof. 
    The bottom-left figure shows the restricted coordinate maximization partition of this problem. 
    The old illegal tripod $\mathcal{T}$ is no longer a tripod, but we get a new illegal tripod $\mathcal{T}'$, which will be resolved recursively. 
    The bottom-right figure shows this new subproblem that we get. 
    Since the number of vertices on the boundary are strictly fewer, the recursion will eventually terminate. 
    }
    \label{fig:construction}
\end{figure}

The following lemma gives us another useful property of constructable optimal star partition, which will be used to design our dynamic programming algorithm in \cref{sec:algorithm}. 

\begin{lemma}
\label{lem:outer-planar}
    Let $\mathcal{Q}$ be a star partition of $P$. If a star center $A_k$ is constructable with respect to $\mathcal{Q}$, then a corner of $P$ appears on the boundary of the piece $Q_k$.
\end{lemma}

\begin{proof}
    We consider the different cases in the construction process which can lead to the star center $A_k$ being marked.
    If $A_k$ is marked by an important sight line $\ell$ of $Q_k$, then the support $D'$ of $\ell$ is a corner of $P$ in $Q_k$, and the lemma holds. 
    If $A_k$ is at the intersection of two edges of $P$, $A_k$ must be a corner of $P$. 
    If $A_k$ is at the intersection of an edge of $P$ and an edge of a pseudo-triangle, $A_k$ must be a corner of $P$, as the intersection of the pseudo-triangle and the boundary of $P$ is just a set of concave corners of $P$. 
    If $A_k$ is at the intersection of two non-parallel edges from the same pseudo-triangle, $A_k$ must be a corner of this pseudo-triangle, and it therefore is also a concave corner of $P$. 
    If $A_k$ is at the intersection of two non-parallel edges from different pseudo-triangles, according to \cref{lem:separate-tripods}, $A_k$ must be a corner of one pseudo-triangle, therefore $A_k$ is also a corner of $P$. 
\end{proof}

From this lemma, we directly have the following corollary, which will be used to prove some combinatorial properties of an optimal partition in \cref{sec:algorithm}. 

\begin{corollary}
    \label{cor:outer-planarity}
    For any constructable optimal star partition $\mathcal{Q} = Q_1, \dots, Q_k$, we have $Q_i \cap \partial P \neq \emptyset$, that is, every star-shaped piece touches the boundary of $P$. 
\end{corollary}

\section{Properties of Area Maximum Partitions}
\label{sec:area-max}

The objective of this section is to compute a set of polynomially many points that contains all the Steiner points for some optimal solutions, given all the star centers of an arbitrary optimal solution. 
We will work on a constructable optimal star partition, with all the construction lines fixed as incisions, and analyze the position of the Steiner points in each connected components independently, as we did in the proof of \cref{thm:constructability}. 
Recall \cref{lem:outer-planar}, all the star-shaped pieces still touch the outer boundary of the input polygon $P$. 

Consider a weakly simple polygon $P'$ with potentially some incisions, a sequence of points $\mathcal A=(A_1,\ldots,A_k)$, and a star partition $\mathcal Q=(Q_1,\ldots,Q_k)$ of the interior of $P'$ where $A_j$ is a star center of $Q_j$.
Recall that $\mathcal Q$ is called \emph{area maximum} with respect to $\mathcal A$ if the vector of areas $a(\mathcal Q)=\langle a(Q_1),\ldots,a(Q_k)\rangle$ is maximum in lexicographic order among all partitions of $P$ with star centers $\mathcal A$.
In \Cref{apx:existence} we argue that this notion of area maximum partition is well-defined.
Note that in the definition of area maximum partitions, the positions of the star centers are fixed.
In this section, we prove some properties of area maximum partitions, in particular that all corners of pieces (i.e.\ Steiner points) are in ``nice'' spots.

Recall that a \emph{sight line} of a piece $Q_i$ is a segment of the form $r=A_iC$, where $A_i$ is the star center of $Q_i$ and $C$ is a corner of $Q_i$.
Here, the point $C$ is called the \emph{end} of $r$.

\begin{lemma}\label{lem:newboundaryinsightlines}
Consider an optimal area maximum partition $\mathcal Q$ with a given set of star centers $\mathcal A$.
For any two pieces $Q_1,Q_2\in\mathcal Q$, let $\gamma \coloneqq \partial Q_1\cap\partial Q_2$ be their shared boundary.
Then $\gamma$ is
\begin{itemize}
\item
empty, or

\item
a single point, or

\item
a single line segment contained in a sight line of $Q_1$ or $Q_2$, or

\item
two adjacent line segments, each contained in a sight line of $Q_1$ or $Q_2$.
\end{itemize}
\end{lemma}

\begin{figure}
\centering
\includegraphics[page=26]{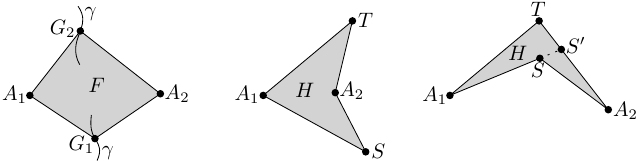}
\caption{
Situations in the proof of \Cref{lem:newboundaryinsightlines}.
\emph{Left:}
If $\gamma$ is not connected, $F$ must contain a third piece which can be subsumed by $Q_1$ and $Q_2$, which is a contradiction.
\emph{Middle:}
When $S$ and $T$ are convex, we have $\gamma=SA_2\cup A_2T$.
\emph{Right:}
When $S$ is concave, we have $\gamma=SS'\cup S'T$.
}
\label{fig:boundaryinsightlines}
\end{figure}

\begin{proof}
We first prove that $\gamma$ is connected.
Otherwise, there exist points $G_1,G_2\in\gamma$ which are not connected by $\gamma$; see \Cref{fig:boundaryinsightlines} (left).
Let $A_i\in\mathcal A$ be the star center of $Q_i$.
Since $G_1$ and $G_2$ are not connected by $\gamma$, the quadrilateral $F=A_1G_1A_2G_2$ is not contained in $Q_1\cup Q_2$, but the boundary $\partial F$ is contained.
Hence, $F$ contains a third piece from $\mathcal Q$.
However, the quadrilateral $F$ can be completely assigned to $Q_1$ and $Q_2$ (possibly after being split in two), so the partition $\mathcal Q$ was not optimal, which is a contradiction.
Hence, $\gamma$ is connected.

Let $S$ and $T$ be the endpoints of $\gamma$.
If $S=T$, the shared boundary is a single point which must be a corner of at least one of the two pieces $Q_1,Q_2$, since otherwise their shared boundary would be longer.

Otherwise, let us traverse $\gamma$ from $S$ to $T$; see \Cref{fig:boundaryinsightlines} (middle and right).
Since $Q_1$ and $Q_2$ are star-shaped, we move around $A_1$ in a monotone way, either clockwise or counterclockwise, and we move in the opposite direction around $A_2$.
Hence, $\gamma$ is contained in the quadrilateral $H=A_1SA_2T$.
Assume without loss of generality that $A_1$ had higher priority than $A_2$ when we maximized the areas of the pieces in $\mathcal Q$.

If $S$ and $T$ are both convex corners of $H$, then all of $H$ can be seen from $A_1$, so we can assign the quadrilateral $H$ to $Q_1$. Then $\gamma$ is a continuous part of $SA_2\cup A_2T$ with $A_2 \in \gamma$, so $\gamma$ is contained in two sight lines of $Q_2$.

Otherwise, assume without loss of generality that $S$ is concave and $T$ is convex.
Let $S'$ be the intersection between the line containing the segment $A_1S$ and the segment $A_2T$.
Then we maximize the area of $Q_1$ by assigning the triangle $A_1S'T$ to $Q_1$ and the rest of $H$ (which is the triangle $A_2SS'$) to $Q_2$.
Hence, we have that $\gamma$ is a continuous part of $SS'\cup S'T$ with $S' \in \gamma$, so $\gamma$ is contained in sight lines of $Q_1$ and $Q_2$, respectively.
Note that it may happen that $S'=T$, so that one of these segments is degenerate.
\end{proof}

We say that a point $B$ is \emph{supporting} a sight line $r=AC$ from a star center $A$ if $B\in r\setminus\{A\}$.
The following lemma characterizes all edges of pieces of an area maximum partition, using \Cref{lem:newboundaryinsightlines}.

\begin{lemma}\label{lem:sltypes}
Consider a piece $Q$ with star center $A$ of an optimal area maximum partition.
Let $r=AC$ be a sight line of $Q$ which contains an edge on the boundary of $Q$.
Then $r$ is of one of the following types:
\begin{enumerate}[label=(\roman*)]
\item \label{sl:1}
$r$ is supported by a corner of $P'$, which is not a star center, or

\item \label{sl:2}
$r$ is supported by a star center of another piece, or

\item \label{sl:3}
$C$ is the end of two non-parallel sight lines of type \ref{sl:1} in other pieces.
\end{enumerate}
\end{lemma}

\begin{proof}
Suppose that a sight line $r=AC$ is supported by no corner of $P'$ and no star center of another piece.
We will prove that the end $C$ must be the end of two sight lines of type \ref{sl:1} of other pieces.
We consider multiple cases.

\begin{figure}
\centering
\includegraphics[page=28]{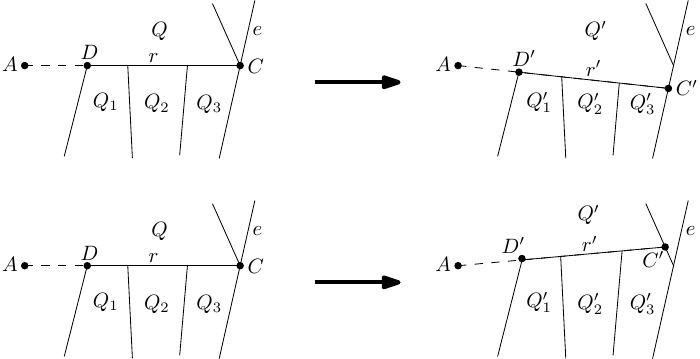}
\caption{Situations in the proof of \Cref{lem:sltypes}.
In Case 1.1 or Case 1.2, we can increase the area of $Q$ or $Q_1,Q_2,Q_3$ by rotating $r$ clockwise or counterclockwise, respectively.}
\label{fig:lem:sltypes:rotate}
\end{figure}

\textbf{\boldmath Case 0: $C$ is a corner of $P'$.} In this case, $r$ is supported by $C$ and of type \ref{sl:1} or \ref{sl:2}. 

\textbf{\boldmath Case 1: $C$ is an interior point of another sight line or an edge of $P'$.}
Let us denote this sight line or edge by $e$; see~\Cref{fig:lem:sltypes:rotate}.
Let $f=DC$ be the edge of $Q$ contained in $r$.
Assume without loss of generality that $r$ is horizontal with the end $C$ to the right and that the interior of $Q$ is above $f$.
Then some other pieces $Q_1,\ldots,Q_i$ are below $f$.
Recall that in an area maximum partition, we maximize the vector of areas in a specific lexicographic order; in other words, each piece has a distinct priority when maximizing the areas. In both of the following cases we obtain a contradiction.

\textbf{\boldmath Case 1.1: $Q$ has higher priority than all of $Q_1,\ldots,Q_i$.}
In this case, we can expand $Q$ a bit by rotating $r$ a bit clockwise around $A$, thus ``stealing'' some area from the pieces $Q_1,\ldots,Q_i$ and increasing the area vector with respect to the lexicographic order.

\textbf{\boldmath Case 1.2: One of the pieces $Q_1,\ldots,Q_i$ has higher priority than $Q$.}
In this case, we can rotate $r$ a bit counterclockwise, thus expanding the pieces $Q_1,\ldots,Q_i$ and increasing the area vector.
We conclude that $C$ is not an interior point of another sight line or an edge of $P$.

Note that if $r$ does not fall into case 0 or case 1, $C$ must be in the interior of $P'$.

\begin{figure}[htbp]
    \centering
    \includegraphics[page=37]{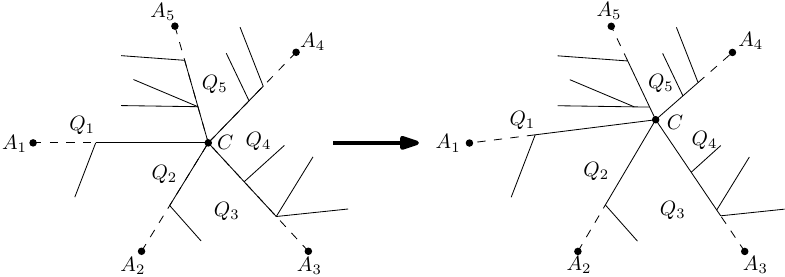}
    \caption{Case 2.1: Illustration of moving point $C$ when none of $r_i$ is supported by a corner or a star center. While moving $C$, we fix the lines containing the boundary segments that touch $r_1, \dots, r_j$, and slide the intersection points accordingly.
    If more than one boundary touches $r_i$ at the same point, we might only extend the one closest to $C$. Since all the new corners we create are convex, all pieces remain star shaped. We only need to move $C$ infinitesimally to obtain a contradiction, so no new crossings will be formed during this process.}
    \label{fig:case-2.11}
\end{figure}

\begin{cfigure}[htbp]
    \centering
    \includegraphics[page=32]{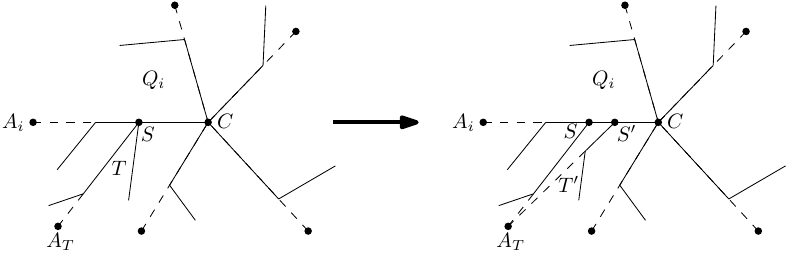}
    \caption{Case 2.1: If $T$ is the highest priority piece and $T \cup r_i = S$ is a single point, then one side of visibility from its star center $A_T$ to $S$ is not blocked, therefore, $T$ take a sufficiently small triangle near $S$. }
    \label{fig:case-2.12}
\end{cfigure}

\textbf{\boldmath Case 2: $C$ is the end of one or more sight lines of other pieces.}
First, note that according to \Cref{lem:newboundaryinsightlines}, if $C$ is on the boundary of a piece, then $C$ is also contained in a sight line of that piece.
Let the sight lines that share the end $C$ be $r_1,\ldots,r_j$ in counterclockwise order (one of these is $r$), and let the associated set of pieces and star centers be $\mathcal Q=\{Q_1,\ldots,Q_j\}$ and $A_1,\ldots,A_j$, respectively.
We first observe that we must have $j\geq 3$:
Clearly $j\geq 2$, so consider the case $j=2$.
If the two sight lines $r_1$ and $r_2$ are not parallel, then $C$ is a concave corner of one of them that causes the piece to not be star-shaped.
If the two sight lines are parallel, then $C$ is not a corner of the pieces, so $r_1$ and $r_2$ are not (complete) sight lines of the pieces.
Hence, $j\geq 3$.

Assume without loss of generality that the interior of each $Q_i$ is to the left of $r_i$.
This has the consequence that if $r_i$ is supported by a corner of $P$, then the two incident edges are to the right of $r_i$ and likewise, if $r_i$ is supported by a star center, then the interior of the associated piece is also to the right.
Suppose towards a contradiction that at most one of the sight lines $r_1,\ldots,r_j$ is supported by a corner of $P$ that is not also a star center---note that otherwise $r$ is a sight line of type \ref{sl:3}.
Let $\mathcal R$ be the set of pieces $R$ for which $R \notin \mathcal Q$ but $\partial R \cap r_i \neq \emptyset$ for some $i \in \{1, \dots, j\}$.
The goal is to improve the priority of the area vector by exchanging area between pieces in $\mathcal{Q} \cup \mathcal{R}$, which leads to a contradiction. 

\textbf{\boldmath Case 2.1: None of $r_1, \dots, r_j$ is supported by a corner or a star center. }
We will show that it is always possible to move $C$ anywhere within a sufficiently small disk; see~\cref{fig:case-2.11}. 
We attach all sight lines $r_i = A_iC$ to the flexible point $C$.
For each boundary segment $s$ touching one of $r_i$, we fix it on the same straight line, extending or contracting with respect to the movement of $C$, so the intersection point of $s$ and $r_i$ is flexible. 
If there are multiple segments touching $r_i$ at the same point, we will extend only the one closest to $C$ if necessary. 
Since $C$ is a convex corner in all of $Q_1, \dots, Q_j$, and all new corners are formed by the intersection of some ray and some straight line, all pieces remain in star shape with respect to their initial star centers. 

Now we will show that we can always improve the priority of the area vector. 
\begin{itemize}
    \item If $Q_i$ has the highest priority among $\mathcal{Q} \cup \mathcal{R}$, we can slide $C$ along the ray $A_iC$ and expand $Q_i$. 
    \item Otherwise, a piece $T \in \mathcal{R}$ has the highest priority among $\mathcal{Q} \cup \mathcal{R}$.
    Without loss of generality, assume $T$ is touching and to the right of $r_i$. If $T$ only touches $r_i$ at a single point $S$, then we can expand $T$ around $S$. See~\cref{fig:case-2.12}. If $T$ has a boundary segment along $r_i$, we can move $C$ to the left of the original ray $A_iC$ and expand $T$. 
\end{itemize}

\begin{figure}[htbp]
    \centering
    \includegraphics[page=42]{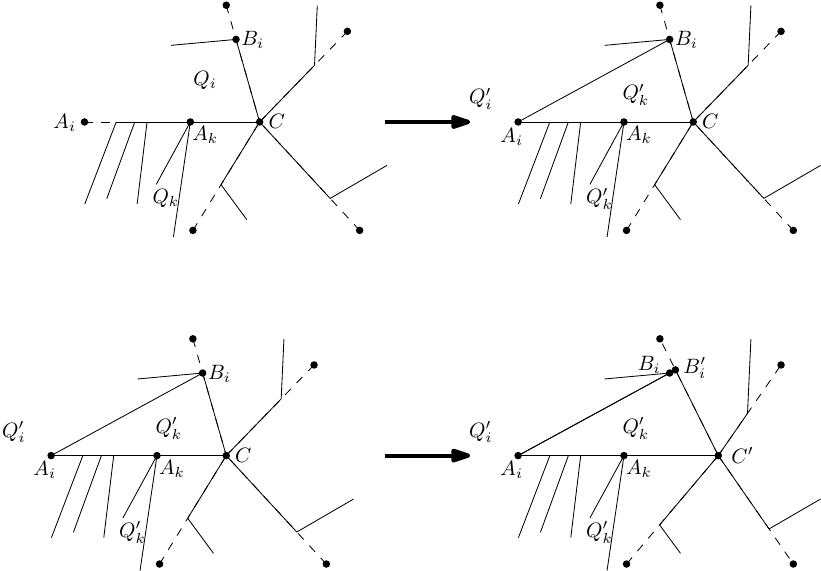}
    \caption{The top two figures illustrate the reduction from case 2.2 to case 2.1.
    By redistributing a triangle between pieces, the role of $r_i = A_iC$ is replaced by $r'_i = A_kC$, which is no longer a sight line of type~\ref{sl:2}. 
    The bottom two figures show how to expand $Q_i$ if $Q_i$ has the highest priority. 
    We can slide $C$ along the ray $A_iC$ to $C'$ and give the triangle $A_iC'B'_i$ back to $Q_i$ afterward. }
    \label{fig:case-2.21}
\end{figure}

\textbf{\boldmath Case 2.2: None of $r_1, \dots, r_j$ is of type~\ref{sl:1}, but can be of type~\ref{sl:2}. }
Compared with case 2.1, each sight line $r_i$ might be supported by star centers. 
We will show that it's always possible to move $C$ infinitesimally so as to increase the area vector.
See~\cref{fig:case-2.21} top. 
We will make some local modifications along each sight line of type~\ref{sl:2} and reduce to the previous case. 
Consider any sight line $r_i$ of type~\ref{sl:2}. 
Let $A_k$ be the farthest star center from $A_i$ that supports $r_i$, $B_iC$ be an edge of $Q_i$ along $r_{i-1}$. 
We will give the triangle $A_iCB_i$ from $Q_i$ to $Q_k$. 
After this exchange, the role of $r_i$ is replaced by a new sight line $r'_i = A_k C$, which is not supported by any star centers. 
After applying this modification along all the sight lines of type~\ref{sl:2}, we reduce to case 2.1, therefore, we can move $C$ anywhere within a sufficiently small disk. 

Now we will show that we can always improve the priority of the area vector in the end. 
Without loss of generality, assume $r_i$ is the initial sight line that touches the highest priority piece, and this piece is either $Q_i$ or a piece $T \in \mathcal{R}$ to the right of $r_i$. 
If $r_i$ is not supported by any star centers, we can make the same modification as in case 2.1. 
Therefore, we will only consider the case that $r_i$ is of type \ref{sl:2}. 
Similarly, let $A_k$ be the farthest star center from $A_i$ that supports $r_i$. 

\begin{figure}[htbp]
    \centering
    \includegraphics[page=38]{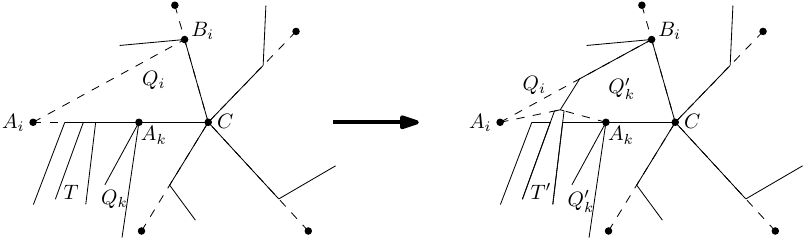}
    \caption{Case 2.2: When $T$ has the highest priority, and $T$ touches $A_iA_k$, we can give an sufficiently small pentagon to $T$. }
    \label{fig:case-2.23}
\end{figure}

\begin{itemize}
    \item If $Q_i$ has the highest priority, we can slide $C$ along the ray $A_iC$. Let $B_iC$ be the edge of $Q_i$ along $r_{i+1}$. After sliding $C$ along the ray $A_iC$ to $C'$, $B_i$ slides to $B'_i$, we can give the triangle $A_iC'B'_i$ back to $Q_i$ and improve the priority of the area vector.
    See~\cref{fig:case-2.21} bottom. 
    \item If some piece $T \in \mathcal{R}$ touching $r_i$ has the highest priority. If it touches $A_k C$, then after transferring the triangle $A_iCB_i$ from $Q_i$ to $Q_k$, it reduces to case 2.1. Therefore, we only need to consider the case when $T$ touches $A_iA_k$. 
    If $\partial T \cap A_iA_k$ is a single point, we can apply the same modification as case 2.1. See~\cref{fig:case-2.12}. 
    If $\partial T \cap A_iA_k$ is a segment, we can extend the two boundary segments end in $A_iA_k$ into the triangle $A_iCB_i$, give a sufficiently small pentagon $T$, and repartition the triangle $A_iCB_i$ accordingly. See~\cref{fig:case-2.23}. 
\end{itemize}

\begin{figure}[htbp]
    \centering
    \includegraphics[page=39]{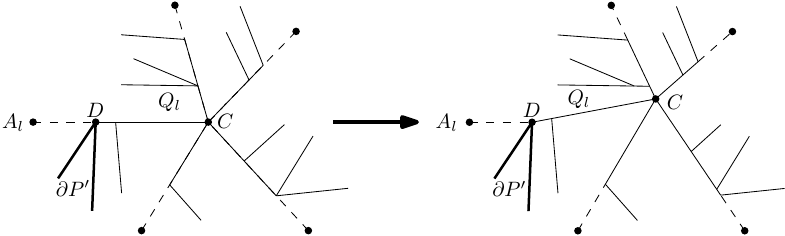}
    \caption{Case 2.3: When $r_l = A_l C$ is supported by a non-center-corner $D$, we treat the segment $DC$ as $r_l$, and make the all the intersection points with $DC$ flexible, as in case 2.1. If there are star centers along $DC$ or any other $r_i$, we treat it in the same way as in case 2.2. This modification keeps all the pieces star shaped as long as $C$ is on or to the left of ray $A_l D$ and within a sufficiently small disk. }
    \label{fig:case-2.31}
\end{figure}

\begin{cfigure}[htbp]
    \centering
    \includegraphics[page=40]{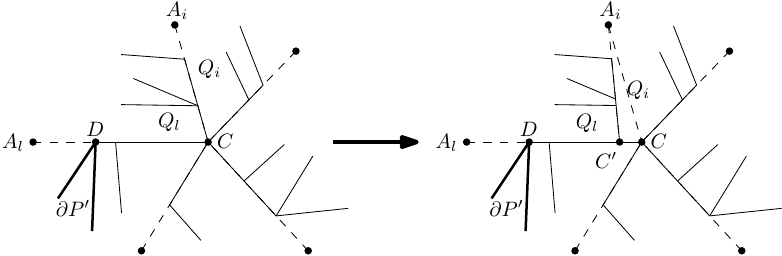}
    \caption{Case 2.3: If there is no star centers along $r_i = A_i C$, we can give the triangle $A_iCC'$ to $Q_i$. }
    \label{fig:case-2.32}
\end{cfigure}

\begin{cfigure}[htbp]
    \centering
    \includegraphics[page=41]{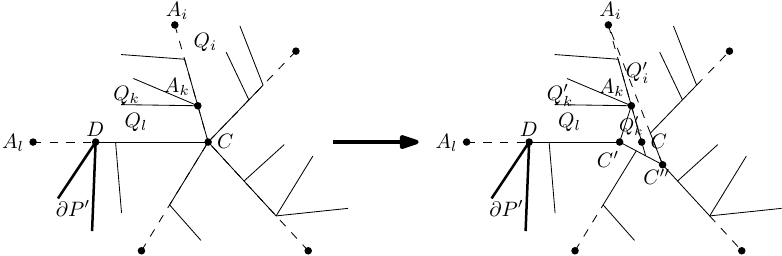}
    \caption{Case 2.3: If there are some star centers along $r_i = A_i C$, let $A_k$ be the farthest from $A_i$. Since $C$ is a convex corner at all pieces, there must be a $r_m$ to the left of $r_i$. Then we can take a sufficiently close point $C''$ along $r_m$ to $C$, and redistribute the quadrilateral $A_i A_k C' C''$ to $Q_i$ and $Q_k$. }
    \label{fig:case-2.33}
\end{cfigure}

\textbf{\boldmath Case 2.3: Some $r_l$ is of type \ref{sl:1}. }
Let $D$ be the corner of $P$ that supports $r_l$, which is not a star center. 
Without loss of generality, assume $r_l$ is horizontal with $C$ to the right. 
In this case, we might not be able to move $C$ to the right of $A_lD$, as $D$ blocks the visibility from $A_l$ to $C$. 
But we can still move $C$ to anywhere within a sufficiently small disk while keeping on or to the left of $A_l D$. 
See~\cref{fig:case-2.31}. 

In this case, we will fix all the pieces in $\mathcal{R}$ touching $A_l D$, and extend the highest priority among the others. 
Let $\mathcal{R}'$ be the set of pieces touching $\{r_1, \dots, r_j, CD\} \setminus \{r_l\}$ that is not in $\mathcal{Q}$. 
As before, assume $r_i$ is the sight line touching the highest priority piece. 

\begin{itemize}
    \item If there is a horizontal $r_k = A_kC$ with $A_k$ to the right of $C$, then we can reduce to case 1 by taking $r_k \cup r_l$ as $e$ and consider any other sight line ends at $C$; 
    \item If $T \in \mathcal{R}'$ has the highest priority, we can apply the modification in case 2.2 and expand $T$. 
    \item Otherwise, some $Q_i$ has the highest priority. If $A_i$ is on or below the straight line $A_l D$, we can slide $C$ along the ray $A_i C$ and expand $Q_i$. Let us assume $A_i$ is above the straight line $A_l D$. If $r_i = A_i C$ is not supported by any star centers, then we can take a sufficiently close point $C'$ to $C$ on $CD$, and give the triangle $A_iCC'$ to $Q_i$. See~\cref{fig:case-2.32}. If $r_i$ is supported by a star center, let $A_k$ be the farthest star center that supports $r_i$ from $A_i$, then $Q_i$ and $Q_k$ can collectively take some region around $A_k C$. See~\cref{fig:case-2.33}. \qedhere{}
\end{itemize}
\end{proof}

\begin{figure}
\centering
\includegraphics[page=4,width=.45\textwidth]{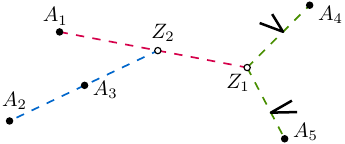}
\caption{Examples of the different sight line types of \Cref{lem:sltypes} and different Steiner points these can give rise to, as characterized by \Cref{lem:sightlines-steinerpoints}. The sight lines $A_4Z_1$ and $A_5Z_1$ are of type \ref{sl:1} (supported by a corner), meeting in point $Z_1$. The sight line $A_2Z_2$ is of type \ref{sl:2} (supported by star center $A_3$). The sight line $A_1Z_1$ is of type \ref{sl:3} (ending in the common endpoint of two sight lines of type \ref{sl:1}). We note that $Z_2$ here is a Steiner point arising from a sight line of type \ref{sl:2} ending on a sight line of type \ref{sl:3}. We will later on show the same diagram for the entire partition of the input polygon. There are five star centers and two corners involved in the definition of Steiner point $Z_2$, which turns out to be the worst case. In our algorithm, we will have $O(n^6)$ potential star centers, making for a total of $O(n^{32})$ potential such Steiner points.}
\label{fig:sltypes}
\end{figure}

See \Cref{fig:sltypes} for an example of the different types of sight lines and Steiner points (i.e.\ corners of the star pieces) needed in the partition.
The following corollary, which characterizes all required Steiner points, is immediate from \Cref{lem:newboundaryinsightlines,lem:sltypes}; indeed, each Steiner point must be on the beginning or end of a sight line. Together with \cref{cor:bit-complexity}, we can bound the complexity to encode each Steiner point. 

\begin{corollary}
\label{lem:sightlines-steinerpoints}
In an area maximum partition of $P'$, each Steiner point is of one of the following types:
\begin{enumerate}

\item The end of a sight line of type \ref{sl:1} or \ref{sl:2} on an edge of $P'$, \label{st:1}

\item The end of a sight line of type \ref{sl:1} or \ref{sl:2} on a sight line of type \ref{sl:1}--\ref{sl:3}, \label{st:2}

\item The common end of two sight lines of type \ref{sl:1}, \label{st:3}

\item
A star center. \label{st:4}
\end{enumerate}
And each Steiner point can be encoded by a sequence of $O(n)$ corners of $P'$. 
\end{corollary}

Combined with the result of \cref{sec:structure}, we can categorize all sight lines and Steiner points in some optimal constructable star partition. 
\begin{lemma}
    There exists an optimal constructable partition of any simple polygon $P$, such that each sight line $r = AC$ containing a segment on the shared boundary is of one of the following types: 
    \begin{enumerate}[label=(\alph*)]
        \item $r$ is supported by a corner of $P$; \label{sl:a}
        \item $r$ is supported by another star center; \label{sl:b}
        \item $C$ is the end of two non-parallel lines of type~\ref{sl:a}. \label{sl:c}
    \end{enumerate}
    
    And each Steiner point is of one of the following types:
    \begin{enumerate}[label=(\alph*)]
        \item The end of a sight line of type \ref{sl:a} or \ref{sl:b} on an edge of $P$, \label{st:a}
        \item The end of a sight line of type \ref{sl:a} or \ref{sl:b} on a sight line of type \ref{sl:a}, \ref{sl:b} or \ref{sl:c}, \label{st:b}
        \item The common end of two sight lines of type \ref{sl:a}. \label{st:c}
        \item
        A star center. \label{st:d}
    \end{enumerate}
    And each Steiner point can be encoded by a sequence of $O(n)$ corners of $P$. 
\end{lemma}

\begin{proof}
    Consider each connected components by adding all the construction lines as incisions. 
    Within each component, we take the area maximum partition. 

    We will first map the corners and edges in each connected component back to objects defined by $P$ and the star centers. 
    Note that in each connected component, all corners are either star centers, or tripod points, or a corner of the initial polygon $P$. 
    Each edge in the connected components, is either an edge of $P$, or a part of a construction segment. 
    Note that each construction segment is a sight line supported by a concave corner of $P$, which is of type \ref{sl:a}. 

    Next we will map all the sight lines in each connected component back to objects defined by $P$ and the star centers. 
    Here we will map them back type by type according to \cref{lem:sltypes}. 
    
    \textbf{\boldmath Type \ref{sl:1} sight lines. }
    All non-star-center corners of each connected component $P'$, must be either corners of $P$ or tripod points. 
    Since tripod points are convex corners in any connected components, it can not support any sight line from the interior. 
    Therefore, any sight line of type \ref{sl:1} must be supported by a corner of $P$, which is of type \ref{sl:a}. 

    \textbf{\boldmath Type \ref{sl:2} sight lines. }
    They are still supported by star centers, so all type \ref{sl:2} sight lines inside a connected component $P'$ are of type \ref{sl:b}. 

    \textbf{\boldmath Type \ref{sl:3} sight lines. }
    Since all type \ref{sl:1} sight lines are of type \ref{sl:a}, $C$ must be the end of two sight lines of type \ref{sl:a}, so all type \ref{sl:3} sight lines inside a connected component are of type \ref{sl:c}. 

    Next we will map all the Steiner points in each connected components to objects defined by $P$ and star centers. 
    Here we will map back class by class according to \cref{lem:sightlines-steinerpoints}. 

    \textbf{\boldmath Corners of a connected component. }
    They could be Steiner points in the partition of $P$ as well. 
    Note that in each connected component, all corners are either star centers (class \ref{st:d}), or tripod points (class \ref{st:c}), or a corner of the initial polygon $P$ (not a Steiner point), so they all fit into one of the classes. 
    
    \textbf{\boldmath Steiner points in class \ref{st:1}. }
    Since all type \ref{sl:1} sight lines in each connected component are of type \ref{sl:a}, all type \ref{sl:2} sight lines in each connected component are of type \ref{sl:b}, every Steiner point $C$ in class \ref{st:1} is the end of a sight line of type \ref{sl:a} or \ref{sl:b}. 
    Each edge in each connected component is either an edge of $P$, or a part of a sight line of type \ref{sl:a}, it falls into class \ref{st:a} or \ref{st:b}. 

    \textbf{\boldmath Steiner points in class \ref{st:2}. }
    Since the types of sight lines match, all Steiner points in class \ref{st:2} are in class \ref{st:b}. 

    \textbf{\boldmath Steiner points in class \ref{st:3}. }
    Since the types of sight lines match, all Steiner points in class \ref{st:3} are in class \ref{st:c}. 

    \textbf{\boldmath Steiner points in class \ref{st:4}. }
    They are in class \ref{st:d}. 
\end{proof}

\section{Algorithm} 
\label{sec:algorithm}
In this section we present our polynomial time algorithm to find a minimum star partition of a polygon. We restate our main \cref{thm:main-theorem} below, that we prove in this section.

\MainTheorem*

\begin{remark}
\label{rem:alg-running-time}
Although it is polynomial time, it is not exactly efficient. Since our main result is that the problem is in \textsf{P} (while previously it was not clear whether the problem was even in \textsf{NP}), we have not tried optimizing the running time.
We believe that it should not be particularly difficult to significantly improve the exponent something like $\approx 50$ by a more refined analysis.
For instance, using a smaller set of potential Steiner points would lead to a smaller state-space of the dynamic program (see \Cref{app:struct}).
Our aim here is to give the simplest possible description of an algorithm with polynomial running time.
Our techniques alone might not be sufficient to bring down the exponent to, say, a single digit.
We leave it as an open question to optimize the running time as far as possible, or conversely provide fine-grained lower bounds.
\end{remark}

\paragraph{Overview.} We begin with a brief overview of our algorithm (see also the technical overview in \cref{sec:technical-overview}).
There are two main challenges to overcome when designing a minimum star partition algorithm:
\begin{itemize}
\item First, even if we are given a set of potential Steiner points, it is not clear how to construct an optimal star-partition. 
\item Second, we need a way to find these potential Steiner points.
\end{itemize}
For the first challenge, we devise a dynamic programming algorithm. For the second, we rely heavily on our structural results in \cref{sec:structure} together with a ``greedy choice'' lemma. In fact, in order to find the potential Steiner-points, we need to invoke the dynamic programming algorithm (which assumes that we know all the potential Steiner points already) on many smaller instances in a recursive fashion.

\paragraph{Dynamic program.}
We begin by assuming that we know a set $\Scen$ of potential star centers. In \cref{sec:dp} we show a dynamic programming algorithm to find a partition of the polygon into a minimum number of star-shaped pieces \emph{such that the star center of each piece is in $\Scen$}. The algorithm runs in $O(\poly(n,|\Scen|))$ time. 
There are a few key properties that we show that allow us to define this dynamic programming algorithm (details can be found in \cref{sec:dp}):
\begin{itemize}
\item We show that using $\Scen$ and the corners of $P$, we can find a set of all potential Steiner points (e.g.\  internal corners of the star pieces). We do this by invoking our structural lemmas about \emph{area maximum partitions} from \cref{sec:area-max}. There will only be $O(\poly(n,|\Scen|))$ many of these potential Steiner points to consider.
\item We argue that each star piece touches the boundary in some optimal partition (\cref{cor:outer-planarity}).
\item
The above observation allows us to define a set of natural separators (see also \cref{fig:dp-separators}) involving at most two star pieces.
For points $B_1,B_2$ on the boundary of $P$, star centers $A_1,A_2\in \Scen$ and a potential internal corner $Z$ on the shared boundary of the two star pieces, we can define a (``long'') separator $B_1$-$A_1$-$Z$-$A_2$-$B_2$. We also consider (``short'') separators of the form $B_1$-$A_1$-$B_2$. These separators allow us to define a sub-region $P'$ of $P$ on one side of the separator, that we can recursively solve using a dynamic programming approach.
\end{itemize}

\paragraph{Finding potential star centers.} Given the above mentioned dynamic programming algorithm, the ultimate challenge is finding some relatively small (i.e.\ polynomial-sized) set of points $\Scen$ such that some optimal solution only uses star centers from $\Scen$. However, this turns out to be quite challenging and we present how we overcome this, together with the full algorithm, in \cref{sec:full-algo}.

A first attempt might be to consider $\Scen$ to be all the $O(n^4)$ points on the intersections of pairs of diagonals of the polygon. This turns out to not be sufficient, as can be seen in \cref{fig:art}. Indeed, the same figure shows that the star center points can have \emph{degree} as high as $\Omega(n)$ (in particular, the position of some star centers depend on up to $\Omega(n)$ corners of the input polygon).

Instead, here we use our crucial structural properties of optimal star partitions proven in \cref{sec:structure}. Essentially, we show there that the only non-trivial structure in some extreme optimal partitions must be \emph{tripods}(see \cref{sec:prelim}), e.g., like those in \cref{fig:art,fig:extripod,fig:greedy}. The tripods must be supported by three corners of the polygon, so there are only $O(n^3)$ such choices where a tripod can appear. However, the location of the tripod point might depend on other tripods (again, see \cref{fig:art} for a recursive construction capturing this). To overcome this, we need a \emph{greedy choice} property that allows us to argue that, for each potential tripod, there is only a single arrangement of this tripod we need to care about: the one that is least restrictive for one of the involved star centers.

To find this greedy arrangement of the tripod, we need to solve the minimum star partition problem on a subregion of the polygon. For this we can recursively call our algorithm to construct potential star centers for this smaller instance, and then use the dynamic programming algorithm to find the optimal star-partition. 

In \cref{sec:structure}, we argue that the tripods of some optimal solution are all oriented in a consistent way. Indeed, recall that each tripod is constructed by its two child star centers and used to construct its parent star center, so only one of the three subpolygons fenced off by the tripod depends on the other two subpolygons. This consistent orientation means that all the tripods can be oriented towards some arbitrary root face (see \cref{fig:tripodconsistent}). This is crucial for our algorithm since this allows us to bound the number of subproblems to $O(n^2)$ (each diagonal of $P$ will define a subproblem on the side not containing this root face, that can be solved first and must not depend on the other side).

\subsection{Dynamic Program} \label{sec:dp}

In this section we prove the following theorem.

\begin{theorem}
\label{thm:dp}
Suppose we are given a polygon $P$ with $n$ corners. Suppose also that we know some set $\Scen$ of potential star centers, such that we are guaranteed that there exists an optimal partition of $P$ into the minimum number of weakly simple star-shaped pieces where: (i) each star piece's center is in $\Scen$, and (ii) each star piece contains a corner of $P$.
Then we can find such an optimal solution in $O(\poly(n,|\Scen|))$ arithmetic operations\footnote{Instead of measuring running time here, we count the number of arithmetic operations. This is since points in $\Scen$ might be complicated to represent exactly. In fact, we will invoke the dynamic programming algorithm with points in $\Scen$ of \emph{degree} (and hence bit-complexity) $\Omega(n)$, so we cannot assume that we can perform computation on these points in $O(1)$ time.}.
\end{theorem}

\begin{remark}
\label{rem:dp}
For our purposes, we will have $|\Scen| = O(n^6)$, and the total running time of the dynamic programming algorithm in \cref{thm:dp} will be $O(n^{105})$ under the RAM model. 
\end{remark}

\subsubsection{Defining Other Steiner Points}
Suppose we are given a polygon $P$ and a set of potential star centers $\Scen$, as in the statement of \cref{thm:dp}. Using these, we will be able to identify all potential Steiner points needed for our dynamic program.
The main idea is to consider an optimal partition that is \emph{area-maximum}, and use our structural results from \cref{sec:area-max} (in particular \cref{lem:sightlines-steinerpoints}, that characterizes all potential Steiner points).
We will define the set $\Sint$ of potential Steiner points to be used as corners of the star-shaped pieces. Moreover, we define a smaller set $\Sbor\subseteq (\Sint\cap \partial P)$ of potential corners of the star pieces that are also on the boundary of $P$. 

\begin{lemma}
\label{lem:dp-points}
Let $P$ be a polygon and $\Scen$ a set of points satisfying the premise of \cref{thm:dp}.
Then we can find sets $\Sbor$ and $\Sint$ of size $\poly(n, |\Scen|)$ such that some weakly simple
minimum star partition $(Q_1, Q_2, \ldots Q_k)$ of $P$ with corresponding star centers $(A_1, A_2, \ldots, A_k)$ satisfies the following properties:
\begin{enumerate}
    \item Each piece $Q_i$ touches the polygon boundary $\partial P$.
    \item All star centers $A_i$ are contained in $\Scen$.
    \item All corners of $Q_i$ are contained in $\Sint$. 
    \item All corners of $Q_i$ that are also on the boundary of $P$ are contained in $\Sbor$. 
\end{enumerate}
\end{lemma}

\paragraph{Construction of Steiner points.}
We use the characterization of area maximum partitions from \cref{sec:area-max} in order to define the sets $\Sbor$ and $\Sint$. \Cref{fig:sltypes} shows the ``worst case'' example where some Steiner point depends on five star centers and two corners of $P$.
We begin by constructing the sight line types as in \cref{lem:sltypes}.

\begin{itemize}
\item Let $L^{(i)}$ be the set of lines passing through a potential star center in $\Scen$ and a distinct corner of $P$. Note that $|L^{(i)}| = O(n|\Scen|)$, and they correspond to sight lines of type~\ref{sl:1}.
\item Similarly, let $L^{(ii)}$ be the set of lines passing through a pair of distinct potential star center in $\Scen$. Note that $|L^{(i)}| = O(|\Scen|^2)$, and they correspond to sight lines of type~\ref{sl:2}.
\item To define $L^{(iii)}$, we first define $S^{(iii)}$ to be the set of points on the intersection of two non-parallel lines in $L^{(i)}$.  Then we define $L^{(iii)}$ to be the lines through a potential star center in $\Scen$ and a distinct point in $S^{(iii)}$.
Note that $|S^{(iii)}| = O(n^2|\Scen|^2)$, so $|L^{(iii)}| = O(n^2|\Scen|^3)$, and that these correspond to sight lines of type~\ref{sl:3}.
\end{itemize}

Now we are ready to use these sight lines to construct all necessary Steiner points, specifically, the different types specified in \cref{lem:sightlines-steinerpoints}.

\begin{itemize}
\item Let $S^{(1)}$ be the intersections of a segment of $P$ and a (non-parallel) line in $(L^{(i)}\cup L^{(ii)})$.
Note that $|S^{(1)}| = O(n|\Scen|^2)$
\item Let $S^{(2)}$ be the intersections of a line in $(L^{(i)}\cup L^{(ii)})$ and a (non-parallel) line in $(L^{(i)}\cup L^{(ii)}\cup L^{(iii)})$. Note that $|S^{(2)}| = O(n^2|\Scen|^5)$
\item Let $S^{(3)} := S^{(iii)}$ be the intersections of two (non-parallel) lines in $L^{(i)}$. Note that $|S^{(3)}| = O(n^2|\Scen|^2)$.
\end{itemize}

Finally, we can, by \cref{lem:sightlines-steinerpoints}, define our ``small'' sets of potential Steiner points to consider:
\begin{itemize}
    \item $\Sint = \text{corners}(P)\cup S^{(1)}\cup S^{(2)}\cup S^{(3)}$ for internal corners of star pieces, with $|\Sint| = O(n^2|\Scen|^5)$.
    \item $\Sbor = \text{corners}(P)\cup S^{(1)}$ for corners of star pieces also on the boundary $\partial P$, with
    $|\Sbor| = O(n|\Scen|^2)$.
\end{itemize}
We additionally note that $\Scen\subseteq S^{(2)}\subseteq \Sint$ (since a point $A\in \Scen$ will lie on at least two lines in $L^{(i)}$ as $P$ has at least three non-collinear corners). Similarly $(\Scen\cap \partial P)\subseteq S^{(1)}\subseteq \Sbor$.

\begin{proof}[Proof of Lemma~\ref{lem:dp-points}]
Consider any minimum star partition $\mathcal{Q} = (Q_1, Q_2, \ldots Q_k)$---with star centers $(A_1, A_2, \ldots A_k)$---of $P$ that satisfies the premise of \cref{thm:dp}: that is each star center is in $\Scen$ and each piece touch the boundary of $P$ at some corner.

We now consider an area maximum partition 
 $\mathcal{Q}' = (Q'_1, Q'_2, \ldots Q'_k)$ with the same star centers $(A_1, A_2, \ldots A_k)$.
By \cref{lem:areamax,lem:sltypes,lem:sightlines-steinerpoints}, this partition must satisfy that each corner of $Q_i$ is in $\Sint$ and if this corner is also on the boundary $\partial P$ it must be in $\Sbor$. Indeed $L^{(i)}, L^{(ii)}, L^{(iii)}$ must contain all possible sight line types of \cref{lem:sltypes}, and so $S^{(1)}$, $S^{(2)}$, $S^{(3)}$ must contain all potential Steiner points as specified in \cref{lem:sightlines-steinerpoints}.

What remains is to argue that each star piece touches the boundary of $P$. This is non-trivial, and unfortunately does not seem to follow directly from area-maximality. Instead we use the fact that the star pieces in the original partition $\mathcal{Q}$ touched the boundary at some corner. For each piece $Q_i$ we can choose an arbitrary sight line $r_i = A_iB_i$ to a corner $B_i$ of $P$. Intuitively, we then ``fix'' this sight line before doing area-maximality. That is, we instead consider $\mathcal{Q}'$ to be an area maximum partition where the star centers $A_i$ are fixed, \emph{and the chosen sight lines $r_i$ must be contained in piece $Q_i$}.

Formally, we can do this by changing the input polygon $P$ into a weakly simple polygon $P'$ defined as follows. For each chosen sight line $r_i$ we add it as an ``incision'' to $P'$ (so now $P'$ is a weakly simple polygon).
Note that these incisions cannot intersect except at their endpoints. The partition $\mathcal{Q}$ is also a minimum star-partition of $P'$, but now each star center is at a corner of $P'$. If $\mathcal{Q}'$ is chosen to be area maximum in this new polygon $P'$, we can then look at $\mathcal{Q}'$ as a partition of $P$ where we assign the ``incision'' $r_i$ to piece $Q_i'$. This means that each piece must touch the boundary (perhaps only because of a degenerate ray from the star center to some corner at $P$, but this is acceptable since we allow the pieces to be weakly simple polygons).

What remains is to argue that $\Sint$ and $\Sbor$ are still sufficient, i.e.\ that we did not introduce any new Steiner points. All new corners of $P'$ were star centers, so we did not introduce any additional sight lines for \cref{lem:sltypes}. The additional edges of $P'$ are the ``incisions'' $r_i$, but these will all be in $L^{(i)}$ (lines from star centers to corners of $P$), so they are considered as potential endpoints of sight lines in item 2 (instead of item 1) in \cref{lem:sightlines-steinerpoints}.
\end{proof}

\subsubsection{Dynamic Programming Algorithm}

\begin{figure}
\centering
\includegraphics[page=3]{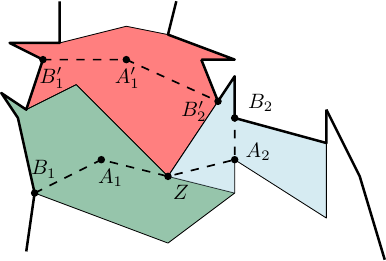}
\caption{A short separator $B'_1$-$A'_1$-$B'_2$ and a long separator $B_1$-$A_1$-$Z$-$A_2$-$B_2$ as part of a star partition.}
\label{fig:dp-separators}
\end{figure}

We now provide our dynamic programming algorithm (\cref{alg:dp}) that will consider each possible star-partition satisfying the properties of \cref{lem:dp-points}, and thus will find an optimal partition given the set $\Scen$. 
Let $B_1, B_2$ be two points on $\partial P$, $P[B_1:B_2] \subset \partial P$ be the chain from $B_1$ to $B_2$ along $\partial P$ in counterclockwise order. 
We consider the separators (see also \cref{fig:dp-separators}):
\begin{itemize}
\item Short separator of the form $B_1$-$A_1$-$B_2$ for $B_1,B_2\in \Sbor$, and $A_1\in \Scen$. 
\item Long separator of the form $B_1$-$A_1$-$Z$-$A_2$-$B_2$ for $B_1,B_2\in \Sbor$, $A_1,A_2\in \Scen$, and $Z\in \Sint$.
\end{itemize}

In the dynamic program, we will, for a given separator, calculate an optimal way to partition the subpolygon $P'$ enclosed by $P[B_1:B_2]$ and the separator, 
given that there are star centers already placed at $A_1$ (and $A_2$ in case of a long separator) on the separator. Since each piece touches the boundary, we will see that it is sufficient to consider separators passing through at most two star pieces.
We describe a few elementary ways to build separators for larger and larger subpolygons $P'$ by e.g.\ merging two separators or moving the common corner point $Z$.
In figure \cref{fig:dp-cases} and the pseudo-code \cref{alg:dp} we can see the different cases we consider for transitions. We also explain the cases here:

\begin{figure}
\centering
\includegraphics[width=.9\textwidth,page=1]{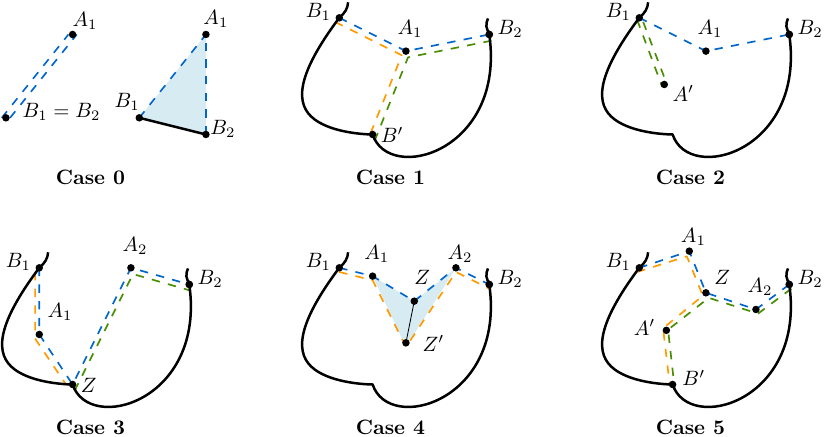}
\caption{The different transitions we need to consider for the dynamic program algorithm.
Cases 0-2 concerns short separators $B_1$-$A_1$-$B_2$, 
and Cases 3-5 concerns long separators $B_1$-$A_1$-$Z$-$A_2$-$B_2$, and we want to solve the subpolygon ``below'' these separators.
Curve parts indicate that the details have now been shown.}
\label{fig:dp-cases}
\end{figure}

\newcommand\DPCase[1]{Case~\hyperref[itm:case#1]{#1}}

\begin{description}
 \item[Case 0:\phantomsection\label{itm:case0}] (Base Case) In the base case we consider trivial short separators, $B_1$-$A_1$-$B_2$ where either $B_1 = B_2$, or $B_2$ is next to $B_1$ in counterclockwise order. 
 Here $B_1 A_1 B_2$ forms a possibly degenerate triangle with one side on the boundary $\partial P$, that can be assigned to the star piece with center $A_1$
 \item[Case 1:\phantomsection\label{itm:case1}] (Merge short + short) A short separator $B_1$-$A_1$-$B_2$ can be seen as the ``merge'' of two other short separators $B_1$-$A_1$-$B'$ and $B'$-$A_1$-$B_2$ for some $B'\in \Sbor \cap P[B_1:B_2]$.
 \item[Case 2:\phantomsection\label{itm:case2}] (New star center) When a short separator $B_1$-$A_1$-$B_2$ is neither trivial (\DPCase{0}) or the merge of two short separators (\DPCase{1}), some other star center $A'$ must be able to see the boundary point $B_1$ too. This becomes a long separator $B_1$-$A'$-$B_1$-$A_1$-$B_2$, where the segment $B_1A'$ is a ``spike'' that occurs twice.
 \item[Case 3:\phantomsection\label{itm:case3}] (Combine short + short) A long separator $B_1$-$A_1$-$Z$-$A_1$-$B_2$ where $Z$ is on the boundary somewhere between $B_1$ and $B_2$ can be decomposed into two short separators $B_1$-$A_1$-$Z$ and $Z$-$A_2$-$B_2$.
 \item[Case 4:\phantomsection\label{itm:case4}] (Move common corner) A long separator $B_1$-$A_1$-$Z$-$A_1$-$B_2$ can also arise by moving the common corner $Z$ from some other point $Z'$, where $B_1$-$A_1$-$Z'$-$A_1$-$B_2$  is also a long separator. Here the triangles $A_1 Z' Z$ and $A_2ZZ'$ can be assigned to the star piece with centers $A_1$ and $A_2$ respectively.
 \item[Case 5:\phantomsection\label{itm:case5}] (Merge long + long) For a long separator $B_1$-$A_1$-$Z$-$A_1$-$B_2$, if neither the common corner $Z$ can be moved (Case 4), nor it is on the boundary (\DPCase{3}), there must exists some other star center $A'$ that can see $Z$. The star piece with center $A'$ must also touch the boundary at some point, say at $B'$. Then 
 our separator is a ``merge'' of two long ones: 
 $B_1$-$A_1$-$Z$-$A'$-$B'$
 and  $B'$-$A'$-$Z$-$A_2$-$B_2$.
\end{description}
Note that it is only in \DPCase{0} and \DPCase{4} where we actually assign some positive area of $P$ to some star piece.
Whenever we say \textbf{``$X$ can see $Y$''} in \cref{alg:dp}, we mean that the segment $XY$ is contained within the subpolygon of $P$ restricted by the separator.

\SetKwFunction{SolveSep}{SolveSeparator}

To use the dynamic programming algorithm to find an optimal star partition, we arbitrarily pick consecutive points $B_1, B_2\in \Sbor$ on the boundary of $P$, where $B_2$ is next to $B_1$ in clockwise order. 
There must be some star piece seeing this segment, so we can simply try each possibility of star centers $A\in \Scen$ and call $\SolveSep(B_1,A,B_2)$ to find the optimal solution given that $A$ sees the segment $B_1B_2$ (and then just return the best solution we found). 
Note that we do not consider $B_1$ and $B_2$ to be ``adjacent'' here for \DPCase{0}, as the region enclosed by $P[B_1:B_2]$ and $B_1$-$A$-$B_2$ is $P \setminus B_1AB_2$. 

\begin{observation}
\label{obs:multiple-solutions}
Note that the above actually gives us all possible positions, in optimal solutions, for star centers $A\in\Scen$ that see the segment $B_1B_2$. This will be useful later in the full algorithm.
\end{observation}

\SetKwFunction{FMain}{SolveSeparator}
\begin{algorithm}[!htb]
\caption{Dynamic programming algorithm}
\label{alg:dp}
\small
\DontPrintSemicolon
\SetKwData{opt}{opt}
\SetKwProg{Fn}{function}{:}{}
\Fn{\FMain{$B_1$, $A_1$, $B_2$}}{

    \tcp{returns the minimum number of additional (not counting $A_1$) star pieces needed to cover the enclosed region of $P[B_1:B_2]$ and $B_1$-$A_1$-$B_2$.}

    Let $P'$ be the region enclosed by $P[B_1:B_2]$ and $B_1$-$A_1$-$Z$-$A_2$-$B_2$, compute visibility of $\Scen \cup \Sbor \cup \Sint$ within $P'$\;
    
    $\opt \gets n$\;
    
    \If(\tcp*[f]{\DPCase{0}: base case}){$B_1 = B_2$ or $B_2$ is next to $B_1$ in counterclockwise order}{
        $\opt \gets 0$\;
    }

    \For(\tcp*[f]{\DPCase{1}: merge short + short}){$B'\in(\Sbor\setminus\{B_1, B_2\}) \cap P[B_1:B_2]$}{
        \If{$A_1$ can see $B'$}{
            $\opt \gets \min(\opt, {\FMain}(B_1, A_1, B') + {\FMain}(B', A_1, B_2))$\;
        }
    }
    
    \For(\tcp*[f]{\DPCase{2}: new star center}){$A'\in(\Scen\setminus\{A_1\}) \cap P'$}{
    \If{$A'$ can see $B_1$}{
        $\opt \gets \min(\opt, 1 + {\FMain}(B_1, A', B_1, A_1, B_2))$\;
        }
    }
   
    \KwRet{\opt}\;
}
\;
\Fn{\FMain{$B_1$, $A_1$, $Z$, $A_2$, $B_2$}}{
    \tcp{returns the minimum number of additional (not counting $A_1$ or $A_2$) star pieces needed to cover the enclosed region of $P[B_1:B_2]$ and $B_1$-$A_1$-$Z$-$A_2$-$B_2$}

    Let $P'$ be the region enclosed by $P[B_1:B_2]$ and $B_1$-$A_1$-$Z$-$A_2$-$B_2$, compute visibility of $\Scen \cup \Sbor \cup \Sint$ within $P'$\;
    
    $\opt \gets n$\;
    
    \If(\tcp*[f]{\DPCase{3}: combine short + short}){$Z\in\Sbor \cap P[B_1:B_2]$}{
        $\opt \gets \min(\opt, {\FMain}(B_1, A_1, Z) + {\FMain}(Z, A_2, B_2))$\;
    }
    
    \For(\tcp*[f]{\DPCase{4}: move common corner}){$Z' \in (\Sint \setminus \{Z\}) \cap P'$}{
        \If{all three of $A_1$, $A_2$ and $Z$ can see $Z'$}{
            $\opt \gets \min(\opt, {\FMain}(B_1, A_1, Z', A_2, B_2))$\;
        }
    }
    
    \For(\tcp*[f]{\DPCase{5}: merge long + long}){$A'\in(\Scen\setminus\{A_1, A_2\}) \cap P'$}{
        \For{$B'\in\Sbor \cap P[B_1:B_2]$}{
            \If{$A'$ can see $B'$ and $Z$}{
                $\opt \gets \min(\opt, 1 + {\FMain}(B_1, A_1, Z, A', B') + {\FMain}(B', A', Z, A_2, B_2))$\;
            }
        }
    }

  \KwRet{\opt}\;
}
                                                                                                                                  
\end{algorithm}

\paragraph{Correctness.}
We now argue that \cref{alg:dp} is correct, that is that the optimal solution can be constructed using the transitions (cases) in \cref{fig:dp-cases}.
Suppose we have some optimal partition satisfying the properties of \cref{lem:dp-points}. We will show that the dynamic program will consider this optimal solution.

Let us first consider the case that we are looking at a short separator $B_1$-$A_1$-$B_2$ in this optimal partition. If either $B_1 = B_2$ or $B_1$ and $B_2$ are consecutive points on the border of $P$ in $\Sbor$, we are in \DPCase{0}. Otherwise, in the optimal partition, either the star piece with center $A_1$ will also touch the boundary somewhere in between $B_1$ and $B_2$, or not. In case it does, it must touch in a point $B'\in \Sbor$, where we naturally have two short separators of sub-regions $B_1$-$A_1$-$B'$ and $B'$-$A_1$-$B_2$, which is handled by \DPCase{1}. In case it does not, there must be some other star piece (say with center $A'$) that touches $B_1$, and then we are in \DPCase{2} with long separator $B_1$-$A'$-$B_1$-$A_1$-$B_2$.

Now suppose instead that we are looking at a long separator $B_1$-$A_1$-$Z$-$A_2$-$B_2$ that is part of the optimal partition. This means that the two pieces with centers $A_1$ and $A_2$ touch. Note that they will touch in a single contiguous internal boundary (\cref{lem:newboundaryinsightlines} give a complete characterization of how this boundary can look; it is either a single point or up to two line segments). Note that $Z$ must be a point on this contiguous internal boundary. If $Z$ is not the last corner on this boundary, we can move it to the next corner $Z'$, as in \DPCase{4}. If $Z$ instead was the last corner on this boundary between pieces with centers $A_1$ and $A_2$, we have two sub-cases: (i) either $Z$ is on the boundary of the polygon $P$, or (ii) else there must be some other star piece touching $Z$. In sub-case (i) it must be the case that $Z\in \Sbor$, and we have two natural short separators for sub-regions: $B_1$-$A_1$-$Z$ and $Z$-$A_2$-$B_2$, as handled by \DPCase{3}. In sub-case (ii), let $A'\in \Scen$ be the star center of the additional piece touching $Z$ in the optimal partition. Note that $A'$ must also touch the boundary of $P$ somewhere (\cref{lem:dp-points}), say in point $B'\in \Sbor$. Again, we have two natural (long) separators for sub-regions: $B_1$-$A_1$-$Z$-$A'$-$B'$ and $B'$-$A'$-$Z$-$A_2$-$B_2$, which is handled by \DPCase{5}.

To conclude, any optimal partition satisfying the properties of \cref{lem:dp-points} must be constructable by Cases \hyperref[itm:case0]{0}-\hyperref[itm:case5]{5}. Since \cref{alg:dp} considers all these cases as transitions, it will find an optimal partition of $P$ into star-shaped pieces.

\paragraph{Running time.}
To make \cref{alg:dp} run in polynomial time we assume standard memoization, i.e.\ that if \FMain is called several times with the same arguments it only needs to be solved once. Since $\Sbor$ and $\Sint$ are both of size $\poly(n,|\Scen|)$, it is clear that we have a polynomial many separators and polynomially many transitions, and therefore \cref{alg:dp} runs in $O(\poly(n,|\Scen|))$ time (proving \cref{thm:dp}). Below we analyze the complexity in more detail. 

Note that in the 3-parameter function \FMain{$B_1, A_1, B_2$} (short separators), we have $O(|\Sbor|^2\cdot|\Scen|)$ states, as $B_1,B_2\in \Sbor$ and $A_1\in \Scen$.
Similarly, in the 5-parameter function \FMain{$B_1, A_1, Z, A_2, B_2$} (long separators), we have $O(|\Sbor|^2\cdot|\Scen|^2\cdot \Sint)$ states, as $B_1,B_2\in \Sbor$, $A_1,A_2\in \Scen$, and $Z\in \Sint$.
We count the number of transitions in the algorithm for each ``Case'':
\begin{description}
\item[\DPCase{0}:] $O(1)$ transitions for $O(|\Sbor|^2\cdot |\Scen|)$ many separators.
\item[\DPCase{1}:] $O(|\Sbor|)$ transitions for $O(|\Sbor|^2\cdot |\Scen|)$ many separators.
\item[\DPCase{2}:] $O(|\Scen|)$ transitions for $O(|\Sbor|^2\cdot |\Scen|)$ many separators.
\item[\DPCase{3}:] $O(1)$ transitions for $O(|\Sbor|^2\cdot |\Scen|^2\cdot |\Sint|)$ many separators.
\item[\DPCase{4}:] $O(|\Sint|)$ transitions for $O(|\Sbor|^2\cdot |\Scen|^2\cdot |\Sint|)$ many separators.
\item[\DPCase{5}:] $O(|\Scen|\cdot |\Sbor|)$ transitions for $O(|\Sbor|^2\cdot |\Scen|^2\cdot |\Sint|)$ many separators.
\end{description}
We see that \DPCase{4} and \DPCase{5} dominate all other cases, for a total of
\begin{equation*}
O(|\Sbor|^3\cdot|\Scen|^3\cdot|\Sint| +
|\Sbor|^2\cdot|\Scen|^2\cdot|\Sint|^2)
\end{equation*}
many transitions.
For each of these transitions, we might need to go through all $O(n)$ segments of the polygon to verify the ``$X$ can see $Y$'' statements, adding another factor of $O(n)$ to the running time. Since $|\Sbor|, |\Sint| = \poly(n,|\Scen|)$, we have proved \cref{thm:dp}.
\begin{remark}
Plugging in $|\Sbor| = O(n|\Scen|^2), |\Sint| = O(n^2|\Scen|^5)$, we see that the total number of transitions we have is
$O(n^6|\Scen|^{16})$, assuming $|\Scen| \ge n$. In the final algorithm we will have $|\Scen| = O(n^6)$, making for a total of $O(n^{102})$ transitions!
This means there are a total of $O(n^{103})$  arithmetic operations to run the dynamic programming algorithm (to check the ``$X$ can see $Y$'' statements). 
According to \cref{cor:bit-complexity} and \cref{lem:sightlines-steinerpoints}, we will consider only points with degree $O(n)$ (i.e.\ all points can be described by $O(n)$ arithmetic operations from the input points), so we can perform arithmetic operations in $O(n^2)$ time naively. 
Therefore the total running time of \cref{alg:dp}, for our purposes, is $O(n^{105})$, as stated in \cref{rem:dp}.
\end{remark}

\subsection{Finding Star Centers \& Full Algorithm}
\label{sec:full-algo}

Now we turn to show our full polynomial time algorithm, thus proving \cref{thm:main-theorem}. We note that if we can find a relatively small set of potential star centers, we can simply use our dynamic programming algorithm (\cref{alg:dp} and \cref{thm:dp}). However, we will see that in order to find such a sufficient set of potential star centers, we will need to solve smaller instances of the same problem (where we need to invoke the algorithm recursively).

\paragraph{Constructable partitions.}
Throughout this section, we will let the \emph{root edge} $r$ be an arbitrary edge of $P$. We will focus our attention to optimal partitions that can be constructed using the process defined in \cref{sec:structure}---specifically by \cref{thm:constructability}---and call such a star-partition \emph{constructable}. In particular, we recall that a \emph{constructable} partition satisfies the following properties:
\begin{enumerate}[label=(\roman*)]
    \item\label{itm:consti} It is optimal, that is it uses a minimum number of star pieces.
    \item\label{itm:outer-planar} Each star piece touches the boundary of $P$ at some corner.
    \item All tripods in the partition are oriented towards the root face (the face with the root edge $r$).
    \item\label{itm:constiv} All star centers are at the intersection of two lines, each of these lines are either an edge of $P$, a diagonal between two concave corners of $P$, or a line through a tripod point and one of its supporting corners. 
\end{enumerate}
Indeed, by \cref{thm:constructability} (and \cref{lem:outer-planar} for item~\ref{itm:outer-planar}) there must exist some constructable partition. However, restricting ourselves to constructable partitions is not enough to get an efficient algorithm: in general there are a double-exponential number of points that appear as star centers in some constructable partition. Therefore we seek to restrict our class of optimal partitions further, and here the \emph{greedy choice} comes into play (defined below in \cref{sec:greedy-choice}).

Before presenting the \emph{greedy choice}, we prove a simple lemma saying that the partition inside the pseudo-triangle of a tripod is not particularly important, and that any partition outside this pseudo-triangle can always be extended to cover the pseudo-triangle too. This will be useful for our algorithm, since we can then focus on solving sub-problems defined by diagonals of $P$ (and not defined by the unknown tripod).

\begin{lemma}
\label{lem:extending-to-tripod}
(See \cref{fig:extending-to-tripod}).
Suppose $T$ is a tripod supported by corners $(D_1, D_2, D_3)$ in some constructable partition, with tripod point $C$. Let $P'$ be one of the three sub-polygons $T$ splits $P$ into (say between corners $D_1$ and $D_3$), and say $A_1\in P'$ is the star centers participating in $T$ through corner $D_1$. Let $\Delta$ be the pseudo-triangle of the points $(D_1, D_2, D_3)$. Note that $P'\setminus \Delta$ consists of several (at least one) sub-polygons, call them $P_1, P_2, \ldots, P_k$, where $A_1\in P_1$.

Then \emph{any} partition of $P_1, P_2, \ldots, P_k$, where $A_1$ is a star center in $P_1$ seeing corner $D_1$, can be extended to a partition of $P'$ (without moving star centers) that $A_1$ sees $C$. 
\end{lemma}

\begin{figure}
\centering
\includegraphics[width=.9\textwidth,page=6]{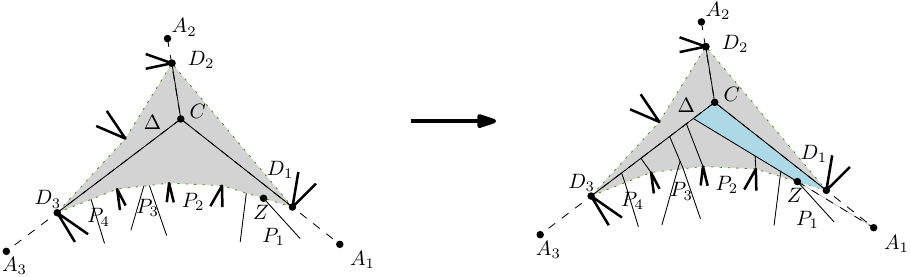}
\caption{Modification as in  \cref{lem:extending-to-tripod}. The pseudo-triangle $\Delta$ splits the polygon up into multiple sub-polygons, where we let $P_1, P_2, \ldots$ be those on the same ``side'' as the star-center $A_1$. Any partitions of $P_1, P_2, P_3, \ldots$ can be extended to the polygonal line $D_3$-$C$-$D_1$ in such a way that $A$ sees the tripod point $C$. We first give the blue region to $Q_1$, then extend all the segments touching the pseudo-diagonal $D_1D_3$ one by one, from right to left, until it meets any existing segment. 
}
\label{fig:extending-to-tripod}
\end{figure}

\begin{proof}
See \cref{fig:extending-to-tripod}. 
Let $s_1, \dots, s_m$ be the set of segments in the partition of $P_1, \ldots, P_k$ that touches $\partial \Delta$, ordered by the touching point from $D_1$ to $D_3$. 
We simply extend $s_1, \ldots, s_m$ one by one, until they meet the separator $D_3$-$C$-$D_1$ or previous extended boundaries. 

We will handle the piece $Q_1$ with star center $A_1$ in a special manner, since we want it to see the tripod point $C$. So let $Z$ be the point on $\partial \Delta$ so that $Q_1$ contains the segment $D_1Z$ on this boundary. Then we extend the line $A_1Z$ until it meets the line $D_3C$ first, and assign the blue region to $Q_1$. 
This clearly leads to a partition of $P'$ while keeping the assignment on $P_1 \cup \ldots \cup P_k$. 
Since $C$ is a convex corner of $P'$, and all new angles are intersection of a ray and a straight line (which must be convex), all pieces must remain in star shape. 

What remains is to argue that everything inside $\Delta$ is covered by our extended partition. This is true as long as the pseudo-diagonal between $D_1$ and $D_3$ contains no edge of $P$ (i.e. $\partial \Delta\cap \partial P$ just contains a finite amount of points, and no line segments). To argue this we use that $T$ was a tripod of some constructable solution. In particular this means that there exists some partition where $T$ is a tripod and no star center lies within the pseudo-triangle $\Delta$. Hence, if there was some edge $e$ of $P$ contained in $\partial \Delta$, then the interior of this edge could not be seen by any star center in such a constructable partition, which is our desired contradiction.
\end{proof}

\subsubsection{Greedy Choice}
\label{sec:greedy-choice}

\begin{figure}
\centering
\includegraphics[width=.95\textwidth,page=7]{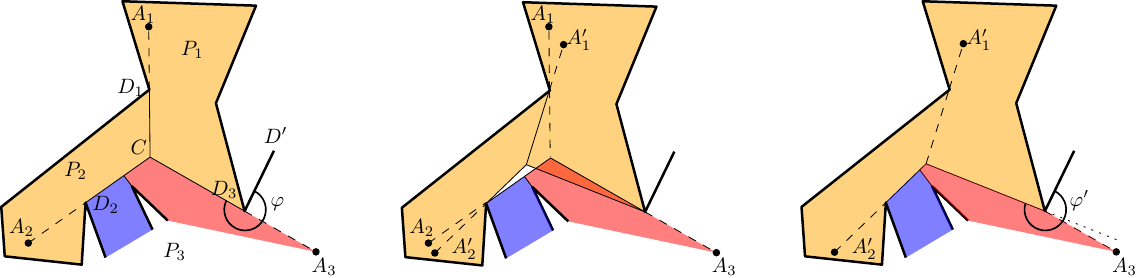}
\caption{\emph{Left:} Two child star-centers $A_1$ and $A_2$ define a (fake) tripod with tripod point $C$, splitting the polygon into three subpolygons: ``childs'' $P_1, P_2$, and ``parent'' $P_3$. The angle $\varphi = \angle C D_3 D'$ of this tripod is a measure on how ``restricted'' a potential star-center $A_3$ defined by this tripod is. \emph{Middle:} Two other star centers $A'_1$ and $A'_2$ define another fake tripod on the same support, which is less restrictive (i.e., with an angle $\varphi' < \varphi$). Inside $P_3$ (the parent-side of the tripod), the same partition is shown (in red and blue) as on the left. \emph{Right:} The star partition of subpolygon $P_3$ can be adjusted (without moving any star centers) to also work with the less restrictive tripod.}
\label{fig:greedy-choice-angle}
\end{figure}

Consider three concave corners $(D_1, D_2, D_3)$ of $P$, that might support a tripod in some \emph{constructable} partition.
We now argue that if there are many possibilities for how the legs of a tripod supported by $(D_1, D_2, D_3)$ look like, then it suffices to consider a single one of these possibilities! We will call this arrangement the \emph{greedy choice} of this tripod. Recall that the tripod point is constructed by two of the sub-polygons $P_1,P_2$, and used to define the third sub-polygon $P_3$ (see \cref{fig:greedy,fig:greedy-choice-angle}). In particular, $P_1$ and $P_2$ are the children in the \emph{tripod tree} (see \cref{fig:tripod-tree}), and $P_3$ the parent. We argue that the greedy choice will be the combination of optimal solutions in the two subpolygons $P_1$ and $P_2$ that give rise to the \emph{least restrictive} tripod-center when constructing the optimal solution for $P_3$. With \emph{least restrictive} we mean the one that minimizes the angle $\varphi = \angle C D_3 D'$ as in \cref{fig:greedy-choice-angle}, as such a tripod will only impose the mildest restrictions on where the star center $A_3\in P_3$ participating in the tripod lies.

We begin by proving that it never hurts replacing a tripod with a less restrictive one, see also \cref{fig:greedy-choice-angle}. 
\begin{lemma}
\label{lem:less-restrictive}
Suppose that there is a tripod $\mathcal{T}$ with tripod point $C$ supported by three corners $(D_1,D_2,D_3)$, part of a constructable partition $\mathcal{Q}$.
Let $\Delta$ be the pseudo-triangle of the tripod, and consider the three sub-polygons $P_1, P_2, P_3$ in $P\setminus \Delta$ participating in the tripod, such that the parent star center of $\mathcal{T}$ is contained in $P_3$, as in \cref{fig:greedy-choice-angle,fig:greedy}.

Suppose now that there is some other constructable sub-partitions of $P_1$ and $P_2$ (using the same number of star pieces as in the original one) giving rise to another fake tripod $\mathcal{T}'$ (supported on the same three concave corners $D_1, D_2, D_3$) with tripod point $C'$, that is less restrictive (the angle $\varphi$ in \cref{fig:greedy-choice-angle} is smaller) for $P_3$. Then there also exists a constructable optimal star partition of $P$ that contains these new sub-partitions of $P_1$ and $P_2$, and the fake tripod $\mathcal{T}'$.
\end{lemma}

\begin{proof}
    By \cref{lem:extending-to-tripod}, there exists an optimal partition $\mathcal{Q}'$ containing the same sub-partition in $P_1$ and $P_2$ that gives $\mathcal{T}'$. Since no star center is in the pseudo-triangle $\Delta$ of $\mathcal{T}$, the pseudo-triangle $\mathcal{T}'$ is also constructable with respect to $\mathcal{Q}'$. 
    Our lemma then follows directly from \cref{lem:construct-from-middle}.
\end{proof}

\paragraph{Greedy-constructable partitions.}
By \cref{lem:less-restrictive}, it never hurts to replace a tripod with a less restrictive one. The next step is to argue that we can assume that \emph{all} tripods in our constructable partition can follow such a greedy choice---which is a very useful property when designing an algorithm. This is a bit subtle, since such an algorithm might not actually find the least restrictive version of a tripod $\mathcal{T}$, but only the least restrictive version \emph{given that all children tripods also follow the greedy choice}. We formalize this by the notion of \emph{greedy-constructable} partitions.
A \emph{greedy-constructable} partition is \emph{constructable} and also satisfies the following extra property, adding to properties \ref{itm:consti}--\ref{itm:constiv} of a constructable partition.
\begin{itemize}
\item[(v)] Any tripod $\mathcal{T}$ in the partition is \emph{greedy} (see \cref{def:greedy-choice} below).
\end{itemize}

\begin{definition}[Greedy Choice \& Less Restrictive Tripods]
\label{def:greedy-choice}
First, we define the \emph{angle} of a fake tripod~$\mathcal{T}$ as in \cref{fig:greedy-choice-angle}, i.e., the angle $\varphi = \angle CD_3D'$, where $D'$ is the next corner of $P$ after $D_3$, in the parent-subpolygon $P_3$.
We say that a fake tripod $\mathcal{T}$ is \emph{less restrictive} than another fake tripod $\mathcal{T}'$ (supported by the same three corners) if the angle is smaller for $\mathcal{T}$ than for $\mathcal{T}'$, i.e.\ if $\mathcal{T}$ imposes a weaker restriction on the potential star center in the parent subpolygon.
We break ties in an arbitrary but consistent manner.

Consider a tripod $\mathcal{T}$ supported by corners $(D_1, D_2, D_3)$ in some constructable partition. Suppose this tripod splits the polygon into the three sub-polygons $P_1, P_2, P_3$ and is oriented towards $P_3$. Consider all the pairs of \emph{greedy-constructable}\footnote{We note that the definition of \emph{greedy choice} tripods and \emph{greedy-constructable} partitions are mutually dependent on each other. With greedy-constructable for a sub-polygon, we mean that all the tripods used in the partition of the sub-polygon abides the greedy choice. Since the tripods form a rooted tree (see \cref{lem:tripod-tree}), this recursive definition is well-defined, as tripods only depend on other tripods deeper down in the tree.} sub-partitions of $P_1$ and $P_2$, giving rise to some (fake) tripods $\mathcal{T}'$ supported by the same corners $(D_1, D_2, D_3)$. Then the tripod $T$ is the \emph{greedy choice} for $(D_1,D_2,D_3)$, or we simply say that $\mathcal{T}$ is \emph{greedy}, if it is the \emph{least restrictive} for $P_3$ among all such (fake) tripods $\mathcal{T}'$.
\end{definition}
 
\begin{lemma}
\label{lem:greedy-constructable}
There exists a greedy-constructable partition. 
\end{lemma}

\begin{proof}
The general idea is to start with a constructable partition (which exists by \cref{thm:constructability}), and replace tripods that are not greedy by their greedy version instead, using \cref{lem:less-restrictive}.
At first it might not be apparent that this will work, since \cref{lem:less-restrictive} might introduce new tripods that are not greedy. We overcome this by carefully eliminating bad tripods in a bottom-up fashion, and continuing in a recursive manner, similar to our proof that there exists \emph{constructable} partitions (proof of \cref{thm:constructability}).

Formally, let us consider some polygon $P$ together with a \emph{constructable} partition $\mathcal{Q}$ of $P$. Throughout this proof, whenever we say \emph{tripod}, we also consider \emph{fake tripods} (see discussion in \cref{sec:tripod-trees,rem:fake-tripod})---i.e.\ those that can be constructed by two children star centers but might not be used to construct a parent star center.  Consider the \emph{rooted tripod tree} (see \cref{lem:tripod-tree,fig:tripodconsistent,fig:tripod-tree}) of this partition. In this tree we mark all tripods that are \emph{not} greedy as bad, along with all ancestor tripods. Call this set of bad tripods $\mathcal{T}_{bad}(\mathcal{Q})$, and let $\mathcal{T}_{good}(\mathcal{Q})$ be all other tripods. That is $\mathcal{T}_{good}(\mathcal{Q})$  is the set of tripods $T$ in $\mathcal{Q}$ for which all tripods in the subtree rooted at $T$ (including $T$ itself) are greedy. If $\mathcal{T}_{bad}(\mathcal{Q})$ is empty, $\mathcal{Q}$ is \emph{greedy-constructable}, so we are done; so suppose that this is not the case. Let $R(\mathcal{Q})$ be the region in $P$ reachable from the root edge $r$ without passing through any tripod from $\mathcal{T}_{good}(\mathcal{Q})$ (see also \cref{fig:greedy-choice-proof}). Note that all tripods in  $\mathcal{T}_{bad}(\mathcal{Q})$ must by definition be in $R(\mathcal{Q})$. We will argue that we can find some other (constructable) partition $\mathcal{Q}'$ such that $R(\mathcal{Q}')$ is smaller. This is enough, since there are only finitely many possible positions for tripods in constructable solutions, so $R(\cdot)$ can only decrease a finite number of times.

\begin{figure}
\centering
\includegraphics[width=.95\textwidth,page=5]{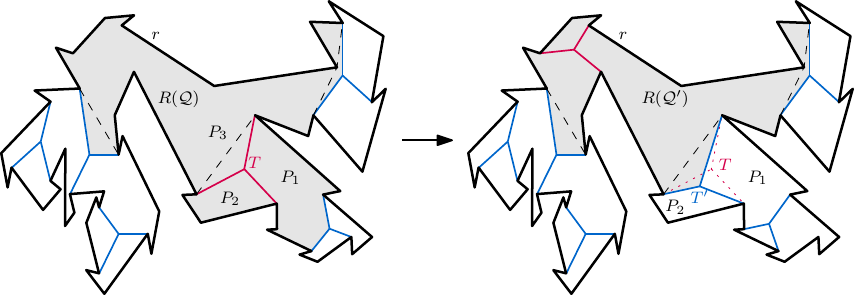}
\caption{The process of replacing a tripod $T$ with the \emph{greedy} one $T'$, as described in the proof of \cref{lem:greedy-constructable}. Bad tripods $\mathcal{T}_{bad}$ are marked in red, good
$\mathcal{T}_{good}$ in blue. The dashed lines indicates ``fences'' where star centers are not allowed to pass through. The gray area $R(\cdot)$ goes down, since the bad tripod $T$ was replaced with a good one $T'$. The partitions (and tripods) inside $P_1, P_2$ and $P'_3=R(\mathcal{Q})\cap P_3$ can change: for example note that the tripods inside $P_1$ changed, and that a new bad tripod inside $R(\mathcal{Q}')$ appeared.}
\label{fig:greedy-choice-proof}
\end{figure}

Consider \cref{fig:greedy-choice-proof}. Let $T\in \mathcal{T}_{bad}(\mathcal{Q})$ be some non-greedy tripod. 
Say it is supported by corners $(D_1, D_2, D_3)$ and splits the polygon into sub-polygons $P_1, P_2, P_3$ where it is oriented towards $P_3$. Without loss of generality, we may assume that, in our partition, the sub-partitions of $P_1$ and $P_2$ are \emph{greedy-constructable}: that is all tripods in $P_1$ and $P_2$ abide the greedy choice (else we can instead choose $T$ to be one of these non-greedy tripods deeper down in the tripod tree). Since $T$ is not greedy, there must exist other sub-partitions of $P_1$ and $P_2$ that are also \emph{greedy-constructable} and give rise to a better (less restrictive) \emph{greedy} tripod $T'$. First note that these new sub-partitions of $P_1$ and $P_2$ must use the same number of pieces as the original partition used for these sub-polygons (otherwise it was not optimal; note that it never makes sense to use an extra piece in a sub-partition to get a less restrictive tripod as we might as well place this extra piece at the tripod point which would make the whole tripod redundant).

Now we want to use \cref{lem:less-restrictive} to replace the tripod $T$ with $T'$. However, we cannot directly apply this lemma as it might destroy some greedy tripods and introduce new non-greedy tripods inside the region $P_3$ in an unpredictable manner. Let us define $P'_3 = P_3 \cap R(\mathcal{Q})$ and $P' = P_1 \cup P_2 \cup P'_3$.
We use \cref{lem:less-restrictive} on the polygon $P'$ instead of $P$, which makes sure that we do not destroy any greedy tripods inside the region $P_3$.
We must be slightly careful also to not move any star centers into the pseudo-triangles of tripods in $P$, whose tripod points might now be corners of $P'$.
For this, we note that it is easy to extend \cref{lem:less-restrictive} to take into account these pseudo-triangles:

In the first part of the proof of \cref{lem:less-restrictive} we argued that there is an optimal (not necessarily constructable) solution without moving any star centers in $P'_3$, and in the second part we used \cref{thm:constructability} (or rather \cref{lem:construct-from-middle}) to argue that then there must also exist a constructable one. When we invoke \cref{lem:construct-from-middle} we can do so on the full polygon $P$, but on our partitions where everything outside $P_3'$ is already constructable (so these star pieces and tripods will not change, and no star center will be moved into the pseudo-triangles of the boundary tripods of $P'_3$). Then \cref{lem:construct-from-middle} would give us an optimal partition where tripods and star centers inside $P'_3$ are also constructable.

To recap, we now have a constructable partition of $P'$ where the greedy tripod $T'$ is used instead of $T$. We extend this to a partition $\mathcal{Q}'$ of $P$ by using all pieces from the original partition $\mathcal{Q}$ which where in $P\setminus P' = P_3\setminus R(\mathcal{Q})$. By design we see that $\mathcal{Q}'$ is \emph{constructable}. Moreover, $R(\mathcal{Q'}) = P'_3 = R(\mathcal{Q})\cap P_3$, since now $T'$ is a greedy tripod (and still all tripods in the subregions $P_1$ and $P_2$ are greedy, although the new partition of these parts can be quite different from the original one). To conclude, $R(\cdot)$ must have gone down, as it is now also restricted by the tripod $T'$ (and the original tripod $T$ was completely contained in $R(\mathcal{Q})$). Hence, by induction, there must exist some greedy-constructable partition.
\end{proof}

\subsubsection{Minimum Star Partition Algorithm}
We are now ready to present the full algorithm, and thus proving \cref{thm:main-theorem}.
To optimally partition a polygon $P$ into a minimum number of star pieces we employ the following strategy. We begin by enumerating all possible positions for tripods (that is triples of concave corners $D_1, D_2, D_3$ of $P$ that might support a tripod). Now, for each of these, we can employ the \emph{greedy choice} (see \cref{def:greedy-choice}) to only have a single tripod-center we need to consider. By \cref{lem:greedy-constructable}, there must exist some optimal solution in which all tripods follow this greedy choice. For now, assume we can actually compute these greedy tripod points (we will get back to this later). That is we have $O(n^3)$ potential tripod-centers in total. By \cref{thm:constructability} we can now construct a set of potential star centers $\Scen$ by considering all intersections of pairs of lines, where each line is either (i) an extension of a diagonal of $P$, or (ii) a line from a (greedy) tripod-center through the corresponding concave corner in this tripod. Note that there are only $O(n^2 + n^3)$ such lines, so we can bound $|\Scen|$ by $O(n^6)$.
Additionally, we note that in (greedy-)constructable partitions, each piece touches the boundary at some corner.
Given $\Scen$, we can hence employ the dynamic programming algorithm (\cref{alg:dp}, \cref{thm:dp}) to find a minumum star partition of the polygon $P$.

\paragraph{Resolving the greedy choice.}
Now, let us return to the issue of actually determining the position of the greedy choice tripod point, of some tripod supported on concave corners $D_1, D_2, D_3$ of $P$. Let $P_1,P_2,P_3$ be the sub-regions (like in \cref{fig:greedy}), such that the tripod is oriented towards $P_3$. Note that $P_1$ and $P_2$ are defined by a diagonal of $P$ and not by the (so far unknown) tripod.
Now, we can find greedy-constructable optimal solutions for $P_1$ and $P_2$ separately, by invoking our algorithm recursively (see also \cref{fig:algorithm-tripod-recurse}): in the subpolygons we again enumerate all potential tripods, solve using the greedy choice, and invoke the dynamic program to obtain an optimal solution. Note that there are only $O(n^2)$ subproblems, since each subproblem is defined by some diagonal between two concave corners of $P$ (here we use the fact that the tripods in a greedy-constructable partition form a \emph{rooted} tree). Moreover, our dynamic program allows us to find all possible positions of the star center $A_1$ in $P_1$ used to define the tripod. Indeed $A_1$ is the star center that sees a prefix of the pseudo-diagonal from $D_1$ to $D_2$, so we can find all possible positions of it by \cref{obs:multiple-solutions}.
Similarly for the star center $A_2$ in $P_2$. Therefore we may simply enumerate over all pairs of possibilities of $A_1$ and $A_2$ and choose the best valid one according to the greedy choice (\cref{def:greedy-choice}, see also \cref{fig:greedy,fig:greedy-choice-angle,fig:algorithm-tripod-recurse}).
By \cref{lem:extending-to-tripod}, any partition we get here can be extended to a partition meeting the tripod legs. 
Conversely, any constructable partition with a tripod supported by $D_1, D_2$ and $D_3$, is also a partition of $P\setminus \Delta$ (as no star center would be inside the pseudo-triangle $\Delta$; so we can apply \cref{lemma:cut} to carve out this part). Hence we do not loose or gain anything by restricting ourselves to finding optimal partitions restricted by the pseudo-diagonal, instead of restricted by the (so far unknown) tripod.

\SetKwFunction{SolveSub}{SolveSubregion}
\SetKwFunction{TripodGreedy}{TripodGreedyChoice}

\begin{figure}[htb!]
\centering
\includegraphics[scale=0.9,page=8]{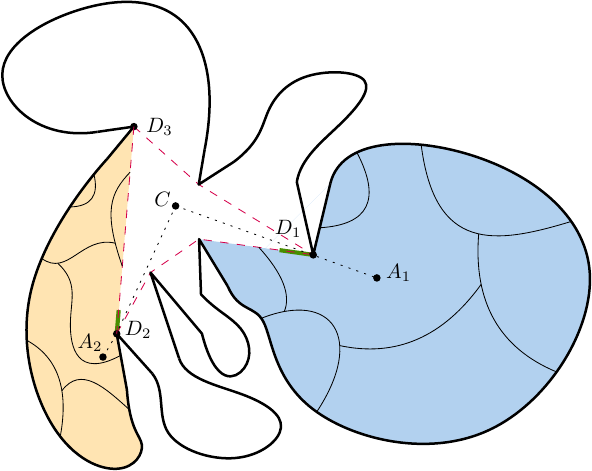}
\caption{An illustration of how \texttt{TripodGreedyChoice}($D_1,D_2,D_3$) works. First the pseudo-triangle $\Delta$ (dashed in red) is computed. Then it calls \texttt{SolveSubregion}() on the yellow and blue subpolygons of $P\setminus \Delta$, to recursively partition these optimally. Moreover, using \cref{obs:multiple-solutions}, we find all potential positions (in greedy-constructable optimal solutions) for star centers $A_1$ and $A_2$, whose piece contains a prefix (marked green in the figure) of the pseudo-diagonals adjacent to $D_1$ and $D_2$ (respectively). The pair which makes for the \emph{least restrictive} (\cref{def:greedy-choice}) tripod point is chosen, and this point $C$ is returned. 
Curved parts indicated that details have been omitted.}
\label{fig:algorithm-tripod-recurse}
\end{figure}

\paragraph{Full algorithm.}
The pseudo-code can be found in \cref{alg:full}, consisting of two mutually recursive functions \TripodGreedy{} and \SolveSub{}. The function \TripodGreedy{} will find the tripod point of the greedy tripod, and \SolveSub{$D_1,D_2$} will optimally solve the sub-polygon enclosed by the diagonal $D_1$-$D_2$. To obtain the optimal partition for the full polygon, we can just call \SolveSub{} on some edge of $P$. The correctness follows from the above discussion.

\setcounter{AlgoLine}{0}
\begin{algorithm}[!htb]
\caption{Minimum Star Partition Algorithm.}
\label{alg:full}
\small
\DontPrintSemicolon
\SetKwProg{Fn}{function}{:}{}
\Fn{\TripodGreedy{$D_1,D_2,D_3$}}{
\tcp{Returns the greedy choice tripod point of the tripod supported by corners $D_1,D_2,D_3$ of $P$. See \cref{fig:algorithm-tripod-recurse}.}
\tcp{Suppose, w.l.o.g.\ (other cases are similar), that we are like in \cref{fig:greedy}: that is the tripod should be oriented towards $D_3$ and the root is in the face fenced of by the $D_2$-$D_3$ pseudo-diagonal.}
Construct the pseudo-diagonal $D_1$-$D_2$, say it goes through points $D_1 = X_1, X_2, ..., X_k = D_2$\;
Calculate the optimal number of star pieces to cover the sub-polygon defined by separator $X_1$-$X_2$, by calling \SolveSub($X_1,X_2$)\;
Additionally, this finds all possible positions, in greedy constructable optimal solutions, of star centers that see a small part of the segment $X_1$-$X_2$ next to $D_1$\;
Do the same for diagonal $D_2$-$D_3$\;
Look at all combinations of star centers and return only the single tripod-center that makes for the greedy choice (if any).\;
}

\BlankLine
\Fn{\SolveSub{$D_1$, $D_2$}}{
\tcp{Requires that $D_1,D_2$ is a diagonal of $P$. Call the sub-polygon on the right side (when looking from $D_1$ to $D_2$) of the diagonal $P'$.}
\tcp{\SolveSub($D_1,D_2$) will optimally partition $P'$ into a minimum number of star-shaped pieces. Moreover, it will consider all possible positions, in greedy constructable solutions, for the star center that sees a small part of the $D_1D_2$ segment next to $D_1$.}
Enumerate all valid positions of tripods (i.e. 3-tuples of concave corners) inside $P'$, and call \TripodGreedy{} on these.\;
Let $L$ be the set of lines that are either (i) the line through two corners of $P$, or (ii) the line through some tripod-center and the corresponding concave corner of $P$.\;
Let $\Scen$ be the set of intersection points of pairs of lines from $L$.\;
Call the dynamic programming algorithm on $P'$ and $\Scen$ to find an optimal solution.\;
By \cref{obs:multiple-solutions} we can additionally find all possible positions of the star center seeing a prefix of the $D_1$-$D_2$ pseudo-diagonal from $D_1$ to $D_2$.\;
}
\end{algorithm}

\paragraph{Running time.} We now analyze the running time of \cref{alg:full}, and again we apply memoization to not recompute the same subpolygons multiple times. We will see that the total running time is $O(\poly(n))$---or, in fact, it takes $O(n^{105})$ arithmetic operations or $O(n^{107})$ time. According to \cref{cor:bit-complexity} and \cref{lem:sightlines-steinerpoints}, every star center or steiner point can be encoded by $O(n)$ corners of $P$, therefore each arithmetic operation can be done in $O(n^2)$ time. 

\begin{itemize}
\item \TripodGreedy{} will be called at most $O(n^3)$ times, since there are only $O(n^3)$ choices for 3-tuples of corners $D_1, D_2, D_3$ of $P$. Constructing the pseudo-diagonals can be done in $O(n)$ time. There are at most $O(n^6)$ possible positions for star centers, so only $O(n^{12})$ possible combinations for the greedy choice. For each of these combinations, we might need to go through the full polygon to see that the legs of the tripod does not intersect the polygon. Hence the computation inside \TripodGreedy{} need in total $O(n^3 \cdot n^{12} \cdot n) = O(n^{16})$ time over the run of the full algorithm.
\item \SolveSub{} will be called at most $O(n^2)$ times, since there are only $O(n^2)$ choices for pairs of corners $D_1, D_2$ of $P$. Enumerating valid tripod-positions can be done in $O(n^4)$ time ($O(n^3)$ many possible 3-tuples of corners, and each can be checked in $O(n)$ time  by going through the polygon and constructing the pseudo-triangle to see if a valid tripod can be formed there). We then construct the set $\Scen$ of size at most $O(n^6)$. Calling the dynamic programming algorithm on this set takes $O(n^{103})$ arithmetic operations and $O(n^{105})$ time (see \cref{thm:dp,rem:dp}). Followed from \cref{obs:multiple-solutions}, in the same time, we can find all possible positions of star centers covering the start of the $D_1$-$D_2$ segment: indeed, we call the dynamic programming algorithm $\SolveSep(D_1,A,X)$ (where $X \neq D_1$ is the closest point to $D_1$ in $\Sbor$ on the $D_1$-$D_2$ segment) for all possible star centers $A\in \Scen$, to see which ones of these give a partition of minimum size. Note that between these calls to $\SolveSep$, we do not need to reset the dp-cache, so in total this takes only $O(n^{103})$ arithmetic operations and $O(n^{105})$ time.
Hence, in total, over the full run of the algorithm, we will spend $O(n^2 \cdot n^{103}) = O(n^{105})$ arithmetic operations and $O(n^2 \cdot n^{105}) = O(n^{107})$ time in \SolveSub.
\end{itemize}
The above discussion concludes the proof of \cref{thm:main-theorem}.

\newpage
\printbibliography

\appendix

\newpage
\section{Existence of Coordinate and Area Maximum Partitions}
\label{apx:existence}

\begin{proof}[Proof of Lemma~\ref{lem:coordmax}]
Recall that the \emph{Hausdorff distance} between two compact sets $A,B \subset \RR^2$ is defined as $d_H(A,B) \coloneqq \inf \{ r \geq 0 \mid A \subset B \oplus D(r) \text{ and } B \subset A \oplus D(r) \}$,
where $\oplus$ is the Minkowski sum and $D(r)$ is the disk of radius $r$ centered at the origin.
We will use the fact that $(M(\RR^2), d_H)$ is a compact metric space, where $M(\RR^2)$ denotes the set of all non-empty closed subsets of $\RR^2$; see for instance \cite[Theorem~4.5]{evans_probability_2008}.

Let $c^*=\sup c(\mathcal Q)$, where the supremum is taken over all optimal star partitions $\mathcal Q$ of $P$.
We claim there exists a star partition which realizes the coordinate vector $c^*$.
Consider a sequence of optimal star partitions $(\mathcal Q_i)_{i\in\NN}$ so that $c(\mathcal Q_i)$ converges to $c^*$ as $i\longrightarrow\infty$.
Let $c(\mathcal Q_i)=\langle A_{1i},\ldots,A_{ki}\rangle$ and let the pieces of $\mathcal Q_i$ be $Q_{1i},\ldots,Q_{ki}$ so that the maximum star center of $Q_{ji}$ is $A_{ji}$.
Be passing to a subsequence, we can assume that the sequence $(Q_{1i})_{i\in\NN}$ converges to a compact set $Q_1^*$ with respect to Hausdorff distance.
Similarly, we can assume that $(Q_{ji})_{i\in\NN}$ converges to $Q_j^*$ for each $j\in\{1,\ldots,k\}$.
Let $\mathcal Q^*=\{Q_1^*,\ldots,Q_k^*\}$ and $c^*=\langle A_1^*,\ldots,A_k^*\rangle$.
We claim that $\mathcal Q^*$ is a star partition of $P$ where $A_j^*$ is the star center of $Q_j^*$.

To see that $Q_j^*$ is star-shaped with star center $A_j^*$, we first observe that $A_j^*\in Q_j^*$.
Otherwise, the star center $A_{ji}$ would not be in $Q_{ji}$ for sufficiently large values of $i$.
Suppose now that $A_j^*$ is not a star center of $Q_j^*$.
This means that there is a point $B\in Q_j^*$ so that the line segment $A_j^*B$ is not in $Q_j^*$.
For each $i\in\NN$, we can choose $B_i\in Q_{ji}$ so that $B_i\longrightarrow B$ for $i\longrightarrow\infty$.
Since $A_{ji}\longrightarrow A_j^*$, $B_i\longrightarrow B$ and $A_j^*B$ is not in $Q_j^*$, it follows that $A_{ji}B_i$ is not in $Q_{ji}$ for sufficiently high values of $i$, which contradicts that $A_{ji}$ is a star center of $Q_{ji}$.
In a similar way, we can argue that the sets in $\mathcal Q^*$ are interior-disjoint; otherwise, the sets in $\mathcal Q_i$ could not be interior-disjoint for sufficiently high values of $i$.
Likewise, we get that $\bigcup \mathcal Q^*=P$, and we conclude that $\mathcal Q^*$ is a coordinate maximum optimal star partition of $P$.
\end{proof}

\begin{proof}[Proof sketch of Lemma~\ref{lem:coordmaxsubset}]
The proof is analogous to that of~\Cref{lem:coordmax}.
Without loss of generality, we can assume that $d=(1,0)$, so that for any set of points $A'_i,\ldots,A'_k$, we have
\[\langle A'_i\cdot d, A'_i\cdot d^\perp,A'_{i+1}\cdot d, A'_{i+1}\cdot d^\perp,\ldots,A'_k\cdot d,A'_k\cdot d^\perp\rangle=\langle A'_i,A'_{i+1},\ldots,A'_k\rangle.\]
Define the supremum $\langle A^*_i,A^*_{i+1},\ldots,A^*_k\rangle=\sup \langle A'_i,A'_{i+1},\ldots,A'_k\rangle$ over all partitions with star centers $A'_i,A'_{i+1},\ldots,A'_k$ in the region $F$ and the rest fixed at the points $A_1,\ldots,A_{i-1}$.
We then consider a convergent sequence of star partitions with fixed star centers $A_1,\ldots,A_{i-1}$ and the rest converging to $A^*_i,\ldots,A^*_k$, and it follows that the limit of the pieces constitute a partition realizing the supremum.
\end{proof}

\begin{proof}[Proof of Lemma~\ref{lem:areamax}]
The proof is similar to that of \Cref{lem:coordmax}.
Namely, let $a^*=\sup a(\mathcal Q)$, where the supremum is taken over all star partitions $\mathcal Q$ of $P$ with star centers $\mathcal A$.
We then consider a sequence of partitions $(\mathcal Q_i)_{i\in\NN}$ with star centers $\mathcal A$ so that $a(\mathcal Q_i)\longrightarrow a^*$ as $i\longrightarrow\infty$.
Let the pieces of $\mathcal Q_i$ be $Q_{1i},\ldots,Q_{ki}$ so that $A_j$ is a star center of $Q_{ji}$.
By passing to a subsequence, we can assume that $(Q_{ji})_{i\in\NN}$ converges to a compact set $Q_j^*$ for each $j\in\{1,\ldots,k\}$.

As in the proof of \Cref{lem:coordmax}, we can conclude that $\mathcal Q^*=\{Q_1^*,\ldots,Q_k^*\}$ is a star partition of $P$ with star centers $\mathcal A$ and that $a(\mathcal Q^*)=a^*$, so $\mathcal Q^*$ is area maximum.
\end{proof}

\section{Structural Theorem}
\label{app:struct}

In this section, we give an elementary and independent proof of a structural theorem about optimal star partitions, which is not used in the rest of the paper.
A point on the boundary of the polygon $P$ is \emph{canonical} if it is a corner of $P$ or the endpoint of an extension of an edge of $P$; see \cref{fig:extensions} (left).

\begin{theorem}\label{thm:2}
Let $k$ be minimum such that there exists a star partition of $P$ consisting of $k$ polygons and assume $k\geq 2$.
There exists a star partition $Q_1,\ldots, Q_k$ of $P$, where each piece $Q_i$ has the following properties:
\begin{enumerate}
\item $\partial Q_i$ contains a concave corner of $P$, and \label{prop:boundary}
\item for each connected component of the shared boundary $\partial P\cap \partial Q_i$, both endpoints are canonical. \label{prop:endpoint}
\end{enumerate}
\end{theorem}

\begin{figure}[htb!]
\centering
\includegraphics[page=9]{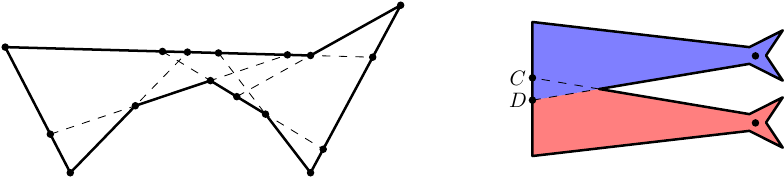}
\caption{Left: A polygon with the extensions of the edges and the canonical points shown.
Right: Using only canonical points as corners of pieces on the boundary of $P$, one of the points $C$ and $D$ must be used.}
\label{fig:extensions}
\end{figure}

Before going into the proof, let us first make a few remarks.
Note that there are at most $3n$ canonical points, since each edge of $P$ creates at most $2$ canonical points that are not corners of $P$.
The algorithm described in this paper considers $O(n^{13})$ Steiner points on the boundary of $P$, so it might be possible to use property~\ref{prop:endpoint} to obtain a faster algorithm.
However, in the proof of the theorem, we change the partition in order to only use canonical Steiner points on the boundary.
It does therefore not follow immediately that our algorithm can also be modified to find the resulting partition, but we are confident that the result can be used to create a faster algorithm.

We find it somewhat surprising that canonical Steiner points suffice on the boundary of $P$, since, as shown in \cref{fig:art}, it is sometimes necessary to use Steiner points in the interior of $P$ of degree $\Omega(n)$, whereas the canonical points have degree at most $1$ (as defined in \cref{sec:intro}).
\Cref{fig:extensions} (right) proves that it is necessary to have least some Steiner points on the boundary of $P$:
In any partition into two star-shaped pieces, we must have a Steiner point on the segment $CD$.

\begin{proof}
Let $Q_1,\ldots,Q_k$ be a minimum star partition of $P$.
Let us fix a star center $A_i$ in each piece $Q_i$.
If $\partial Q_i\cap \partial Q_j\neq\emptyset$ and $i\neq j$, we say that $Q_i$ and $Q_j$ are \emph{neighbours}.
The shared boundary $\partial Q_i\cap \partial Q_j$ of two neighbours $Q_i$ and $Q_j$ is a collection of open polygonal curves.
An \emph{interior point} of an open curve is a point on the curve which is not an endpoint.

\begin{claim}\label{claim1}
Let $\gamma$ be an open curve in the shared boundary $\partial Q_i\cap \partial Q_j$ of two neighbours $Q_i$ and $Q_j$.
We can assume that $\gamma$ is either a line segment or two line segments that have one of the star centers $A_i$ and $A_j$ as a common endpoint.
If a star center of one piece is a corner of $\gamma$, the corner is convex with respect to that piece and concave with respect to the other piece.
\end{claim}

\begin{subproof}

\begin{figure}
\centering
\includegraphics[page=10]{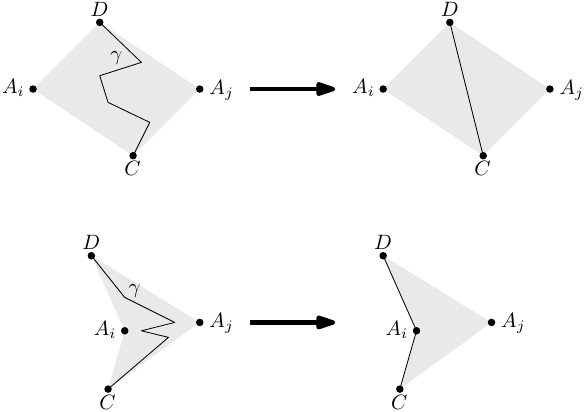}
\caption{The process simplifying the shared boundary between $Q_i$ and $Q_j$ in order to satisfy \cref{claim1}.
The gray region is the quadrilateral $F$.}
\label{fig:moveDnew}
\end{figure}

Let the endpoints of $\gamma$ be $C$ and $D$; see \cref{fig:moveDnew}.
Since the segments $A_i C$ and $A_i D$ are in $Q_i$, and $A_j C$ and $A_j D$ are in $Q_j$, we have a well-defined quadrilateral $F=A_i C A_j D$.
If $A_i$ and $A_j$ are both convex corners of $F$, we replace $\gamma$ by the segment $CD$, which must be a diagonal of $F$.
Otherwise, consider without loss of generality the case that $A_i$ is a concave corner of $F$.
We then replace $\gamma$ by $CA_i \cup A_i D$.
In either case, the modification clearly leaves $Q_i$ and $Q_j$ star-shaped, and we are left with a shared boundary of the claimed type.
\end{subproof}

\noindent
\textbf{Property~\ref{prop:boundary}.}
In order to prove that there is a partition satisfying Property~\ref{prop:boundary}, suppose that $\partial Q_i$ does not contain a concave corner of $P$.
We show how to modify the partition so that the property is eventually satisfied.
In essence, we expand the piece $Q_i$ until Property~\ref{prop:boundary} is eventually satisfied.

\begin{claim}\label{claim2}
We can assume that each concave corner $D$ of $Q_i$ is a concave corner of $P$ or a star center $A_j$ of a neighbour $Q_j$ of $Q_i$.
\end{claim}

\begin{subproof}
We describe a way to expand $Q_i$ so that we eliminate concave corners of $Q_i$ which are neither corners of $P$ nor star centers of neighbours (note that if a concave corner of $Q_i$ touches $P$, then it touches $P$ at a concave corner).
Let $CD$ and $DE$ be maximal segments on the boundary of $Q_i$ such that $D$ is a concave corner of $Q_i$ and no corner of $P$ and no star center of a neighbour is an interior point of $CD\cup DE$.
It follows by \cref{claim1} that this point $D$ can be assumed to lie on the boundary of two polygons $Q_j$ and $Q_l$ whose boundaries share a line segment $DF$.
Informally, we now move $D$ towards $F$.
This will expand $Q_i$ and shrink $Q_j$ and $Q_l$ and possibly also other neighbours of $Q_i$ whose boundaries contain segments or points on $CD\cup DE$.
To be precise, define $D'$ to be the first point on $DF$ from $D$ such that one of the following cases holds:
(i) one of the segments $CD'$ or $D'E$ contains a corner of $P$,
(ii) $D'$ is the intersection of $DF$ and $CE$ (if it exists),
(iii) one of the segments $CD'$ or $D'E$ contains a star center $A_m$, $m\neq i$,
(iv) $D'=F$.
The cases are shown in \cref{fig:elimConcave}.
We then assign the quadrilateral $CD'ED$ to $Q_i$, which will increase $Q_i$, decrease the pieces intersected by $CD'\cup D'E$ and, by \cref{lemma:cut}, all the involved pieces remain star-shaped.
We repeat this operation, and it remains to argue that the process eventually terminates.

\begin{figure}[h]
\centering
\includegraphics[page=2]{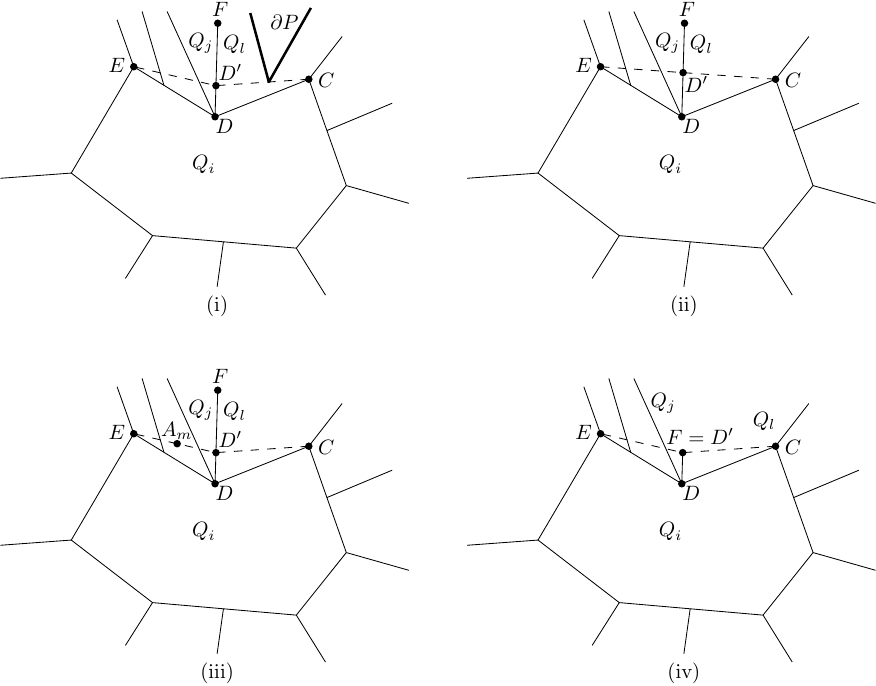}
\caption{The process of eliminating concave corners of $Q_i$ that are not star centers of neighbouring pieces.}
\label{fig:elimConcave}
\end{figure}

In case (i), $Q_i$ now touches $P$ at a corner which clearly must be concave.
In case (ii), we have eliminated a concave corner of $Q_i$.
In case (iii), we have increased the number of star centers on the boundary of $Q_i$, which can happen at most $k-1$ times.
In case (iv), we eliminated a segment $DF$ of $Q_j$ and $Q_l$, which decreases the total number of segments of the pieces.
We conclude that the operation can be repeated at most a finite number of times and the process therefore eventually terminates.
\end{subproof}

A \emph{star neighbour} of $Q_i$ is a neighbour $Q_j$ whose star center $A_j$ is on $\partial Q_i$.
Recall that we assume $\partial Q_i$ does not contain a concave corner of $P$.
Now, additionally assume that $Q_i$ is non-convex, i.e., $Q_i$ has a concave corner.
As $\partial Q_i$ does not contain a concave corner of $P$, we can assume that this concave corner is the star center of a star neighbour by \cref{claim2}.
To obtain a contradiction with the minimality of the start partition, we show that $Q_i$ can be subsumed by the star neighbours.
To this end, we consider a triangulation of $Q_i$.
The diagonals of the triangulation that have an endpoint at a concave corner of $Q_i$ partition $Q_i$ into convex polygons $R_1,\ldots,R_m$, as illustrated in \cref{fig:eat} (left), where solid diagonals have an endpoint at a concave corner.
\begin{figure}[h]
\centering
\includegraphics[page=3]{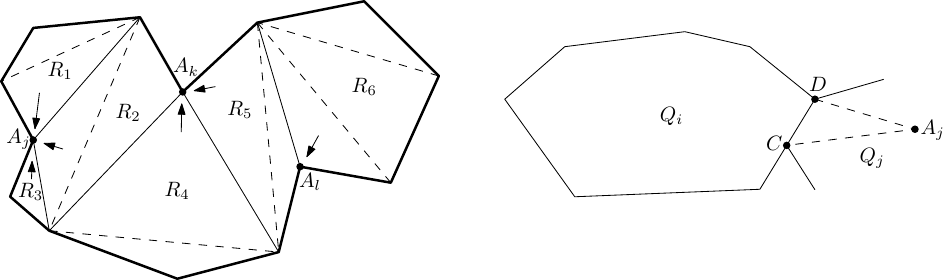}
\caption{Left: We reassign the piece $Q_i$ to the pieces of the concave star neighbours.
Right: We expand the convex piece $Q_i$ by adding the triangle $A_jCD$.}
\label{fig:eat}
\end{figure}
We assign a polygon $R_p$ to $Q_j$ if $A_j$ is on the boundary of $R_p$ and $R_p$ has not already been assigned to another star neighbour, as shown by the arrows in the figure.
Thus, the considered star partition did not consist of a minimum number of polygons, which is a contradiction. Consequently, $Q_i$ must be convex.

As we assume $k \geq 2$, there has to be a piece $Q_j$ that is a neighbour of $Q_i$ so that the common boundary $\partial Q_i\cap\partial Q_j$ contains a segment $CD$.
We expand $Q_i$ and shrink $Q_j$ by adding the triangle $A_jCD$ to $Q_i$, as shown in \cref{fig:eat} (right), which keeps $Q_i$ as well as $Q_j$ star-shaped.
This can introduce concave corners $C$ and $D$ on $Q_i$, but we can proceed as in the proof of \cref{claim2} and expand $Q_i$ until we hit a concave corner of $P$ or obtain that all concave corners are star centers.
If we do not hit a concave corner of $P$, we can again argue that $Q_i$ can be subsumed by $Q_j$ and potentially other star neighbours, contradicting the minimality of the star partition.

Note that if a concave corner of $P$ appears on the boundary of a piece before the changes described above (including in the proof of \Cref{claim2}), then the corner also appears on the piece afterwards.
We conclude that we can expand each polygon $Q_i$ that does not contain a concave corner of $P$ until it eventually does.
In the end, we obtain a minimum star partition with Property~\ref{prop:boundary}.
\medskip

\noindent
\textbf{Property~\ref{prop:endpoint}.}
We now prove that we can obtain Property~\ref{prop:endpoint}.
Consider an edge $CD$ of $P$.
For each piece $Q_i$, it holds by the optimality of the partition that the intersection $CD\cap\partial Q_i$ is either empty or a single segment $EF$ (which may be a single point); see \cref{fig:canonical}.
We call such a segment $EF$ a \emph{shared} segment, and we say that a shared segment is \emph{canonical} if both endpoints are canonical.
We define the \emph{canonical prefix} of $CD$ to be all the shared segments from $C$ to (and excluding) the first non-canonical shared segment.
We show that we can modify the partition such that the number of segments in the canonical prefix increases.
It therefore holds that the process must stop, so that all shared segments on $CD$ have canonical endpoints in the end.
The process does not change whether shared segments on other edges of $P$ are canonical, so repeating the process for all edges yields a star partition with Property~\ref{prop:endpoint}.

\begin{figure}[h]
\centering
\includegraphics[page=4]{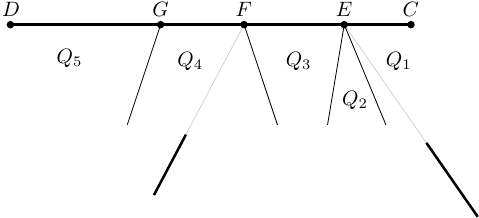}
\caption{The fat segments are edges of $P$ and the thin black segments indicate the boundaries of pieces in the interior of $P$.
The corners $C$ and $D$ as well as the interior points $E$ and $F$ are canonical, but $G$ is not.
There are three segments in the canonical prefix of $CD$, namely the intersections of $CD$ with each of $\partial Q_1,\partial Q_2,\partial Q_3$.
}
\label{fig:canonical}
\end{figure}

Consider the first non-canonical segment $EF$, where $E$ is canonical but $F$ is not.
Let $F'$ be the first point on $FD$ from $F$ such that either
(i) $F'$ is canonical,
(ii) $A_iF'$ contains another star center $A_j$, or
(iii) $A_iF'$ contains a corner $G$ of $P$.
Since the endpoint $D$ is canonical, the point $F'$ is well-defined.
The cases are illustrated in \cref{fig:canonicalcasei,fig:canonicalcaseii,fig:canonicalcaseiii}.

In case (i), we assign the triangle $A_iFF'$ to $Q_i$, which according to Lemma~\ref{lemma:cut} keeps all pieces star-shaped as the triangle does not contain any star centers of other pieces than $Q_i$.
We have then increased the number of segments in the canonical prefix.

In case (ii), we assign the triangle $A_iEF'$ to the piece $Q_j$.
The shared segment $s=CD\cap \partial Q_j$ of the piece $Q_j$ now starts at the point $E$.
The other endpoint of $s$ is either $F'$ (as in \cref{fig:canonicalcasei}) or a later point.
In the latter case, it is possible that $s$ is canonical, and we have increased the number of segments in the canonical prefix.
If $s$ is not canonical, we repeat the process of repairing the non-canonical endpoint of $s$.
Since the star center $A_j$ must be closer to the edge $CD$ than $A_i$, we encounter case (ii) less than $k$ times before we end in case (i) or (iii).

\begin{figure}[h]
\centering
\includegraphics[page=8]{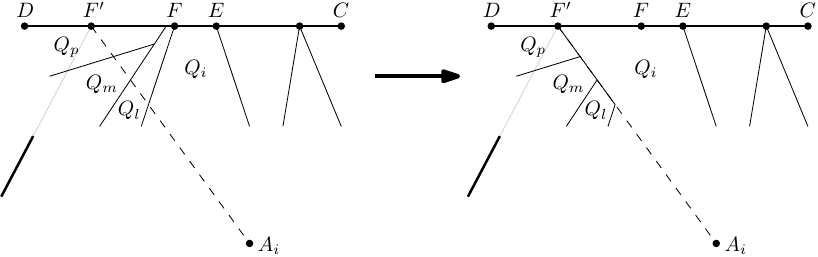}
\caption{Case (i).
We expand $Q_i$ with the triangle $A_i FF'$.
}
\label{fig:canonicalcasei}
\end{figure}

\begin{cfigure}[h]
\centering
\includegraphics[page=6]{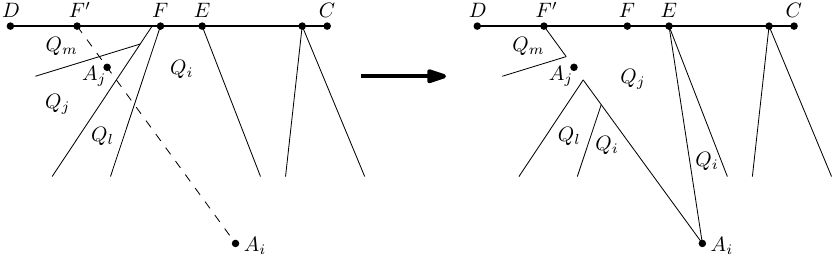}
\caption{Case (ii).
We expand $Q_j$ with the triangle $A_i FF'$.
}
\label{fig:canonicalcaseii}
\end{cfigure}

It remains to consider case (iii), namely that $A_iF'$ contains a corner $G$ of $P$.
We know that $G\neq D$, since otherwise, we would have been in case (i).
Let $F''$ be the first canonical point on the segment $F'E$ from $F'$, which is well-defined as $E$ is canonical.
Furthermore, we know that $F''$ is on the segment $EF$, since otherwise, we should have been in case (i).
Let $\mathcal Q'=\{Q'_1,\ldots,Q'_{f'}\} \setminus \{Q_i\}$ be the set of pieces intersected by the interior of $GF'$ in order from $G$ and excluding $Q_i$.

In a first step, we assign the triangle $A_iF''F'$ to the piece $Q_i$.
This can reduce pieces intersecting $A_i F'$, but \Cref{lemma:cut} ensures that they remain star-shaped.
However, the point $F'$ need not be canonical.
In a second step, we therefore remove the triangle $\Delta=GF''F'$ from $A_i$ and instead distribute $\Delta$ among the pieces $\mathcal Q'$, as follows.
For each $j=1,\ldots,f'-1$, we consider the segment $s$ on the shared boundary of $Q'_j$ and $Q'_{j+1}$ with an endpoint on $GF'$.
We then extend $s$ into $\Delta$ until we reach one of the other segments $GF''$ or $F''F'$ bounding $\Delta$, or we meet an extension that was already added for a smaller value of $j$.
Since $F''$ was chosen as the first canonical point on $F'E$, this results in a star partition of $P$. 
Furthermore, the shared segment of $Q_i$ on $CD$ has now become the segment $EF''$ (which might just be a single point), which is canonical, so we have increased the number of segments in the canonical prefix. Recall that we consider the number of segments in the prefix here and not the geometric length of the prefix; the geometric length of the prefix indeed can decrease.
  
\begin{figure}[h]
\centering
\includegraphics[page=7]{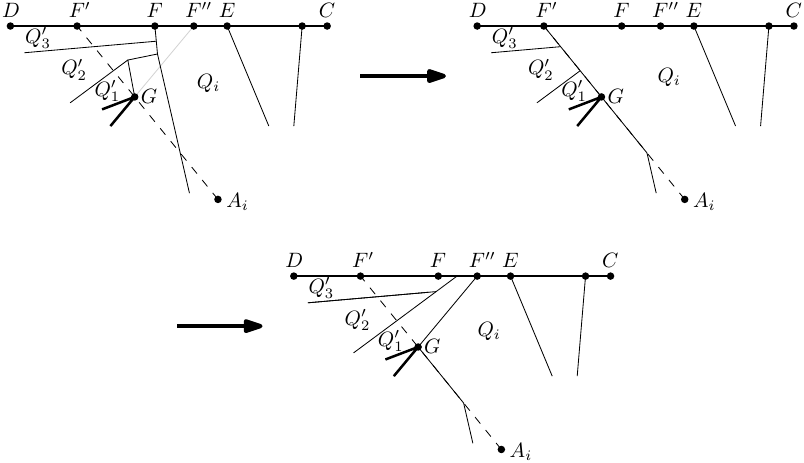}
\caption{Case (iii).
In two steps, we distribute the triangle $A_iF''F'$ among $Q_i$ and the pieces $Q''_1,Q''_2,Q''_3$ intersected by $G F'$.
}
\label{fig:canonicalcaseiii}
\end{figure}

Finally, note that none of the above modifications of the star partition cause a piece $Q$ to become non-adjacent to a concave corner $H$ of $P$ if $Q$ was adjacent to $H$ before.
More precisely, only the modifications in case (iii) changes the neighbourhood of a corner $G$ of $P$.
However, as there are no star centers in the triangle $A_i EF'$, the modifications will not remove $G$ from the boundary of any piece.
\end{proof}

\end{document}